\documentclass[twoside,11pt]{article}

\usepackage[latin1]{inputenc}
\usepackage[T1]{fontenc}
\usepackage[english]{babel}
\usepackage{color}
\usepackage{tocloft}
\usepackage{verbatim}
\usepackage{tikz}

\usepackage[
  hmarginratio={1:1},     % equal left and right margins
  vmarginratio={1:1},     % equal top and bottom margins
  textwidth=17cm,        % new text width
  textheight=24cm,
  heightrounded          % always useful
]{geometry}

\usepackage{amsmath,amsthm,amssymb}
\usepackage{mathrsfs, mathtools}
\usepackage{dsfont}

\allowdisplaybreaks

\usepackage[all]{xy} 
\SelectTips{cm}{11} % Makes all xy arrows match the usuadifferl LaTeX arrows 
\usepackage[hidelinks]{hyperref}

\usepackage[numbers,sort&compress]{natbib}
\makeatletter
\renewcommand{\@biblabel}[1]{[#1]\hfill}
\makeatother

\usepackage[mathscr]{euscript}
\usepackage{calrsfs}
\DeclareMathAlphabet{\pazocal}{OMS}{zplm}{m}{n}

\usepackage[shortlabels]{enumitem}

\theoremstyle{plain}
\newtheorem{theorem}[equation]{Theorem}
\newtheorem{maintheorem}{Theorem}

\newtheorem{corollary}[equation]{Corollary}

\newtheorem{lemma}[equation]{Lemma}
\newtheorem{proposition}[equation]{Proposition}

\theoremstyle{definition}
\newtheorem{definition}[equation]{Definition}
\newtheorem{example}[equation]{Example}
\newtheorem{remark}[equation]{Remark}

\numberwithin{equation}{subsection}
\makeatletter
\binoppenalty=\maxdimen
\relpenalty=\maxdimen
\newlength{\@listleftmargin}
\setlist{leftmargin=\@listleftmargin,itemsep=0pt,topsep=0pt,partopsep=0pt,parsep=\parskip}
\setlist[enumerate]{align=left,labelsep=*,leftmargin=\@listleftmargin,itemsep=0pt,topsep=0pt,partopsep=0pt,parsep=\parskip}

\makeatother

\newcommand{\pr}{\mathrm{pr}}
\newcommand{\pt}{*}
\newcommand{\rmpar}{\mathrm{par}}

\newcommand{\rmH}{\mathrm{H}}
\renewcommand{\sf}{{s\hspace{-0.04cm}f}}

\newcommand{\ori}{{\mathsf{or}}}

\newcommand{\rmh}{\mathrm{h}}
\newcommand{\rmD}{\mathrm{D}}

\newcommand{\uni}{\mathrm{uni}}

\newcommand{\ev}{\mathrm{ev}}
\newcommand{\coev}{\mathrm{coev}}

\newcommand{\dR}{\mathrm{dR}}

\newcommand{\dd}{\mathrm{d}}
\newcommand{\End}{\mathrm{End}}

\newcommand{\Hom}{\mathrm{Hom}}
\newcommand{\Diff}{\mathrm{Diff}}

\newcommand{\opp}{\mathrm{op}}
\newcommand{\rank}{\mathrm{rk}}
\newcommand{\RP}{\mathrm{RP}}

\newcommand{\ind}{\mathrm{ind}}
\newcommand{\curv}{\mathrm{curv}}

\newcommand{\scM}{{\mathscr{M}}}
\newcommand{\scL}{\mathscr{L}}

\newcommand{\scC}{\mathscr{C}}
\newcommand{\scD}{\mathscr{D}}

\newcommand{\scQ}{Q}
\newcommand{\scZ}{\mathscr{Z}}
\newcommand{\scF}{\mathscr{F}}

\newcommand{\scR}{\mathscr{R}}

\newcommand{\CG}{\pazocal{G}}
\newcommand{\CI}{\pazocal{I}}

\newcommand{\CT}{\pazocal{T}}
\newcommand{\CV}{{\pazocal{V}}}

\newcommand{\CF}{\pazocal{F}}

\newcommand{\CE}{\pazocal{E}}

\newcommand{\CS}{\pazocal{S}}
\newcommand{\CA}{\pazocal{A}}

\newcommand{\sfR}{\mathsf{R}}
\newcommand{\sfn}{\mathsf{n}}
\newcommand{\rev}{\mathsf{rev}}
\newcommand{\sfU}{\mathsf{U}}

\newcommand{\sfr}{\mathsf{r}}

\newcommand{\hol}{\mathsf{hol}}

\newcommand{\sfSigma}{\mathsf{\Sigma}}

\newcommand{\FC}{\mathbb{C}}
\newcommand{\bbD}{{\mathbb{D}}}
\newcommand{\NN}{\mathbb{N}}
\newcommand{\FR}{\mathbb{R}}
\newcommand{\RZ}{\mathbb{Z}}

\newcommand{\One}{\mathds{1}}
\newcommand{\bbS}{{\mathbb{S}}}

\newcommand{\HLBdl}{{\pazocal{HL}\hspace{-0.02cm}\pazocal{B}un}}
\newcommand{\Grb}{{\pazocal{G}rb}}
\newcommand{\TrivGrb}{{\pazocal{T}\hspace{-0.1cm}riv\pazocal{G}rb}}
\newcommand{\Vect}{{\pazocal{V}\hspace{-0.025cm}ect}}

\newcommand{\HVBdl}{{\pazocal{H}\hspace{-0.025cm}\pazocal{V}\hspace{-0.025cm}\pazocal{B}un}}
\newcommand{\LBdl}{{\pazocal{LB}un}}
\newcommand{\VBdl}{{\pazocal{V}\hspace{-0.025cm}\pazocal{B}un}}

\newcommand{\Mfd}{{\mathscr{M}\mathrm{fd}}}
\newcommand{\Dfg}{{\mathscr{D}\mathrm{fg}}}

\newcommand{\Cart}{{\mathscr{C}\mathrm{art}}}

\newcommand{\Desc}{{\mathscr{D}\mathrm{esc}}}
\newcommand{\OCBord}{{\mathscr{OCB}\mathrm{ord}}}
\newcommand{\CBord}{\mathscr{CB}\mathrm{ord}}
\newcommand{\Bord}{\mathscr{B}\mathrm{ord}}
\newcommand{\OCFFT}{{\mathrm{OCFFT}}}
\newcommand{\FFT}{{\mathrm{FFT}}}
\newcommand{\rmRP}{\mathrm{RP}}
\newcommand{\OCTQFT}{{\mathrm{OCTQFT}}}
\newcommand{\KFrob}{\mathrm{KFrob}}
\newcommand{\RPKFrob}{\mathrm{RP}\text{-}\KFrob}
\newcommand{\wtfrob}{\widetilde{\mathrm{frob}}}
\newcommand{\TBG}{\mathrm{TBG}}

\newcommand{\CHom}{{\pazocal{H}om}}

\newcommand{\ul}[1]{\underline{#1}}

\newcommand{\dslash}{{/\hspace{-0.1cm}/}}

\newcommand{\arisom}{\overset{\cong}{\longrightarrow}}

\newcommand{\iu}{\mathrm{i}}
\newcommand{\tr}{\mathrm{tr}}

\newcommand{\<}{\langle}
\renewcommand{\>}{\rangle}

\newcommand{\qen}{\hfill$\triangleleft$}

\newcommand{\qandq}{\quad \text{and} \quad }

\def\quot#1{``#1''}

\mathtoolsset{showonlyrefs,showmanualtags}

\begin{document}

\begin{flushright}
\small
\textsf{Hamburger Beiträge zur Mathematik Nr.\,765}\\
\textsf{ZMP--HH/19-21}
\end{flushright}

%\vspace{0.25cm}

\begin{center}
\LARGE{\textbf{Smooth Functorial Field Theories\\from B-Fields and D-Branes}}
\end{center}
\begin{center}
\large Severin Bunk and Konrad Waldorf
\end{center}

\begin{abstract}
\noindent
In the Lagrangian approach to 2-dimensional sigma models, B-fields and D-branes contribute topological terms to the  action of worldsheets of both open and closed strings.
We show that these terms naturally fit into a 2-dimensional, smooth open-closed functorial field theory (FFT) in the sense of Atiyah, Segal, and Stolz-Teichner.
We give a detailed construction of this smooth  FFT, based on the definition of a suitable smooth bordism category. In this bordism category, all  manifolds are equipped with a smooth map to a spacetime target manifold. Further, the object manifolds are allowed to have boundaries; these are the endpoints of open strings stretched between D-branes.
The values of our FFT are obtained from the B-field and its D-branes via transgression.
Our construction  generalises work of Bunke-Turner-Willerton to include open strings.
At the same time,  it generalises work of Moore-Segal  about open-closed TQFTs to include target spaces.
We provide a number of further features of our FFT: we show that it depends functorially on the B-field and the D-branes, we show that  it is thin homotopy invariant, and we show that it comes equipped with a positive reflection structure in the sense of Freed-Hopkins.
Finally, we describe how our construction is related to the classification of open-closed TQFTs obtained by  Lauda-Pfeiffer.
\end{abstract}

\tableofcontents

\section{Introduction}
\label{sec:intro}

A \emph{topological quantum field theory} (TQFT) is a symmetric monoidal functor
\begin{equation}
\scZ \colon \Bord_d \to \Vect\,,
\end{equation} 
where $\Bord_d$ is a suitable category of closed oriented $(d{-}1)$-manifolds and $d$-dimensional bordisms, symmetric monoidal under the disjoint union of manifolds, and where $\Vect$ is the category of finite-dimensional complex vector spaces, monoidal under the tensor product.
The axioms of the symmetric monoidal functor $\scZ$ implement abstractly the sewing laws of the path integral.
The idea of this formalisation goes back to Atiyah~\cite{Atiyah:TQFTs} and Segal~\cite{Segal1987,Segal:Def_of_CFT}.
We refer to~\cite{Baez2006} for more information about the physical perspective.
TQFTs are very rich and interesting objects.
A crucial feature is that they can be classified in terms of algebraic objects: one chooses a presentation of the category $\Bord_d$ in terms of generators and relations and then translates these into algebraic data on the target side.
In this way, for example, 2-dimensional  TQFTs correspond to commutative Frobenius algebras~\cite{Dijkgraaf1989,Abrams1996,Kock2003}. 

There exist many interesting variations of the notion of a TQFT, which arise by including additional structure.
In this article, we study  four  modifications, all  at the same time, and show that the  2-dimensional sigma model with B-field and D-branes fits into this framework. Until now, only two special cases of this  picture have been worked out. We shall sketch below separately  the four modifications of the bordism category we study; all details are fully worked out in the main text.

\emph{(1) D-Branes.}
We consider a set $I$ of \emph{brane labels} and form a new category $\OCBord_d^{I}$, whose objects are compact oriented $(d{-}1)$-manifolds $Y$ with boundary, with each boundary component equipped with a brane label $i \in I$.
Morphisms $\Sigma \colon Y_0 \to Y_1$ are compact oriented $d$-manifolds with corners that now have an incoming boundary $\partial_0 \Sigma \cong Y_0$, an outgoing boundary $\partial_1 \Sigma \cong Y_1$, as well as additional \emph{brane boundary} $\partial_2 \Sigma$, whose components carry brane labels compatible with those of $Y_0$ and $Y_1$.
Symmetric monoidal functors \smash{$\scZ \colon \OCBord_d^{I} \to \Vect$} are called \emph{open-closed TQFTs with D-brane labels $I$}.
In two dimensions, open-closed TQFTs  have been discussed and classified by Lazaroiu~\cite{Lazaroiu:OCFFT_in_2D}, Moore-Segal~\cite{MS--2DTFTs}, and Lauda-Pfeiffer~\cite{LP--Open-closed_TQFTs}; they correspond to so-called $I$-coloured knowledgeable Frobenius algebras.
Our motivation to include brane labels is string theory, where one needs to consider  open and closed strings at the same time, and where the end-points of open strings are constrained to D-branes.

\emph{(2) Reflection-positivity.}
The bordism category $\Bord_d$ admits two canonical involutions, called \emph{dual} $(..)^\vee$ and \emph{opposite} $\overline{(..)}$.
The dual implements duals with respect to the symmetric monoidal structure, while the opposite is an a priori different operation that is usually not considered in the TQFT literature. The category $\Vect$ has similar involutions: the usual dual   $V^{\vee}$ and the complex conjugate  $\overline{V}$ of a vector space $V$. While any monoidal functor sends duals to duals, demanding that a functor $\scZ:\Bord_d \to \Vect$ sends opposite bordisms to complex conjugate vector spaces is a constraint, and requires an additional structure called a \emph{reflection structure} \cite{FH:Reflection_positivity}. Basically, it consists of natural isomorphisms $\overline{\scZ(Y)} \cong \scZ(\overline{Y})$ for all objects $Y$ of $\Bord_d$. It turns out that  a reflection structure induces an isomorphism between $\scZ(Y)^\vee$ and $\overline{\scZ(Y)}$, and hence a non-degenerate hermitean  form.
A reflection structure on $\scZ$ is then called \emph{positive} if that form is positive definite, for every object $Y$.

\emph{(3) Target spaces.}
Searching for new invariants of manifolds, Turaev considered bordism categories where all manifolds are equipped with a homotopy class of maps into a fixed topological space~\cite{Turaev}.
Stolz-Teichner~\cite{ST:What_is_an_elliptic_object} considered an even more refined bordism category $\Bord_d(M)$, where all manifolds are endowed with smooth maps to a smooth manifold $M$, \emph{not} taken up to homotopy.
Symmetric monoidal functors $\scZ \colon \Bord_d(M) \to \Vect$ will be called \emph{functorial field theories (FFTs) on $M$}.
A FFT $\scZ$ on $M$ may be invariant under changing the maps to $M$ by homotopies and thus reduce to one of Turaev's homotopy invariant FFTs.
Often, however, FFTs are only invariant under \emph{thin homotopies}, i.e.~homotopies whose differential has at most rank $d$; these FFTs will be called \emph{thin homotopy invariant FFTs on $M$}. Motivated by our construction of a FFT, we add a further  property that has not been considered before:  we call a  FFT $\scZ$ \emph{superficial} if it is  thin homotopy invariant and, in addition, the values of $\scZ$ agree on  two morphisms in $\Bord_d(M)$ with the same source and the same target whenever they have the same underlying $d$-manifold $\Sigma$ and their (not necessarily homotopic) smooth maps $\sigma, \sigma' \colon \Sigma \to M$ are thin in the sense that their differential is of rank strictly less than $d$ everywhere.

\emph{(4) Smoothness}. This becomes relevant upon including target spaces.
In physical field theories, it is generally crucial not only to describe the fields themselves, but also how the fields change in space and over time.
For instance, in order to derive the classical equations of motion, one  analyses how the action functional changes under smooth variations of the fields.
Thus, we consider \emph{smooth families} of manifolds, rather than just individual ones.
The formalism we use here has been invented by Stolz-Teichner~\cite{ST:SuSy_FTs_and_generalised_coho} and is based on presheaves of categories.
We define a presheaf $\Bord_d(M)$ of symmetric monoidal categories on a suitable category of test spaces.
Then, we define a \emph{smooth FFT on $M$} to be a morphism $\scZ \colon \Bord_d(M) \to \VBdl$ of presheaves of symmetric monoidal categories, where $\VBdl$ is the presheaf that assigns to a test space its symmetric monoidal category of vector bundles. Inserting the one-point test space always brings us back to the previous (discrete) setting, but in general this looses information.

Our motivation for passing from TQFTs to smooth FFTs is to think of $M$ as the background spacetime of a classical field theory.
The passage comprises a quite drastic change of perspective, since in general smooth FFTs are much richer than TQFTs: in our formalism, TQFTs turn out to be smooth FFTs on $M=\{\ast\}$. In fact, this is a theorem that we prove -- to our best knowledge -- here for the first time (Theorem \ref{st:ev_* is equivalence}):  the presheaf formalism of \emph{smooth} FFTs disappears for $M=\{\ast\}$ automatically and reduces the formalism indeed to the one of TQFTs.   
Consequently, smooth FFTs provide a common framework for classical and quantum theories.
The original motivation of TQFTs to represent the sewing laws of a path integral is now enlarged to a more fundamental statement about the integrands under the path integral.

Let us briefly describe known examples of smooth FFTs on a target space $M$, in small dimensions $d$. In dimension $d=1$, every vector bundle with connection over $M$ defines a 1-dimensional smooth FFT, see \cite{barret1,schreiber3,BEP:Smooth_1D_TFTs_are_VBdls_with_conns,LS:Smooth_1D_FFTs}, and each of these papers proves (in different formalisms) that in fact every 1-dimensional smooth FFT arises this way. Analogously, in dimension $d=2$,  a bundle gerbe with connection over $M$ defines a 2-dimensional smooth FFT on $M$. Again in a slightly different setting, this has been demonstrated by Bunke-Turner-Willerton~\cite{BTW--Gerbes_and_HQFTs}. The same smooth FFT can also be obtained by restricting the FFT we construct in this paper to the closed sector, as we show in Section~\ref{sec:closedsubsector}, but our construction is functorial, rather than defined on isomorphism classes only.   
Each of these FFTs turns out to be superficial in our sense.
For $M$ a  Lie group, many aspects of the relation between gerbes with connection and smooth FFTs have been treated earlier by Gaw\c edzki \cite{gawedzki16} and Freed~\cite{Freed2} in the process of understanding mathematical properties of Wess-Zumino-Witten models.
We remark that Bunke-Turner-Willerton~\cite{BTW--Gerbes_and_HQFTs} also prove that all 2-dimensional, invertible, thin homotopy invariant, smooth FFTs on $M$ arise from a bundle gerbe with connection over $M$. 

Next, we describe in more detail the framework we set up in this article, which provides a unified treatment of all four modifications described above. The target space is  a pair $(M,Q)$ of a smooth manifold $M$ and a family $\scQ = \{Q_i\}_{i \in I}$ of submanifolds $Q_i \subset M$; these submanifolds are supposed to support the D-branes, indexed by brane labels $i \in I$.
We define a presheaf  \smash{$\OCBord_d(M,\scQ)$} of oriented open-closed $d$-dimensional bordisms on the target space $(M,\scQ)$. If an object manifold has a boundary, then each connected component of the boundary is equipped with a brane label $i \in I$, and this component is mapped to $Q_i \subset M$.
The following is the central definition of this article, see Definition \ref{def:smooth field theories} in the main text: a \emph{smooth open-closed functorial field theory (OCFFT) on $(M,Q)$} is a morphism
\begin{equation}
        \scZ \colon \OCBord_d(M,\scQ) \longrightarrow \VBdl
\end{equation}
of presheaves of symmetric monoidal categories.
We also describe carefully the conditions under which we call a smooth OCFFT invertible (Definition \ref{def:invertible}), (thin) homotopy invariant, or superficial (Definition \ref{def:thinhomotopyinvariantocfft}). Moreover, we define reflection structures on smooth OCFFTs and explain positivity (Definitions \ref{def:reflection structure} and \ref{def:positive reflection structure}).

The main result of this paper is the functorial construction of a 2-dimensional, invertible, reflection-positive, superficial, smooth OCFFT on a target space $(M,\scQ)$, taking as input a  \emph{target space brane geometry} on $(M,\scQ)$.
These are pairs $(\CG,\CE)$ consisting of a bundle gerbe $\CG$ with connection on $M$ (in string theory called a \quot{B-field}), and of a family $\CE = \{\CE_i\}_{i \in I}$ of twisted vector bundles with connection over the submanifolds $Q_i$ (the \quot{Chan-Paton bundles}).
They enter  the OCFFT precisely as expected and as proposed by string theory: the bundle gerbe connection contributes a Wess-Zumino-term~\cite{WZ:Anomalous_Ward,Witten:Current_algebra,gawedzki3,Gawedzki:branes_in_WZW-models_and_gerbes}, and the twisted vector bundles describe the coupling of the end points of open strings to the D-branes~\cite{Kapustin:D-branes_in_nontriv_B-fields,Gawedzki-Reis:WZW-branes_and_gerbes,CJM--Holonomy_on_D-branes}. Our construction simultaneously generalises the FFT of Bunke-Turner-Willerton to include open strings and morphisms of target space brane geometries, and the open-closed TQFT of Lazaroiu, Moore-Segal and Lauda-Pfeiffer to include a target space and smoothness.

Let us now outline  some important steps in our constructions and describe how the paper is organised.
\begin{enumerate}[itemsep=-0.1cm, leftmargin=*, topsep=0cm, label=(\alph*)]
\item
The complex vector spaces that our OCFFT assigns to 1-manifolds are constructed in Section~\ref{sect:background}.
A substantial part of the construction has been performed in our previous paper~\cite{BW:Transgression_and_regression_of_D-branes}:  the transgression of bundle gerbes and D-branes to spaces of loops and paths in $M$. We review the relevant parts in Section~\ref{sect:bundles on path spaces} and Section~\ref{sect:LBdl_on_loop_space} of the present article.
While in~\cite{BW:Transgression_and_regression_of_D-branes} it sufficed to parameterise paths and loops in $M$ by $[0,1]$ and $\bbS^1$, respectively, for our present purposes we need to extend the formalism to oriented manifolds $Y$ that are only diffeomorphic to either $[0,1]$ or $\bbS^1$.
This is achieved using an enriched two-sided simplicial bar construction and descent theory~\cite{Bunk:Coh_Desc}; we summarise the construction in Section~\ref{sect:pull-push and coherence}.
The vector spaces our OCFFT assigns to closed 1-manifolds are obtained from the fibres of the hermitean line bundle  over the loop space of $M$ (the transgression of the bundle gerbe) and the vector spaces assigned to open 1-manifolds are fibres of hermitean vector bundles over spaces of paths connecting two D-branes (the transgression of the twisted vector bundles).
In particular, all these vector spaces are equipped with hermitean inner products, which are -- as OCFFTs take values in bare vector bundles -- discarded in this step.

\item
The linear maps that our smooth OCFFTs assign to 2-dimensional bordisms are constructed  in Section~\ref{sect:amplitudes}.
There, we actually adopt a dual picture and construct instead the scattering amplitude associated to the bordism for given incoming and outgoing states.
Our construction extends and conceptually simplifies well-known constructions of~\cite{Kapustin:D-branes_in_nontriv_B-fields,Gawedzki-Reis:WZW-branes_and_gerbes,CJM--Holonomy_on_D-branes}, while also providing a careful treatment of corners.

\end{enumerate}

\noindent   
The definition of  the presheaf $\OCBord_d{(M,\scQ)}$ of open-closed bordisms as well as the corresponding definition of smooth OCFFTs, as described above, is presented in Section~\ref{sect:Smooth TFTs}. 
In Section~\ref{sect:TFT_construction} we show the following main result of this article, see Theorem \ref{st:BGrb plus D-branes yields OC TFT}, Proposition \ref{prop:reflectionstructure}, Corollary \ref{co:positivity}, and Theorem \ref{st:Z_(-) is a functor}.

\begin{maintheorem}
\label{st:summary thm 1}
Consider a target space brane geometry $(\CG, \CE)$ on a target space $(M,Q)$.
The constructions (a) and (b) yield a smooth OCFFT
\begin{equation}
        \scZ_{\CG,\CE} \colon \OCBord_2{(M,\scQ)} \longrightarrow \VBdl\,.
\end{equation}
Moreover, this smooth OCFFT $\scZ_{\CG,\CE}$ has the following properties:
\begin{itemize}

\item
It is invertible.

\item 
It is superficial, and in particular thin homotopy invariant.

\item
There is a canonical positive reflection structure on $\scZ_{\CG,\CE}$.
It recovers precisely those hermitean inner products on the vector spaces in the image of $\scZ_{\CG,\CE}$ that have been discarded earlier. 

\item
The dependence on the target space brane geometry is functorial.
That is, we obtain a functor
\begin{equation}
        \scZ \colon \rmh_1 \TBG(M,Q) \longrightarrow \rmRP\text{-}\OCFFT_{2}^{\sf}(M,\scQ)^{\times}\,,
        \quad
        (\CG,\CE) \longmapsto \scZ_{\CG,\CE}
\end{equation}
from the homotopy groupoid of target space brane geometries on $(M,Q)$ to the groupoid of 2-dimensional, invertible, superficial, reflection-positive, smooth OCFFTs on $(M,Q)$. 
\end{itemize}
\end{maintheorem}

A number of further properties of the  OCFFT $\scZ_{\CG,\CE}$ are investigated in  Section~\ref{sect:Subsectors and kinetic terms}.
In Section \ref{sec:closedsubsector} we look at closed subsectors of the theory; we relate the restriction of $\scZ_{\CG,\CE}$ to these subsectors to previous work, in particular to that of Bunke-Turner-Willerton \cite{BTW--Gerbes_and_HQFTs}. In Section \ref{sec:flatgerbes} we prove the result (Theorem \ref{th:flat}) that the  OCFFT $\scZ_{\CG,\CE}$ is homotopy invariant if and only if the connection on the bundle gerbe $\CG$ is flat.
In the remaining subsections we concentrate on the reduction of our results to the case of a one-point target space  $M = \{\ast\}$, and explore in detail the relation to the work of Lazaroiu~\cite{Lazaroiu:OCFFT_in_2D}, Lauda-Pfeiffer~\cite{LP--Open-closed_TQFTs} and Moore-Segal~\cite{MS--2DTFTs}.
The following theorem summarises these relations.
\begin{maintheorem}
\label{st:summary thm 2}
There is a strictly commutative diagram of equivalences of categories:
\begin{equation}
\xymatrix@C=2cm@R=1cm{
        \mathrm{h}_1\TBG(I) \ar[d]_{\wtfrob} \ar[r]^-{\scZ} & \rmRP\text{-}\OCFFT_{2}(I)^{\times} \ar[d]^-{\ev_*}
        \\
        \RPKFrob^{I}_\FC & (\rmRP\text{-}\OCTQFT_{2}^{I})^{\times} \ar[l]^-{\scF}
}
\end{equation}
\end{maintheorem}

The categories in the corners of this diagram are the following. $\TBG(I)$ is the bicategory of target space brane geometries for the one-point target space (its only information is the index set of brane labels), and $\mathrm{h}_1\TBG(I)$ is its homotopy groupoid.  $\RPKFrob^{I}_\FC$ is the category of  $I$-coloured knowledgeable Frobenius algebras whose bulk algebra is isomorphic to $\FC$, equipped with a version of a reflection-positive structure. The functor $\wtfrob$  is described in detail in Section~\ref{sect:Reduction to point}; it arises from the geometric formalism in~\cite{BW:Transgression_and_regression_of_D-branes}.
$(\rmRP\text{-}\OCTQFT_{2}^{I})^{\times}$ is the category of  2-dimensional, invertible,  reflection-positive, open-closed TQFTs, and $\scF$ is the functor that establishes the classification result of Lauda-Pfeiffer~\cite{LP--Open-closed_TQFTs}, enhanced to include reflection structures. Finally,  $\rmRP\text{-}\OCFFT_{2}(I)^{\times}$ is the category of 2-dimensional, invertible, reflection-positive, smooth OCFFTs on the one point target space.
The functor $\ev_{*}$ evaluates a morphism of presheaves of categories on the one-point test space.
As mentioned above, it is an important statement on its own that this functor is an equivalence, which is the content of Theorem~\ref{st:ev_* is equivalence}.

In upcoming work we aim to show that our functor
\begin{equation}
        \scZ \colon \rmh_1\TBG(M,Q) \longrightarrow \rmRP\text{-}\OCFFT^{\sf}_{2}(M,\scQ)^{\times}\,,
        \qquad
        (\CG, \CE) \longmapsto \scZ_{\CE,\CG}
\end{equation}
is an equivalence of categories, for all target spaces $(M,\scQ)$. This would yield a complete classification of 2-dimensional, invertible, reflection-positive, superficial, smooth OCFFTs by target space brane geometries.

\subsection*{Acknowledgments}
We thank Ulrich Bunke, Tobias Dyckerhoff, Lukas Müller, Alexander Schenkel, Urs Schreiber, Stephan Stolz, Richard Szabo, and Peter Teichner for interesting discussions and comments.
S.B. was partially supported by RTG 1670, \emph{Mathematics Inspired by String Theory and Quantum Field Theory}.

\section{Transgression of B-fields and D-branes}
\label{sect:background}

\subsection{Target space brane geometry}
\label{sect:BGrbs}

In this section we recall the basic definitions of bundle gerbes and D-branes as well as important results on their structure on which our constructions in the later sections will be based.
Bundle gerbes have been introduced in~\cite{Murray:Bundle_gerbs}.
The notion of morphisms that we are going to employ originated in~\cite{Murray-Stevenson:Bgrbs--stable_isomps_and_local_theory} and has been generalised in~\cite{Waldorf--More_morphisms,Waldorf--Thesis}.
For detailed proofs of the statements on the 2-category of bundle gerbes we refer the reader to the above references as well as~\cite{BSS--HGeoQuan,Bunk--Thesis}; for a non-technical introduction see also~\cite{Bunk-Szabo:Fluxes_brbs_2Hspaces}.
The conventions used in this article are compatible with our previous article~\cite{BW:Transgression_and_regression_of_D-branes}.
We denote by $\Mfd$ the category of smooth manifolds and smooth maps.

The approach to bundle gerbes we present here is the  compact treatment worked out in~\cite{NS--Equivariance_in_higher_geometry}.
Let $\HVBdl^\nabla(M)$ denote the category of hermitean vector bundles with connection on a manifold $M$.
To every manifold $M$ we assign a 2-category $\TrivGrb^\nabla(M)$ whose objects are 2-forms $B \in \Omega^2(M)$ and whose morphism categories $\CHom(B_1,B_2)$ are  the full subcategories
\begin{equation}
        \HVBdl^\nabla(M)^{B_2 - B_1} \subset \HVBdl^\nabla(M)
\end{equation}
on those hermitean vector bundles $E$ over $M$ with (unitary) connection whose curvature satisfies
\begin{equation}
\label{eq:curv}
\frac{1}{\rank(E)}\, \tr \big( \curv(E) \big) = B_2 - B_1\,.
\end{equation}
The composition of 1-morphisms and the horizontal composition of 2-morphisms are induced by the tensor product in $\HVBdl^\nabla(M)$.
The assignment $M \mapsto \TrivGrb^\nabla(M)$ defines a presheaf of symmetric monoidal 2-categories on $\Mfd$.

The sheafification of the presheaf $\TrivGrb^\nabla$ with respect to the Grothendieck topology of surjective submersions yields a sheaf $\Grb^\nabla$ of symmetric monoidal 2-categories on $\Mfd$. Its sections are called bundle gerbes with connections.
In short, a bundle gerbe $\CG$ with connection consists of a surjective submersion $\pi:Y \to M$, a 2-form $B\in \Omega^2(Y)$, a hermitian line bundle $L$ with connection over the 2-fold fibre product $Y^{[2]} = Y \times_M Y$, whose curvature satisfies $\mathrm{curv}(L)=\pr_2^{*}B-\pr_1^{*}B$ (called the \emph{curving} of $\CG$), and a unitary, connection-preserving line bundle isomorphism $\mu: \pr_{12}^{*}L \otimes \pr_{23}^{*}L \to \pr_{13}^{*}L$ over $Y^{[3]}$ which satisfies an associativity condition over $Y^{[4]}$. A morphism between bundle gerbes $\CG_{1}$ and $\CG_2$ with connections consists of a surjective submersion $\zeta: Z \to Y_1 \times_M Y_2$, a hermitian vector bundle $E$ over $Z$ with connection whose curvature satisfies \eqref{eq:curv}, and a unitary, connection-preserving vector bundle isomorphism
\begin{equation}
	\alpha \colon \zeta^{[2]*}\pr_{Y_1}^{[2]*} L_1 \otimes \pr_1^*E \longrightarrow \pr_2^* E \otimes \zeta^{[2]*}\pr_{Y_2}^{[2]*} L_2
\end{equation}
over $Z^{[2]}$ that is compatible with the line bundle morphisms $\mu_1$ and $\mu_2$.
Here, all fibre products are taken over $M$, the maps $\pr_{Y_i}$ are the canonical projections $Y_1 \times_M Y_2 \to Y_i$, and the projections $\pr_i$ are the canonical projections $Z^{[2]} \to Z$ onto the $i$-th factor.
In the following, for the sake of legibility we will usually not display these pullbacks explicitly unless it might cause confusion.
One can show that a morphism $\CE \colon \CG \to \CG'$ of bundle gerbes is invertible if and only if its  vector bundle $E \to Z$ is of rank one; see~\cite{Waldorf--More_morphisms}.

\begin{example}
Of crucial importance are the \emph{trivial bundle gerbes} $\CI_\rho$, where $\rho \in \Omega^2(M)$ is any 2-form on $M$.
These objects are exactly those in the image of the canonical inclusion
\begin{equation}
        \TrivGrb^\nabla(M) \hookrightarrow \Grb^\nabla(M)\,.
\end{equation}
Concretely, the gerbe $\CI_\rho$ has as its surjective submersion the identity $1_M \colon M \to M$, as its line bundle with connection the trivial line bundle $M \times \FC \to M$ with the trivial connection, and its structure isomorphism $\mu$ acts as $\mu((x,z_1), (x,z_2)) = (x, z_1 z_2)$ for all $x \in M$ and $z_1, z_2 \in \FC$.
Finally, the curving of $\CI_\rho$ is given by the 2-form $\rho \in \Omega^2(M)$.

A \emph{trivialisation} of a bundle gerbe $\CG$ with connection on $M$ is a 1-isomorphism $\CT \colon \CG \arisom \CI_\rho$ for some $\rho \in \Omega^2(M)$.
The bundle gerbe $\CI_0$ is the monoidal unit of $\Grb^\nabla(M)$.
\qen
\end{example}

\begin{example}
\label{eg:sfR}
Consider a 1-morphism $\CE \colon \CI_\omega \to \CI_\rho$ between two trivial bundle gerbes on $M$.
This is a hermitean vector bundle $E \rightarrow Z$ with connection over a surjective submersion $\zeta \colon Z \rightarrow M$, and a connection-preserving isomorphism $\alpha \colon \pr_1^*E \rightarrow \pr_0^*E$ over $Z^{[2]}$, satisfying a cocycle condition over $Z^{[3]}$.
In other words, $\CE$ is a descent datum for a hermitean vector bundle $\sfR \CE$ with connection on $M$.
Descent thus induces an equivalence of categories
\begin{equation}
        \sfR \colon \CHom(\CI_\omega, \CI_\rho) \arisom \HVBdl^\nabla(M)^{\rho - \omega}\,.
\end{equation}
An inverse of this functor is given by the canonical inclusion
\begin{equation}
        \HVBdl^\nabla(M)^{\rho - \omega} \hookrightarrow \CHom(\CI_\omega, \CI_\rho)
\end{equation}
as a full subcategory.
\qen
\end{example}

We briefly recall that for any bundle gerbe with connection $\CG \in \Grb^\nabla(M)$ and any 2-forms $\omega_1, \omega_2 \in \Omega^2(M)$ there exists a functor
\begin{equation}
        \Delta \colon \CHom(\CG, \CI_{\omega_2}) \times \CHom(\CG, \CI_{\omega_1})^\opp \longrightarrow \HVBdl^\nabla(M)^{\omega_2 - \omega_1}\,,
\end{equation}
which derives from the closed module category structure of the morphism categories of $\Grb^\nabla(M)$ over the category $\HVBdl^\nabla(M)^0$~\cite{BSS--HGeoQuan,Bunk--Thesis}.
Explicitly, $\Delta$ can be defined as follows, see~\cite{BW:Transgression_and_regression_of_D-branes}.
We denote by
\begin{equation}
        (-)^* \colon \Grb^\nabla(M)^{\opp} \to \Grb^\nabla(M)
\end{equation}
the 2-functor that is the identity on objects, that sends a 1-morphism to the morphism defined by the dual vector bundle, and that sends a 2-morphism to the 2-morphism induced by the fibre-wise transpose of the original 2-morphism.
The functor $(-)^*$ reverses the direction of 1-morphisms and 2-morphisms.
Then we define $\Delta$ to be the composition
\begin{equation}
\xymatrix@C=3cm{\CHom(\CG,\CI_{\omega_2}) \otimes \CHom(\CG, \CI_{\omega_1})^\opp \ar[r]^{{1 \times (-)^*}} \ar@{-->}[d]_{\Delta}
        & \CHom(\CG,\CI_{\omega_2}) \otimes \CHom(\CI_{\omega_1}, \CG) \ar[d]^{{(-) \circ (-)}}
        \\
        \HVBdl^\nabla(X)^{\omega_2 - \omega_1} & \CHom(\CI_{\omega_1}, \CI_{\omega_2}) \ar[l]^{\sfR}}
\end{equation}
where $\sfR$ is the descent functor as in Example~\ref{eg:sfR}.
That is, we set
\begin{equation}
\label{eq:def Delta}
        \Delta(\CE, \CF) \coloneqq \sfR( \CE \circ \CF^*)\,.
\end{equation}
If $\CT \colon \CG \to \CI_{\omega_1}$ is an isomorphism, we have 
$\Delta(\CE,\CT) \cong \sfR( \CE \circ \CT^{-1})$;
this follows from the Definition~\eqref{eq:def Delta} of the functor $\Delta$ and the fact that $\CT^* = \CT^{-1}$ for any 1-isomorphism of bundle gerbes~\cite{Waldorf--Thesis}.
We further recall from~\cite[Remark~2.1.1]{BW:Transgression_and_regression_of_D-branes} that there are canonical morphisms
\begin{align}
        \Delta(\CE_3, \CE_2) \otimes \Delta(\CE_2, \CE_1) &\longrightarrow \Delta(\CE_3,\CE_1)\,,
        \\*
        \label{eq:structural isos for Delta}
        \Delta(\CE_2, \CE_1) &\arisom \Delta(\CE_1, \CE_2)^\vee\,,
        \\*
        \Delta(\CE_2 \circ \CA,\, \CE_1 \circ \CA) &\arisom \Delta(\CE_2, \CE_1)
\end{align}
for all 1-morphisms $\CE_a \colon \CG \to \CI_{\omega_a}$ for $a = 1,2,3$ and isomorphisms $\CA \colon \CG' \to \CG$ in $\Grb^\nabla(M)$.
Here, the first morphism is induced by composition, the second is induced by the dual in the category of hermitean vector bundles with connection, and the third is induced from the evaluation isomorphism for $\CA$.
Moreover, there exists a canonical isomorphism
\begin{equation}
\label{eq:1-mors and sections of Delta}
        \Grb^\nabla(M)(\CE_1, \CE_2) \cong \Gamma_{\rmpar, \uni} \big( M, \Delta(\CE_2, \CE_1) \big)\,,
\end{equation}
where $\Gamma_{\rmpar, \uni}$ is the functor that takes parallel unit-length global sections of a hermitean vector bundle with connection.
This can be checked directly from the definition of $\Delta$.

To conclude this section, we recall the definitions of a D-brane and of a target space brane geometry:

\begin{definition}
\label{def:D-brane}
Let $\CG \in \Grb^\nabla(M)$ be a bundle gerbe over $M$, and let $Q \subset M$ be a submanifold.
A \emph{D-brane for $\CG$} supported on $Q$ is a morphism $\CE \colon \CG_{|Q} \to \CI_\omega$ of bundle gerbes over $Q$ for some 2-form $\omega \in \Omega^2(Q)$.
\end{definition}

We call a pair $(M,Q)$ of a manifold $M$ and a collection $\scQ = \{Q_i\}_{i \in I}$ of submanifolds in $M$ a \emph{target space}.
Given a target space $(M,Q)$, we define a 2-groupoid \smash{$\TBG(M,\scQ)$} of \emph{target space brane geometries on $(M,\scQ)$} as follows.
Its objects are pairs $(\CG,\CE)$ of a bundle gerbe with connection $\CG \in \Grb^\nabla(M)$ and of a family $\CE=\{\CE_{i}\}_{i\in I}$ of D-branes $\CE_i \colon \CG_{|Q_i} \to \CI_{\omega_i}$ supported on the submanifolds $Q_i$.
A 1-morphism 
\begin{equation*}
(\CA, \xi) \colon (\CG, \CE) \to (\CG', \CE')
\end{equation*}
in \smash{$\TBG(M,\scQ)$} consists of a 1-isomorphism $\CA \colon \CG \to \CG'$ in $\Grb^\nabla(M)$ together with a family $\xi=\{\xi_i\}_{i\in I}$ of 2-isomorphism $\xi_i \colon \CE_i \to \CE'_i \circ \CA_{|Q_i}$.
Finally, a 2-isomorphism $(\CA, \xi) \to (\CA', \xi')$ in $\TBG(M,Q)$ is given by a  2-isomorphism $\psi \colon \CA \to \CA'$ in $\Grb^\nabla(M)$ such that the diagram
\begin{equation}
\label{eq:2-morphism in TBG}
\xymatrix{\CE_i \ar[rr]^-{\xi_i} \ar[dr]_-{\xi'_i} &&  \CE'_i \circ \CA_{|Q_i} \ar[dl]^{1_{\CE'_i} \circ \psi_{|Q_i}}
        \\
        & \CE'_i \circ \CA'_{|Q_i}}
\end{equation}
commutes for every D-brane label $i \in I$.
We refer to \cite{BW:Transgression_and_regression_of_D-branes} for a more detailed discussion of target space brane geometry.

\subsection{The transgression bundles over paths between D-branes}
\label{sect:bundles on path spaces}

Let $(M,Q)$ be a target space, and let $(\CG,\CE) \in \TBG(M,\scQ)$ be a target space brane geometry.
In this section we review how  $(\CG,\CE)$ give rise to vector bundles over the spaces of paths between the submanifolds $Q_i$~\cite{BW:Transgression_and_regression_of_D-branes}.

Consider two arbitrary brane labels $i, j \in I$, and let $P_{ij}M$ be the diffeological space of smooth paths $\gamma:[0,1] \to M$ in $M$ with sitting instants that start in $Q_i$ and end in $Q_j$.
For a path $\gamma \in P_{ij}M$ and a trivialisation $\CT \colon \gamma^* \CG \to \CI_0$ of the pullback gerbe over the interval, we set
\begin{equation}
\label{eq:defnitionCE}
        \scR_{ij|\gamma}(\CT) \coloneqq \Hom \big( \Delta(\CE_{i|\gamma(0)}, \CT_{|0}),\, \Delta(\CE_{j|\gamma(1)}, \CT_{|1}) \big)\,.
\end{equation}
For later use, we point out that one can alternatively write
\begin{equation}
\label{eq:CE via composition}
        \scR_{ij|\gamma}(\CT) \cong \Hom \big( \sfR(\CE_{i|\gamma(0)} \circ \CT^{-1}_{|0}),\, \sfR(\CE_{j|\gamma(1)} \circ \CT^{-1}_{|1}) \big)\,.
\end{equation}
By construction, $\scR_{ij|\gamma}(\CT)$ is a Hilbert space of finite dimension
\begin{equation}
        \dim_\FC \big( \scR_{ij|\gamma}(\CT) \big) = \rank(\CE_j) \cdot \rank(\CE_i)\,.
\end{equation}
The Hilbert space $\scR_{ij|\gamma}(\CT)$ depends on the choice of a trivialisation $\CT$ of $\gamma^*\CG \in \Grb^\nabla([0,1])$.
However, changing the trivialisation $\CT$ to a trivialisation $\CT'$ changes the Hilbert space $\scR_{ij|\gamma}(\CT)$ only up to a canonical isomorphism.
To see this, we note that there always exists a 2-isomorphism $\psi \colon \CT \to \CT'$ of trivialisations of $\gamma^* \CG$ --  this follows from~\eqref{eq:1-mors and sections of Delta}, for instance.
Any such 2-isomorphism induces a canonical isomorphism
\begin{equation}
\label{eq:def of r_(CT,CT')}
        r_\psi \colon \scR_{ij|\gamma}(\CT') \longrightarrow \scR_{ij|\gamma}(\CT)\,,
        \quad
        \varphi \longmapsto \psi_1 \circ \varphi \circ \psi_0^{-1}\,,
\end{equation}
where we have abbreviated $\psi_t \coloneqq \Delta(1, \psi_{|t})$ for $t \in \{0,1\}$. The following result was proved in \cite[Section~4.2]{BW:Transgression_and_regression_of_D-branes}.

\begin{lemma}
The following statements hold true:
\begin{enumerate}[(1)]
        \item For any other isomorphism $\psi' \colon \CT \to \CT'$, we have $r_\psi = r_{\psi'}$.
        We can thus denote this isomorphism by $r_{\CT,\CT'} \colon \scR_{ij|\gamma}(\CT') \longrightarrow \scR_{ij|\gamma}(\CT)$.
        
        \item Given a third trivialisation $\CT''$ of $\gamma^*\CG$, we have
        \begin{equation}
        \label{eq:Cocyle_PsiSS'}
                r_{\CT'',\CT'} \circ r_{\CT',\CT} = r_{\CT'',\CT}
                \qandq
                r_{\CT,\CT} = 1_{\scR_{ij|\gamma}(\CT)}\,.
        \end{equation}
\end{enumerate}
\end{lemma}

\noindent
We thus define
\begin{equation}
\label{eq:def R_ij}
\begin{aligned}
        \scR_{ij|\gamma} &\coloneqq \big\{ (\CT,\phi)\, \big|\, \CT \text{ trivialisation of } \gamma^*\CG,\, \phi \in \scR_{ij|\gamma}(\CT) \big\}/{\sim}\,,
        \\*[0.1cm]
        &\quad \text{where} \quad (\CT, \phi) \sim \big( \CT', \phi' \big)
        \Longleftrightarrow \phi' = r_{\CT',\CT}(\phi)\,.
\end{aligned}
\end{equation}
The finite-dimensional Hilbert space $\scR_{ij|\gamma}$ comes with canonical isomorphisms
\begin{align}
\label{eq:colim str mors r_CT}
        &r_\CT \colon \scR_{ij|\gamma}(\CT) \rightarrow \scR_{ij|\gamma}\,,
        \quad
        \phi \mapsto [\CT, \phi]      
\end{align}
satisfying $r_{\CT'} \circ r_{\CT',\CT} = r_{\CT'}$ for any trivialisations $\CT, \CT'$ of $\gamma^*\CG$.
The disjoint union
\begin{equation}
\label{eq:E_ij as disjoint union}
        \scR_{ij} \coloneqq \bigsqcup_{\gamma \in P_{ij} M} \scR_{ij|\gamma}
\end{equation}
comes with a canonical map $\pi \colon \scR_{ij} \to {P_{ij} M}$ which is a diffeological hermitean vector bundle of rank $\rank(\scR_{ij}) = \rank(\CE_i)\, \rank(\CE_j)$ over the diffeological space $P_{ij}M$~\cite[Proposition~4.2.3]{BW:Transgression_and_regression_of_D-branes}.

For later use, we briefly recall the definition of the diffeology on $\scR_{ij}$ from~\cite{BW:Transgression_and_regression_of_D-branes}.
Instead of using open subsets $U \subset \FR^n$ to define plots of diffeological spaces,  we only use cartesian spaces $U \in \Cart$, i.e.~embedded submanifolds $U \subset \FR^m$ that are diffeomorphic to $\FR^n$ for some $n,m \in \NN_0$.
Since any of the former test spaces can be covered by the latter, the resulting categories of diffeological spaces are equivalent.
A map $\hat{c} \colon U \to \scR_{ij}$ from an arbitrary cartesian space $U \in \Cart$ to $\scR_{ij}$ is a plot if it has the following properties:
we demand that the composition $c \coloneqq \pi \circ \hat{c} \colon U \to P_{ij}M$ is a plot of $P_{ij}M$.
This means, equivalently, that the adjoint map
\begin{equation}
\xymatrix{c^\dashv \colon U \times [0,1] \ar[r]^-{c \times 1} & P_{ij}M \times [0,1]  \ar[r]^-{\ev} & M}
\end{equation}
is a smooth map. Since $U \cong \FR^n$ for some $n \in \NN_0$, there exists a trivialisation $\CT \colon c^{\dashv*} \CG \to \CI_\rho$.
Define the inclusion maps $\iota_t \colon U \hookrightarrow U \times [0,1]$, $x \mapsto (x,t)$, for any $t \in [0,1]$.
Over $U$ we obtain the hermitean vector bundles
\begin{equation}
        F_i \coloneqq \Delta \big( (c^\dashv \circ \iota_0)^*\CE_i, \iota_0^*\CT \big)
        \qandq
        F_j \coloneqq \Delta \big( (c^\dashv \circ \iota_1)^*\CE_j, \iota_1^*\CT \big)\,,
\end{equation}
and we demand that there exists an open covering $\{U_a\}_{a \in A}$ of $U$, together with morphisms $\psi_a \colon F_{i|U_a} \to F_{j|U_a}$ over $U_a$ for all $a \in A$ such that
\begin{equation}
        \hat{c}_{|U_a}(x) = [ \CT_{|\{x\} \times [0,1]},\, \psi_{a|x} ]
        \qquad
        \forall\, a \in A,\, x \in U_a\,.
\end{equation}

We move on to briefly recall from~\cite{BW:Transgression_and_regression_of_D-branes} the construction of the connection on the bundles $\scR_{ij} \to P_{ij}M$, which is defined via its parallel transport.
Consider a smooth path $\Gamma \colon [0,1] \to P_{ij}M$, and denote its adjoint map by $\Gamma^\dashv \colon [0,1]^2 \to M$.
Further, we write $\Gamma_t \coloneqq \Gamma^\dashv(-,t) \colon [0,1] \to M$, with $t \in \{0,1\}$, for the paths of endpoints.
We let $\CT \colon \Gamma^{\dashv*}\CG \to \CI_\rho$ be a trivialisation and set $\CT^t \coloneqq \CT_{|[0,1] \times \{t\}}$ for $t = 0,1$.
This induces hermitean vector bundles with connection
\begin{equation}
\label{eq:aux bdls for pt_ij}
        \CV_{(i,\CT^0)} \coloneqq \Delta \big( \Gamma_0^*\CE_i, \CT^0 \big)\,,
        \qandq
        \CV_{(j,\CT^1)} \coloneqq \Delta \big( \Gamma_1^*\CE_j, \CT^1 \big)
\end{equation}
over the interval $[0,1]$.
Observe that  there is a canonical isomorphism
\begin{equation}
        \Gamma^{\dashv *} \scR_{ij} \cong \Hom(\CV_{(i,\CT^0)}, \CV_{(j,\CT^1)})\,,
\end{equation}
which is induced by the isomorphisms $r_{(\CT_{|\{s\} \times [0,1]})}$ from~\eqref{eq:colim str mors r_CT}.
We then define
\begin{equation}
\label{eq:pt in R_ij}
        pt_{ij|\Gamma}(\phi) \colon \scR_{ij|\Gamma(0)}(\CT_{\{0\} \times [0,1]})
        \longrightarrow \scR_{ij|\Gamma(1)}(\CT_{\{1\} \times [0,1]})\,,
        \qquad
        \phi \longmapsto \exp \bigg( \int_{[0,1]^2} - \rho \bigg) \cdot pt(\phi)\,,
\end{equation}
where on the right-hand side, $pt$ denotes the parallel transport in $\Hom(\CV_{(i,\CT^0)}, \CV_{(j,\CT^1)})$ along the interval.
Here we orient the boundary of $[0,1]^2$ using an \emph{inward-pointing} normal vector field to match our later conventions -- this explains the different sign as compared to~\cite{BW:Transgression_and_regression_of_D-branes}.  The following result is \cite[Proposition~4.3.1]{BW:Transgression_and_regression_of_D-branes}.
Recall the notion of a superficial connection on a diffeological vector bundle over $P_{ij}M$ from~\cite[Def.~A.2.2]{BW:Transgression_and_regression_of_D-branes}.

\begin{proposition}
The morphism $pt_{ij}$ of Equation~\eqref{eq:pt in R_ij} defines a superficial connection on the vector bundle $\scR_{ij}$ over $P_{ij}M$.
\end{proposition}

\subsection{Equivariant structure on the transgression bundles}
\label{sect:equivar strs on R_ij}

Let $\Diff^+([0,1])$ denote  the group of orientation-preserving diffeomorphisms of the interval.
Note that such diffeomorphisms automatically fix the boundary points.
The group $\Diff^+([0,1])$ is a diffeological group when endowed with the usual mapping space diffeology, where a map $f \colon U \to \Diff^+([0,1])$ is a plot if the adjoint map
\begin{equation}
\xymatrix{
        f^\dashv \colon U \times [0,1] \ar[rr]^-{f \times 1_{[0,1]}} & & \Diff^+([0,1]) \times [0,1] \ar[rr]^-{\ev} & & {[0,1]}
}
\end{equation}
is a smooth map.

\begin{lemma}
\label{st:Diff action on P_ij M}
Given a target space $(M,Q)$ and D-brane labels $i,j \in I$, the group $\Diff^+([0,1])$ acts smoothly on the space $P_{ij}M$ via
\begin{equation}
        R \colon P_{ij}M \times \Diff^+([0,1]) \longrightarrow P_{ij}M\,,
        \qquad
        (\gamma, \tau) \longmapsto \tau^*\gamma = \gamma \circ \tau\,.
\end{equation}
\end{lemma}

\begin{proof}
This follows by the construction of the mapping space diffeology~\cite[Paragraph~1.59]{Iglesias-Zemmour:Diffeology}.
\end{proof}

\begin{lemma}
The action $R$ lifts to a smooth action
\begin{equation}
\label{eq:Diff action on R_ij}
        R \colon \scR_{ij} \times \Diff^+ \big( [0,1] \big) \longrightarrow \scR_{ij}\,,
        \qquad
        \big( [\CT, \psi], \tau \big) \longmapsto [\tau^*\CT, \psi]\,.
\end{equation}
This turns $\scR_{ij}$ into a $\Diff^+([0,1])$-equivariant hermitean vector bundle on $P_{ij}M$.
\end{lemma}

Here we have used that $\tau$ is the identity on $\partial[0,1]$, so that the action on $\psi$ is trivial.

\begin{proof}
We let $f \colon U \to \Diff^+([0,1])$ be a plot, and we consider a plot $\hat{c} \colon U \to \scR_{ij}$.
In the notation of Section~\ref{sect:bundles on path spaces} that means that there exists a trivialisation $\CT \colon c^{\dashv*}\CG \to \CI_\rho$ (where $c \coloneqq \pi \circ \hat{c}$), an open covering $\{U_a\}_{a \in A}$ and morphisms $\psi_a \colon F_{i|U_a} \to F_{j|U_a}$ such that
\begin{equation}
        \hat{c}_{|U_a}(x) = [ \CT_{|\{x\} \times [0,1]},\, \psi_{a|x} ]
        \qquad
        \forall\, a \in A,\, x \in U_a\,.
\end{equation}
We need to show that $R \circ (\hat{c} \times f) \circ \Delta_U$ defines a plot of $\scR_{ij}$, where $\Delta_U \colon U \to U \times U$ denotes the diagonal map.
First, observe that
\begin{equation}
        \pi \circ \big( R \circ (\hat{c} \times f) \circ \Delta_U \big)
        = R \circ (c \times f) \circ \Delta_U\,,
\end{equation}
which is a plot of $P_{ij}M$ by Lemma~\ref{st:Diff action on P_ij M}.
The plot $f$ defines a fibre-wise diffeomorphism
\begin{equation}
        \hat{f} \colon U \times [0,1] \to U \times [0,1]\,,
        \qquad
        (x,t) \mapsto \big( x, f(x)(t) \big)\,.
\end{equation}
Over $U_a$, we  then have
\begin{equation}
        \big( R \circ (\hat{c} \times f) \circ \Delta_{U_a} \big)
        = R \big( [\CT_{|U_a}, \psi_a], \hat{f}_{|U_a} \big)
        = \big[ (\hat{f}^*\CT)_{|U_a}, \psi_a \big]\,.
\end{equation}
Thus, the trivialisation $\hat{f}^* \CT$ and the bundle morphisms $\psi_a$ render the map $R \circ (\hat{c} \times f) \circ \Delta_U$ a plot of $\scR_{ij}$.
\end{proof}

A path $\Gamma \colon [0,1] \to P_{ij}M$ is \emph{thin} if the adjoint map $\Gamma^\dashv \colon [0,1]^2 \to M$ satisfies $\rank(\Gamma^\dashv_{*|t}) < 2$ for every $t \in [0,1]^2$, where $\Gamma_{*}^{\dashv}$ denotes the differential of $\Gamma^{\dashv}$.
The equivariant structure $R$ on $\scR_{ij}$ can be induced from the connection on $\scR_{ij}$ in the following way.

\begin{proposition}
\label{st:Diff action via pt}
Let $F \colon [0,1] \to \Diff^+([0,1])$ be any smooth path with $F(0) = 1_{[0,1]}$.
For every $\gamma \in P_{ij}M$ this induces a thin smooth path
\begin{equation}
        R_F (\gamma) \colon [0,1] \to P_{ij}M\,,
        \qquad
        t \mapsto R_{F(t)} (\gamma)
\end{equation}
from $\gamma$ to $\gamma \circ F(1)$.
We have that
\begin{equation}
\label{eq:Diff action via pt}
        pt_{ij| R_F(\gamma)} = R \big( -, F(1) \big)
        \quad \text{in} \quad
        \Hom (\scR_{ij|\gamma},\, \scR_{ij|\gamma \circ F(1)})\,.
\end{equation}
\end{proposition}

\begin{proof}
Since the expressions~\eqref{eq:pt in R_ij} and~\eqref{eq:Diff action on R_ij} are well-defined on equivalence classes, it suffices to prove the identity~\eqref{eq:Diff action via pt} with respect to any one representative, i.e.~with respect to any one trivialisation of the pullback gerbe $(R_F \gamma)^{\dashv*}\CG$ over $[0,1]^2$.

First, we consider the map
\begin{equation}
        F^\dashv \colon [0,1]^2 \to [0,1]\,,
        \qquad
        (s,t) \mapsto F(s)(t)\,,
\end{equation}
which satisfies
\begin{equation}
        R_F \gamma = \gamma \circ F^\dashv\,.
\end{equation}
Let $\CT_0 \colon \gamma^*\CG \to \CI_0$ be a trivialisation.
We obtain a trivialisation $F^{\dashv *}\CT_0 \colon (R_F \gamma)^{\dashv *}\CG \to \CI_0$ over $[0,1]^2$, which has the following properties:
\begin{enumerate}[(1)]
        \item It has 2-form $\rho = 0$.
        
        \item For any $s \in [0,1]$, we have
        \begin{equation}
                (F^{\dashv *}\CT_0)_{|\{s\} \times [0,1]} = \big( F(s) \big)^*\CT_0\,.
        \end{equation}
        
        \item Since $F(s)(0) = 0$ and $F(s)(1) = 1$ for all $s \in [0,1]$, we have that $(F^{\dashv *}\CT_0)_{|[0,1] \times \{t\}}$ is the pullback of a trivialisation of a bundle gerbe over the point for $t \in \{0,1\}$.
\end{enumerate}
Combining property (3) with the fact that $R_F \gamma (s,t) = \gamma(t)$ for all $s \in [0,1]$ and $t \in \{0,1\}$, we obtain that in this special case the bundles $\CV_{(i,\CT^0)}$ and $\CV_{(j,\CT^1)}$ (cf.~\eqref{eq:aux bdls for pt_ij}) are pullbacks of bundles over the point.
Thus, their parallel transport is trivial.
Inserting this insight and properties (1) and (2) into the definition~\eqref{eq:pt in R_ij} of the parallel transport $pt_{ij}$ readily yields the identity~\eqref{eq:Diff action via pt}.
\end{proof}

\begin{corollary}
\label{st:R_ij equivar as bdl with connection}
The lifted action $R$ commutes with the parallel transport on the bundle $\scR_{ij}$.
Consequently, $\scR_{ij}$ is $\Diff^+([0,1])$-equivariant as a hermitean diffeological vector bundle with connection.
\end{corollary}

\begin{proof}
This corollary follows from combining Proposition~\ref{st:Diff action via pt} with the fact that $pt_{ij}$ is superficial (see~\cite[Definition~A.2.2, Proposition~4.3.2]{BW:Transgression_and_regression_of_D-branes}):
if $f \in \Diff^+([0,1])$ is any diffeomorphism, we find some smooth path in $\Diff^+([0,1])$ from $1_{[0,1]}$ to $f$ with sitting instants.
For example take $F \colon [0,1] \to \Diff^+([0,1])$, $F(s)(t) = (1-s)t + s\, f(t)$ -- this is smooth and strictly increasing for any fixed $s \in [0,1]$, and hence $F(s)$ is a diffeomorphism for any $s$.
If $\Gamma \colon [0,1] \to P_{ij}M$ is any smooth path, we obtain a rank-two homotopy $R_F \Gamma \colon [0,1] \to P(P_{ij}M)$ from $\Gamma$ to $R_f \Gamma$ (in the notation of the proof of Proposition~\ref{st:Diff action via pt}).
Since $pt_{ij}$ is superficial, property (ii) of~\cite[Definition~A.2.2]{BW:Transgression_and_regression_of_D-branes} and Proposition~\ref{st:Diff action via pt} imply that
\begin{equation}
        R_f \circ pt_{ij\, |\Gamma} = pt_{ij\, |(R_f \circ \Gamma)} \circ R_f\,,
\end{equation}
as claimed.
\end{proof}

\begin{remark}
Corollary~\ref{st:R_ij equivar as bdl with connection} can also be deduced directly from the explicit form~\eqref{eq:pt in R_ij} of the parallel transport on $\scR_{ij}$, using the fact that the induced diffeomorphism $\hat{F} \colon [0,1]^2 \to [0,1]^2$ is a \emph{fibre-wise} diffeomorphism.
This implies the invariance of the integral term in~\eqref{eq:pt in R_ij}.
\qen
\end{remark}

To conclude this section, we note that for every $i,j \in I$, the diffeomorphism $\rev \colon [0,1] \to [0,1]$, $t \mapsto 1-t$ induces an isomorphism of diffeological spaces
\begin{equation}
\label{def:alpha}
        R_\rev \colon P_{ij}M \to P_{ji}M\,,
        \qquad
        \gamma \mapsto \gamma \circ \rev\,.
\end{equation}
This isomorphism lifts to a bundle isomorphism~\cite[Section~4.7]{BW:Transgression_and_regression_of_D-branes}
\begin{equation}
\label{eq:orientation involution on scR}
        \alpha_{ij} \colon \scR_{ij} \to R_\rev^*\overline{\scR_{ji}}\,,
        \qquad
        [\CT,\psi] \mapsto [\rev^*\CT, \psi^*]\,.
\end{equation}
Here,  $\overline{\scR_{ij}}$ denotes the complex conjugate vector bundle, and $\psi^*$ is the fibre-wise hermitean adjoint of $\psi$.
Equivalently, we have a commutative square of diffeological spaces
\begin{equation}
\label{eq:equivar str on P_ijM and alpha_ij}
\xymatrix{        & \scR_{ij} \ar[dd] \ar[rr]^{\alpha_{ij}} & & \overline{\scR_{ji}} \ar[dd]       \\
        \Diff^+([0,1]) \times \scR_{ij} \ar[rr]_-<<<<<<<<<{a_\rev \times \alpha_{ij}} \ar[dd] \ar[ur]^{R}
        & & \Diff^+([0,1]) \times \overline{\scR_{ji}} \ar[dd] \ar[ur]^{\overline{R}} &        \\
        & P_{ij}M \ar[rr]^-<<<<<<<<<<<<<<<<<<<{R_\rev} & & P_{ji}M
        \\
        \Diff^+([0,1]) \times P_{ij}M \ar[rr]_{a_\rev \times R_\rev} \ar[ur]^{R}
        & & \Diff^+([0,1]) \times P_{ji}M \ar[ur]_{R}  &}
\end{equation}
where $a_\rev$ denotes conjugation by $\rev$ in $\Diff^+([0,1])$.
Finally, we point out that there is an isomorphism
\begin{equation}
        \beta_{ij} \colon R_\rev^*\scR_{ji} \to \scR_{ij}^\vee
\end{equation}
defined implicitly by
\begin{equation}
\label{eq:evaluation of beta_ij}
        \< \beta_{ij}[\rev^*\CT,\psi],\, [\CT,\phi]\> = \tr(\psi \circ \phi)\,,
\end{equation}
where $\<-,-\>$ denotes the evaluation pairing.
Alternatively, we can write
\begin{equation}
\label{eq:beta_ij}
        \beta_{ij} = R_\rev^*\flat_{h_{ji}} \circ (\overline{\alpha_{ij}})^{-1}\,,
\end{equation}
where $\flat_{h_{ji}} \colon \overline{\scR_{ji}} \to \scR_{ji}^\vee$ is the musical isomorphism induced by the hermitean metric $h_{ji}$ on $\scR_{ji}$.

\subsection{The transgression line bundle over the loop space}
\label{sect:LBdl_on_loop_space}

Here we recall the construction of the transgression line bundle of a bundle gerbe from~\cite{Waldorf--Trangression_II}.
Let $\CG \in \Grb^\nabla(M)$ be a bundle gerbe with connection over $M$.
We denote by $LM$ the diffeological free loop space of $M$.
There exists a principal $\sfU(1)$-bundle $L\CG$ over $LM$, whose fibre over a loop $\gamma$ is the set of isomorphism classes of trivialisations of $\gamma^*\CG$.
This is a $\sfU(1)$-torsor, since the groupoid of trivialisations of a bundle gerbe with connection is a torsor groupoid over the groupoid of hermitean line bundles with connection~\cite{Waldorf--More_morphisms}.
A line bundle $J \in \HLBdl^\nabla(\bbS^1)$ acts on the fibre $L\CG_{|\gamma}$ as
\begin{equation}
        [\CT] \mapsto [\CT \otimes J]\,,
\end{equation}
where $\CT$ is a trivialisation of $\gamma^*\CG$ and $[\CT]$ denotes its 2-isomorphism class.
It has been shown in~\cite{Waldorf--Trangression_II} that $L\CG$ is a diffeological $\sfU(1)$-bundle on $LM$.
Further, $L\CG$ carries a symmetrising fusion product and a superficial, fusive connection~\cite{Waldorf--Trangression_II} (though these are not relevant here).
The transgression line bundle $\scL$ is the associated diffeological hermitean line bundle
\begin{equation}
        \scL \coloneqq L\CG \times_{\sfU(1)} \FC\,.
\end{equation}
Thus, the elements of the fibre $\scL_{|\gamma}$ over a loop $\gamma \in LM$ are equivalence classes $[[\CT],z]$, consisting of an isomorphism class of a trivialization $\CT \colon \gamma^* \CG \to \CI_0$ and of a complex number $z \in \FC$.
The equivalence relation identifies representatives $(\CT,z) \sim (\CT',z')$ if there exists a hermitean line bundle $J$ on $\bbS^1$ with connection such that $\CT \otimes J \cong \CT'$ and $z = z' \cdot \hol(J)$.

Let $\Diff^+(\bbS^1)$ denote the diffeological group of orientation-preserving diffeomorphisms of $\bbS^1$.
This acts on $LM$ by pre-composition, i.e.~via the map
\begin{equation}
        R \colon LM \times \Diff^+(\bbS^1) \to LM\,,
        \qquad
        (\gamma, \tau) \mapsto \gamma \circ \tau\,.
\end{equation}
The map $R$ is smooth by arguments analogous to those in Lemma~\ref{st:Diff action on P_ij M}.
Further, it lifts to a map
\begin{equation}
\label{eq:Diff-action on L}
        R \colon \scL \times \Diff^+(\bbS^1) \to \scL\,,
        \qquad
        \big( \big[ [\CT], z \big], \tau \big) \mapsto \big[ [\tau^*\CT], z \big]\,.
\end{equation}
The latter map can be expressed in terms of parallel transport along paths induced by isotopies that connect $\tau$ and $1_{\bbS^1}$, in analogy with Proposition~\ref{st:Diff action via pt} (using the fact that $\Diff^+(\bbS^1)$ is connected).
However, note that there are non-homotopic isotopies of this kind, since $\bbS^1$ is not simply connected.
The map described above is well-defined nevertheless, due to the superficiality of the connection on $\scL$ (see~\cite[Definition~2.2.1(i), Corollary~4.3.3]{Waldorf--Trangression_II}).
The relation between the map $R$ from~\eqref{eq:Diff-action on L} and parallel transport was worked out in~\cite[Remark~4.3.7]{Waldorf--Trangression_II}.
In particular, the map $R$ is smooth by~\cite[Proposition~2.2.5]{Waldorf--Trangression_II}.

Finally, again in analogy to Section~\ref{sect:equivar strs on R_ij}, there is a $\RZ_2$-action on $LM$ via $R_\rev \colon \gamma \mapsto \gamma \circ \rev$, with $\rev \colon \bbS^1 \to \bbS^1$ denoting the orientation-reversing diffeomorphism $\exp(2\pi\, \iu\, t) \mapsto \exp(- 2\pi\, \iu\, t)$.
This descends to a $\RZ_2$-action $R_\rev$ on $LM$, which lifts to
\begin{equation}
\label{eq:lifted path reversal on loops}
        \tilde{\lambda} \colon \scL \to R_\rev^*\overline{\scL}
\end{equation}
respectively.
As for $\scR_{ij}$, this can equivalently be cast as an isomorphism
\begin{equation}
\label{eq:varrho mp}
        \tilde{\varrho} \coloneqq R_\rev^*\flat_{h_\scL} \circ \overline{\tilde{\lambda}}^{-1}
        \colon R_\rev^*\scL \to \scL^\vee\,.
\end{equation}

\subsection{The coherent pull-push construction}
\label{sect:pull-push and coherence}

In Section~\ref{sect:TFT_construction} it will be crucial to not just consider the mapping spaces $P_{ij}M$ of paths between D-branes that are parameterised over the unit interval $[0,1]$.
We would like to consider, for any oriented 1-manifold $Y$ diffeomorphic to $[0,1]$, the spaces $P^Y_{ij}M$ of smooth maps $\gamma \colon Y \to M$ with sitting instants around the initial point $y_0 \in Y$ and the endpoint $y_1 \in Y$ as defined by the orientation of $Y$, and such that $\gamma(y_0) \in Q_i$ and $\gamma(y_1) \in Q_j$.
The set $P^Y_{ij}M$ of such maps is a subset of the diffeological mapping space $M^Y$, and we endow it with the subspace diffeology.
The equivariant structures on the bundles $\scR_{ij} \to P_{ij}M$ from Section~\ref{sect:equivar strs on R_ij} allow us to transfer the bundles $\scR_{ij}$ to the spaces $P^Y_{ij}M$ in a coherent way.
Analogously, we need to transfer the hermitean line bundle $\scL \to LM$ to the diffeological mapping space $L^Y M$ of smooth maps $Y \to M$ for every oriented manifold $Y$ that is diffeomorphic to $\bbS^1$.
In this section we make precise what we mean by this and outline a construction that achieves this transfer.
The full details of this construction can be found in the separate article~\cite{Bunk:Coh_Desc}.

For $Y_0,Y_1$ oriented manifolds (possibly with boundary), we write $\rmD(Y_0,Y_1) \coloneqq \Diff^+(Y_0,Y_1)$ for the diffeological space of orientation-preserving diffeomorphisms from $Y_0$ to $Y_1$.
Note that if $Y_0 = Y_1 = Y$, the space $\rmD(Y,Y) \eqqcolon \rmD(Y)$ canonically has the structure of a diffeological group.
For $Y \cong [0,1]$, this group acts on the diffeological space $P_{ij}^Y M$ by pre-composition; there is a smooth map
\begin{equation}
        P_{ij}^Y M \times \rmD(Y) \rightarrow P_{ij}^Y M\,,
        \qquad
        (\gamma, f) \mapsto \gamma \circ f\,,
\end{equation}
which is compatible with the group structure on $\rmD(Y)$.
Similarly, for $Y_0, Y_1 \cong [0,1]$ there are smooth maps
\begin{equation}
        d_0 \colon P_{ij}^{Y_0} M \times \rmD(Y_1,Y_0) \rightarrow P_{ij}^{Y_1} M\,,
        \qquad
        (\gamma, f) \mapsto \gamma \circ f\,,
\end{equation}
which fit into a commutative square
\begin{equation}
\xymatrix{P_{ij}^{Y_0} M \times \rmD(Y_1,Y_0) \times \rmD(Y_2,Y_1) \ar[r] \ar[d]
        & P_{ij}^{Y_1} M \times \rmD(Y_2,Y_1) \ar[d]
        \\
        P_{ij}^{Y_0} M \times \rmD(Y_2,Y_0) \ar[r]
        & P_{ij}^{Y_2} M
}
\end{equation}
in $\Dfg$ for any oriented manifolds $Y_0,Y_1,Y_2 \cong [0,1]$.
We also define maps
\begin{equation}
        d_1 \colon P_{ij}^{Y_0} M \times \rmD(Y_1,Y_0) \rightarrow P_{ij}^{Y_0} M\,,
        \qquad
        (\gamma, f) \mapsto \gamma\,,
\end{equation}
and
\begin{alignat}{3}
        d_0 \colon P_{ij}^{Y_0} M \times \rmD(Y_1,Y_0) \times \rmD(Y_2,Y_1) &\rightarrow P_{ij}^{Y_1} M \times \rmD(Y_2,Y_1)\,,
        &&\qquad
        (\gamma, f_1,f_2) &&\mapsto (\gamma \circ f_1, f_2)\,,
        \\
        d_1 \colon P_{ij}^{Y_0} M \times \rmD(Y_1,Y_0) \times \rmD(Y_2,Y_1) &\rightarrow P_{ij}^{Y_0} M \times \rmD(Y_2,Y_0)\,,
        &&\qquad
        (\gamma, f_1,f_2) &&\mapsto (\gamma, f_1 \circ f_2)\,,
        \\
        d_2 \colon P_{ij}^{Y_0} M \times \rmD(Y_1,Y_0) \times \rmD(Y_2,Y_1) &\rightarrow P_{ij}^{Y_0} M \times \rmD(Y_1,Y_0)\,,
        &&\qquad
        (\gamma, f_1,f_2) &&\mapsto (\gamma, f_1)\,.
\end{alignat}
Note that this notation stems from the use of simplicial techniques, which are at work in the background here; for more details on this, see~\cite{Bunk:Coh_Desc}.
Let $\scM_{[0,1]}$ be the groupoid of oriented manifolds $Y$ that are isomorphic (as oriented manifolds) to $[0,1]$ with its standard orientation.
The morphisms in $\scM_{[0,1]}$ are the orientation-preserving diffeomorphisms.
Note that $Y \mapsto P^Y_{ij}M$ defines a functor $P^{(-)}_{ij}M \colon \scM_{[0,1]}^\opp \to \Dfg$.

\begin{definition}
\label{def:coherent R_ij and L}
Let $(M,Q)$ be a target space.
A \emph{coherent hermitean vector bundle on $P_{ij}^{(-)}M$} is a pair $(E,\mu)$ of a family $E=\{E_{Y_0}\}_{Y_0 \in \scM_{[0,1]}}$ of hermitean vector bundles $E_{Y_0} \to P_{ij}^{Y_0} M$, together with a family $\mu=\{\mu_{Y_1,Y_0}\}_{Y_1,Y_0 \in \scM_{[0,1]}}$ of isomorphisms
\begin{equation}
        \mu_{Y_1,Y_0} \colon d_1^*E_{Y_0} \rightarrow d_0^*E_{Y_1}
\end{equation}
of hermitian vector bundles over $P_{ij}^{Y_0} M \times \rmD(Y_1,Y_0)$, such that
\begin{equation}
        d_0^* \mu_{Y_2,Y_1} \circ d_2^* \mu_{Y_1,Y_0} = d_1^* \mu_{Y_2,Y_0}
\end{equation}
over $P_{ij}^{Y_2} M \times \rmD(Y_2,Y_1) \times \rmD(Y_1,Y_0)$ for every $Y_0,Y_1,Y_2 \in \scM_{[0,1]}$.
A \emph{morphism of coherent hermitean vector bundles on \smash{$P_{ij}^{(-)}M$}}, written $\psi \colon (E,\mu) \to (F,\nu)$, consists of a family $\psi=\{\psi_{Y_0}\}_{Y_0\in \scM_{[0,1]} }$  of hermitian vector bundle morphisms $\psi_{Y_0} \colon E_{Y_0} \to F_{Y_0}$ that intertwine the morphisms $\mu$ and $\nu$.
This defines the category $\HVBdl_{coh}(P_{ij}^{(-)}M)$ of coherent hermitean vector bundles on $P_{ij}^{(-)}M$.
\end{definition}

We introduce the simplicial diffeological space $\big( P_{ij} M \dslash \rmD([0,1]) \big)_\bullet \in \Dfg^{\Delta^\opp}$ by setting
\begin{align}
        \big( P_{ij} M \dslash \rmD([0,1]) \big)_n &\coloneqq P_{ij} M \times  \rmD([0,1])^n\,;
\end{align}
i.e., it is  the nerve of the action groupoid of the $\rmD([0,1])$-action on $P_{ij} M$.
The following definition spells out what a hermitian vector bundle over the simplicial diffeological space $(P_{ij} M \dslash \rmD([0,1]) )_\bullet$ is.
\begin{definition}
An \emph{equivariant hermitean vector bundle on $P_{ij}M$} is a pair $(E',\mu')$ of a hermitean vector bundle $E' \to P_{ij}M$, together with an isomorphism $\mu' \colon d_1^*E' \to d_0^*E'$ over $P_{ij} M \times \rmD([0,1])$, such that
\begin{equation}
        d_2^*\mu' \circ d_0^*\mu' = d_1^*\mu'
\end{equation}
over $P_{ij}M \times \rmD([0,1])^2$.
A \emph{morphism of equivariant hermitean vector bundles on $P_{ij}M$}, denoted $\psi' \colon (E',\mu') \to (F',\nu')$, consists of a morphism $\psi \colon E' \to F'$ that intertwines the morphisms $\mu'$ and $\nu'$.
This defines the category $\HVBdl(P_{ij}M)^{\rmD([0,1])}$ of equivariant hermitean vector bundles on $P_{ij}M$.
\end{definition}

Analogously, we define the category $\HVBdl(LM)^{\rmD(\bbS^1)}$ of \emph{equivariant hermitean vector bundles on $LM$} and the category $\HVBdl_{coh}(L^{(-)}M)$ of \emph{coherent hermitean vector bundles on $L^{(-)}M$}.
Now consider the diffeological space $P_{ij}^Y M \times \rmD([0,1],Y)$.
It comes with two smooth maps
\begin{alignat}{3}
        \Phi^Y_0 \colon P_{ij}^Y M \times \rmD([0,1],Y) &\rightarrow P_{ij}M\,,
        && \qquad
        (\gamma, g) &&\mapsto \gamma \circ g\,,
        \\
        \Psi^Y \colon P_{ij}^Y M \times \rmD([0,1],Y) &\rightarrow P_{ij}^Y M\,,
        && \qquad
        (\gamma, g) &&\mapsto \gamma\,.
\end{alignat}
For any $Y \in \scM_{[0,1]}$, the map $\Phi^Y_0$ extends to a morphism
\begin{equation}
        \Phi^Y_\bullet \colon P_{ij}^Y M \times \rmD([0,1],Y)^\bullet \rightarrow (P_{ij}M \dslash \rmD([0,1]))_\bullet
\end{equation}
of simplicial diffeological spaces.
Observe that the source of $\Phi^Y_\bullet$ is the \v{C}ech nerve of the subduction $\Psi^Y$.
Since it is simplicial, pullback along $\Phi^Y_\bullet$ induces a functor
\begin{equation}
        (\Phi_\bullet^{Y})^{*} \colon \HVBdl(P_{ij} M)^{\rmD([0,1])} \rightarrow \Desc(\HVBdl,\Psi^Y)
\end{equation}
for any $i,j \in I$, where $\Desc(\HVBdl,\Psi^Y)$ is the category of descent data for hermitean vector bundles with respect to the subduction $\Psi^Y$.
We can now use the fact that $\HVBdl$ admits a functorial descent $\Desc(\HVBdl,\Psi^Y) \to \HVBdl(P_{ij}^Y M)$ in order to obtain a hermitean vector bundle on $P_{ij}^Y M$ from any equivariant hermitean vector bundle on $P_{ij}M$.
The above constructions and the following theorem will be discussed in more general context in \cite{Bunk:Coh_Desc}. 
\begin{theorem}
\label{st:Psi_* Phi^* is equivalence}
Considering all $Y \in \scM_{[0,1]}$, the composition of pullback along $\Phi^Y_\bullet$ and descent along $\Psi^Y$ naturally assemble into an equivalence of categories
\begin{equation}
        \Psi_* \Phi^* \colon \HVBdl(P_{ij} M)^{\rmD([0,1])} \longrightarrow \HVBdl_{coh}(P_{ij}^{(-)}M)\,.
\end{equation}
A completely analogous construction yields an equivalence
\begin{equation}
        \Psi_* \Phi^* \colon \HVBdl(LM)^{\rmD(\bbS^1)} \longrightarrow \HVBdl_{coh}(L^{(-)}M)\,.
\end{equation}
\end{theorem}

Let $(\scR_{ij}, R) \in \HVBdl(P_{ij}M)^{\rmD([0,1])}$ and $(\scL, R) \in \HVBdl(LM)^{\rmD(\bbS^1)}$ be the equivariant hermitean vector bundles from Section~\ref{sect:bundles on path spaces} and Section~\ref{sect:LBdl_on_loop_space}, respectively.
Applying the coherent pull-push of Theorem \ref{st:Psi_* Phi^* is equivalence} we obtain coherent vector bundles
\begin{alignat}{3}
        (\widehat{\scR}_{ij}, \widehat{R}) &\coloneqq \Psi_*\Phi^*(\scR_{ij},R)
        &&\quad
        &&\in \HVBdl_{coh} \big( P^{(-)}_{ij}M \big) \qandq
        \\*
        (\widehat{\scL}, \widehat{R}) &\coloneqq \Psi_*\Phi^*(\scL,R)
        &&\quad
        &&\in \HLBdl_{coh} \big( L^{(-)}M \big)\,.
\end{alignat}
We will write $\scR^Y_{ij} \to P^Y_{ij}M$ and $\scL^Y \to L^YM$ for their components over $P^Y_{ij}M$ and over $L^Y M$, respectively.

\begin{remark}
Since the $\rmD([0,1])$-equivariant structure $R$ on the bundle $\scR_{ij}$ is compatible with the connection on $\scR_{ij}$, we could even construct the extended bundle $\scR^Y_{ij}$ as a coherent diffeological hermitean vector bundle with connection on \smash{$P^{(-)}_{ij}M$} (and similarly for $\scL$), but for our purposes we will only need the connection on \smash{$\scR_{ij}$}, as given in Section~\ref{sect:bundles on path spaces}.
\qen
\end{remark}

By construction, there are canonical identifications $\scR^{[0,1]}_{ij} \cong \scR_{ij}$ and \smash{$\scL^{\bbS^1} \cong \scL$}.
The map $\rev \colon [0,1] \to [0,1]$, $t \mapsto 1-t$ induces an isomorphism
\begin{align}
        R_\rev \times a_\rev^\bullet \colon P_{ij}M \dslash\rmD([0,1])_\bullet
        &\longrightarrow P_{ji}M \dslash \rmD([0,1])_\bullet\,,
        \\
        (\gamma, f_1, \ldots, f_n) &\longmapsto (\gamma \circ \rev, \rev ^{-1} f_1 \rev, \ldots, \rev^{-1} f_n \rev)
\end{align}
of simplicial diffeological spaces (compare also diagram~\eqref{eq:equivar str on P_ijM and alpha_ij}).
We can thus use $R_\rev \times a_\rev^\bullet$ to pull back the equivariant bundle $\overline{\scR_{ji}} \to P_{ji}M$ to an equivariant bundle $R_\rev^* \overline{\scR_{ji}} \to P_{ij}M$.
Then, the morphism $\alpha_{ij} \colon \scR_{ij} \to R_\rev^* \overline{\scR_{ji}}$ induces an isomorphism (which we also denote $\alpha_{ij}$) of $\rmD([0,1])$-equivariant hermitean vector bundles on $P_{ij}M$.
By the functoriality of $\Psi_* \Phi^*$ and the compatibility of descent of vector bundles with taking the complex conjugate vector bundle, we thus obtain an isomorphism
\begin{equation}
\label{eq:coherent alpha_ij}
        \Psi_* \Phi^*(\alpha_{ij}) \colon \Psi_* \Phi^*(\scR_{ij}) \longrightarrow \Psi_* \Phi^*(R_\rev^* \overline{\scR_{ji}})
        \cong \overline{ \Psi_* \Phi^*(R_\rev^* \scR_{ji}) }
\end{equation}
of coherent hermitean vector bundles on $P^{(-)}_{ij}M$.
We denote this isomorphism by $\widehat{\alpha}_{ij}$.

We can give yet a different perspective on this isomorphism:
observe that the oriented manifold $\overline{[0,1]}$, i.e.~the unit interval with the opposite orientation, is an element of $\scM_{[0,1]}$.
This is established by the orientation-\emph{preserving} diffeomorphism $\rev \colon [0,1] \to \overline{[0,1]}$.
Consequently, there is a commuting triangle of diffeological spaces
\begin{equation}
\label{eq:r_ij}
\xymatrix{  P^{[0,1]}_{ij}M \ar[rr]^{R_\rev} \ar[dr]_{r_{ij}} & & P^{[0,1]}_{ji}M \ar[dl]^-{{(P^{(-)}_{ji}M) (\rev)}}
        \\
        & P^{\overline{[0,1]}}_{ji}M &
}
\end{equation}
The morphism at the top pre-composes a path $\gamma$ by $\rev$, while still seeing $0 \in [0,1]$ as the initial point of the new path $\gamma \circ \rev$.
The right-hand morphism also pre-composes by $\rev$, but for the resulting path $\gamma \circ \rev \colon [0,1] \to M$ we view $1 \in [0,1]$ as its initial point.
Finally, the left-hand map just sends a map $\gamma$ to itself, but now views $1$ as the initial point in the parameterising manifold $[0,1]$.
Using the  coherent structure $\widehat{R}$ on $\widehat{\scR}_{ij}$ (see Definition~\ref{def:coherent R_ij and L}) and the morphism $\alpha_{ij}$, we obtain a commutative diagram
\begin{equation}
\label{eq:def alpha hat}
\xymatrix{
        \scR^{[0,1]}_{ij} \ar[rr]^{\alpha_{ij}} \ar[dr]_{\widehat{\alpha}^{[0,1]}_{ij}} & & \overline{\scR^{[0,1]}_{ji}} \ar[dl]^{R_{|(-,\rev)}}
        \\
        & \overline{\scR^{\overline{[0,1]}}_{ji}} &
}
\end{equation}
of isomorphisms of $\rmD([0,1])$-equivariant hermitean vector bundles that covers diagram~\eqref{eq:r_ij}, where the isomorphism $\widehat{\alpha}^{[0,1]}_{ij}$ is defined by this diagram.
Extending $\widehat{\alpha}^{[0,1]}_{ij}$ via $\Psi_* \Phi^*$ yields isomorphisms
\begin{equation}
\label{eq:def alpha widehat}
        \widehat{\alpha}_{ij} \colon \widehat{\scR}_{ij} \rightarrow r_{ij}^*\overline{\widehat{\scR}_{ji}}
\end{equation}
of coherent hermitean vector bundles on $P^{(-)}_{ij}M$ for all $i,j \in I$.

Analogously, the map $\rev \colon \bbS^1 \to \bbS^1$ as in Section~\ref{sect:LBdl_on_loop_space} and the isomorphism $\tilde{\lambda}$ from the same section in diagram~\eqref{eq:def alpha hat} yield an isomorphism
\begin{equation}
        \widehat{\lambda} \colon \widehat{\scL} \to r^* \overline{\widehat{\scL}}
\end{equation}
of coherent hermitean vector bundles over $L^{(-)}M$.
This also extends the isomorphisms $\beta_{ij}$ from~\eqref{eq:beta_ij} and $\tilde{\varrho}$ from~\ref{eq:varrho mp} to isomorphisms
\begin{equation}
\label{eq:whbeta whvarrho}
        \widehat{\beta} \colon r^*_{ij} \widehat{\scR}_{ji} \rightarrow \widehat{\scR}_{ij}^\vee
        \qandq
        \widehat{\varrho} \colon r^* \widehat{\scL} \rightarrow \widehat{\scL}^\vee\,.
\end{equation}

\section{Surface amplitudes}
\label{sect:amplitudes}

In this section we use the coherent bundles $\widehat{\scR}_{ij}$ and $\widehat{\scL}$ to extend the usual holonomy of bundle gerbes to amplitudes for surfaces with corners, whose boundary is partly contained in D-branes.
In Section~\ref{sect:TFT_construction} we assemble the resulting amplitudes into a smooth functorial field theory which describes the B-field-dependent part of open-closed bosonic string amplitudes.

\subsection{Scattering diagrams}
\label{sect:Mfds_with_corners}

First, we recall the geometric tools necessary to describe surfaces with corners.
Our main reference for this interlude is~\cite{SP--Thesis}.
An \emph{$m$-dimensional manifold with corners $N$} is a topological manifold with (possibly empty) boundary, equipped with a maximal smooth atlas whose charts are continuous maps
\begin{equation}
        \varphi \colon U \to \FR^m_+
\end{equation}
that are homeomorphisms onto their images, with $U \subset N$ open and $\FR_+$ denoting the set of non-negative real numbers.
The \emph{index} of a point $x \in N$ is the number of coordinates of $\varphi(x)$ that are zero. 
Compatible charts yield the same index; thus, each point $x \in N$ has a well-defined index $\ind(x) \in \{0,\ldots,m\}$. 
A \emph{connected face} of $N$ is the closure of a connected component of $\{x \in N\, |\, \ind(x) = 1 \}$, while a \emph{face} of $N$ is a disjoint union of connected faces.
A \emph{manifold with faces} is a manifold with corners such that each point $x \in N$ belongs to $\ind(x)$ different  faces.

\begin{definition}
\label{def:<n>-mfd}
An $m$-dimensional \emph{$\left \langle n \right \rangle$-manifold} is an $m$-dimensional manifold $N$ with faces together with a tuple $(\partial_0 N, \dots, \partial_{n-1} N)$ consisting of faces $\partial_i N$ of $N$ such that
\begin{enumerate}[(1)]
        \item 
        $\partial_0 N \cup \ldots \cup \partial_{n-1} N = \partial N$, where $\partial N \subset N$ is the subset of points of non-zero index, and
        
        \item
        for all $a \neq b \in \{0,\ldots,n\}$, the intersection $\partial_a N \cap \partial _b N$ is either empty or a face of $\partial_a N$ and of $\partial_b N$.
\end{enumerate}
A morphism of $\<n\>$-manifolds $N \to N'$ is a continuous map $f \colon N \to N'$ whose representatives in all charts are smooth, and such that $f_{|\partial_a N} \colon \partial_a N \to \partial_a N'$ for all $a = 0, \ldots, {n-1}$, i.e.\ $f$ is compatible with the partitions of $\partial N$ and $\partial N'$.
A \emph{$\<3\>^*$-manifold} is a $\<3\>$-manifold $N$ where every $x \in N$ with $\ind(x) \geq 2$ is contained in either $\partial_0 N \cap \partial_2 N$ or $\partial_1 N \cap \partial_2 N$.
Morphisms of $\<3\>^*$-manifolds are the same as those of $\<3\>$-manifolds.
\end{definition}

Note that, in particular, a $\<3\>^*$-manifold satisfies $\partial_0 N \cap \partial_1 N = \emptyset$.
If $N$ is an oriented manifold with corners, all connected faces of $N$ carry an induced orientation, which we define using an \emph{inward-pointing} normal vector field.
We will be concerned with compact (oriented) 2-dimensional $\< 3 \>^*$-manifolds $N$.
Each connected face $c \subset \partial N$ is a compact 1-dimensional manifold with boundary, and hence either diffeomorphic to $\bbS^1$ or to $[0,1]$.
Throughout this paper we call a 1-manifold with corners \emph{closed} if it is diffeomorphic to $\bbS^1$ and \emph{open} if it is diffeomorphic to $[0,1]$.
If $c_a$ is a connected face in $\partial_a N$ and $c_b$ is a connected face in $\partial_b N$ with $a \neq b$ then $c_a \cap c_b$ is either empty, one point, or two points.

Consider a target space $(M,Q)$ and a target space brane geometry $(\CG,\CE) \in \TBG(M,Q)$.
Let $\Sigma$ be an oriented, compact, 2-dimensional  $\<3\>^*$-manifold and let $\sigma \colon \Sigma \to M$ be a smooth map.
By Definition~\ref{def:<n>-mfd}, the boundary of $\Sigma$ comes with a partition $\partial \Sigma = \partial_0 \Sigma \cup \partial_1 \Sigma \cup \partial_2 \Sigma$.
We then think of $\partial_0 \Sigma$ as the \emph{incoming string boundary of $\Sigma$}, of $\partial_1 \Sigma$ as the \emph{outgoing string boundary of $\Sigma$}, and of $\partial_2 \Sigma$ as the \emph{brane boundary of $\Sigma$}.
In order to compute a surface amplitude of $\CG$ over $\Sigma$, we need the following decorations of $\Sigma$.

\begin{enumerate}[leftmargin=3.5em, label=(SD\arabic*)]

\item
\emph{Corners lie on D-branes:}
for every corner $x$ of $\Sigma$ (that is, a point $x \in \Sigma$ with $\ind(x)=2$), we choose a D-brane index $i(x) \in I$ such that $\sigma(x) \in Q_{i(x)}$.

\item
\emph{String endpoints move in D-branes:}
for each connected face $b \subset \partial_2 \Sigma$ in the brane boundary, we choose a D-brane index $i(b) \in I$ such that $\sigma(b) \subset Q_{i(b)}$, and satisfying $i(x) = i(b)$ for all corners $x \in b \subset \partial_2 \Sigma$.
Note that $\partial b$ may be empty.

\item
\emph{Incoming and outgoing states:}
per assumption on $\Sigma$, the boundaries
\begin{equation}
        \partial_0\Sigma = \bigsqcup_{u=1}^{n_0} c_{0,u}\ \sqcup\ \bigsqcup_{v=1}^{m_0} s_{0,v}
        \qquad
        \partial_1\Sigma = \bigsqcup_{u=1}^{n_1} c_{1,u}\ \sqcup\ \bigsqcup_{v=1}^{m_1} s_{1,v}
\end{equation}
are disjoint unions of connected faces $c_{0,u}, c_{1,u} \cong \bbS^1$ and $s_{0,v}, s_{1,v} \cong [0,1]$.
For a less cluttered notation, if $i,j \in I$ are the brane labels assigned to the initial and end points, respectively, of the oriented edge $s_{0,v}$, we just write $\scR^{s_{0,v}} \coloneqq \scR^{s_{0,v}}_{ij}$ for the vector bundle constructed in Section~\ref{sect:bundles on path spaces}; the brane labels are then understood from the data of $s_{0,v}$.
We choose an \quot{incoming state} vector
\begin{equation}
\label{eq:incoming state}
        \psi_0\ \in\ \bigotimes_{u=1}^{n_0} \scL^{c_{0,u}}_{|(\sigma_{|c_{0,u}})}\
        \otimes\ \bigotimes_{v=1}^{m_0} \scR^{s_{0,v}}_{|(\sigma_{|s_{0,v}})}
        \eqqcolon V_0(\Sigma,\sigma)\,.
\end{equation}
Note that if $\partial_0\Sigma = \emptyset$, we have $V_0(\Sigma, \sigma) = \FC$.
Similarly, we choose an \quot{outgoing state} vector 
\begin{equation}
\label{eq:outgoing state}
        \psi_1^\vee\ \in\ \bigotimes_{u=1}^{n_1} \scL^{c_{1,u}}_{|(\sigma_{|c_{1,u}})}\
        \otimes\ \bigotimes_{v=1}^{m_1} \scR^{s_{1,v}}_{|(\sigma_{|s_{1,v}})}
        \eqqcolon V_1(\Sigma,\sigma)^\vee\,,
\end{equation}
Note that if $\partial_1\Sigma = \emptyset$, we have $V_1(\Sigma, \sigma)^\vee = \FC$.
\end{enumerate}

\begin{remark}
The reason why we use the notation $V_1(\Sigma, \sigma)^\vee$ here will become evident in Section~\ref{sect:TFT_construction}, where the correct vector space to assign to the outgoing boundary of $(\Sigma, \sigma)$ is the dual of the vector space considered here; the identification between $V_1(\Sigma,\sigma)^\vee$ and $V_1(\Sigma, \sigma)$ will rely on the isomorphisms~\eqref{eq:whbeta whvarrho}.
\qen
\end{remark}

\begin{definition}
\label{def:parameterised_scattering_digram}
Let $(M,Q)$ be a target space, and let $(\CG, \CE) \in \TBG(M,Q)$ be a target space brane geometry.
\begin{enumerate}[(1)]
        \item A quadruple $(\Sigma,\sigma, \psi_1^\vee, \psi_0)$ of an oriented, compact, 2-dimensional $\<3\>^*$-manifold $\Sigma$ and a smooth map $\sigma \in \Mfd(\Sigma,M)$, endowed with auxiliary data as in (SD1)--(SD3) is called a~\emph{scattering diagram} for $(\CG,\CE)$.
        
        \item We call two scattering diagrams $(\Sigma, \sigma, \psi_1^\vee,\psi_0)$ and $(\Sigma', \sigma', \psi_1'{}^\vee,\psi'_0)$ \emph{equivalent} if there exists an orientation-preserving diffeomorphism $t \colon \Sigma \to \Sigma'$ of $\<3\>^*$-manifolds such that
        \begin{itemize}[leftmargin=1.5em,align=left,labelsep=-0.5em]
                \item $t$ preserves maps to $M$, i.e.~$\sigma' \circ t = \sigma$,
                
                \item $t$ preserves brane labels,
                
                \item the states $\psi_a$ and $\psi'_a$ agree under the isomorphism $V_a(\Sigma,\sigma) \cong V_a(\Sigma', \sigma')$ induced by evaluating the coherence isomorphisms $\widehat{R}$ of $\widehat{\scR}$ and $\widehat{\scL}$ on the restrictions of $t$ to the connected components of $\partial_a \Sigma$, for $a = 0,1$.
        \end{itemize}
        We denote the equivalence class of a scattering diagram $(\Sigma, \sigma, \psi_1^\vee,\psi_0)$ under this equivalence relation by $[\Sigma, \sigma, \psi_1^\vee,\psi_0]$.
\end{enumerate}
\end{definition}

\subsection{Definition of the surface amplitude}
\label{sect:amplitudes-definition}

In this subsection we define the surface amplitude for scattering diagrams $(\Sigma, \sigma, \psi_1^\vee, \psi_0)$ (Definition~\ref{def:parameterised_scattering_digram}) and then show that it depends only on the equivalence class $[(\Sigma, \sigma, \psi_1^\vee, \psi_0)]$.
We first consider the case where the vectors $\psi_0$ and $\psi_1^\vee$ in the tensor product vector spaces~\eqref{eq:incoming state} and~\eqref{eq:outgoing state} are pure vectors and use the following \emph{auxiliary data}, on which the amplitude will not depend:
\begin{itemize}
\item \emph{String boundary parameterisations:}
for each connected face $s \subset \partial _0 \Sigma$ or $s \subset \partial_1 \Sigma$ we fix an orientation-preserving diffeomorphism $\gamma_s \colon \bbS^1 \to s$ if $s \cong \bbS^1$, or $\gamma_s \colon [0,1] \to s$ if $s \cong [0,1]$.

\item \emph{Trivialisation:}
we fix a trivialisation $\CT \colon \sigma^* \CG \to \CI_\rho$.
\end{itemize}
The  pure incoming state vector $\psi_0$ is then represented (under the coherence isomorphism $\widehat{R}$) by a tensor product
\begin{equation}
\label{eq:reps for pure in vectors}
        \psi_0 = \bigotimes_{u=1}^{n_0} \psi_{0,u}
        \ \otimes \
        \bigotimes_{v=1}^{m_0} \psi_{0,v}
        = \bigotimes_{u=1}^{n_0} [[\gamma_{c_{0,u}}^*\CT], z_{c_{0,u}}]
        \ \otimes \
        \bigotimes_{v=1}^{m_0} [\gamma_{s_{0,v}}^*\CT, \psi_{s_{0,v}}]\,.
\end{equation}
Here, $n_0$ is the number of connected components $c_{0,u}$ of $\partial_0 \Sigma$ that are diffeomorphic to $\bbS^1$, and $m_0$ is the number of connected components $s_{0,v}$ of $\partial_0 \Sigma$ that are diffeomorphic to $[0,1]$.
Further, we have used the explicit form of the bundles $\scR_{ij}$ and $\scL$ from Section~\ref{sect:bundles on path spaces} and~\ref{sect:LBdl_on_loop_space}.
Observe that once the parameterisations $\gamma_s$ and the trivialisation $\CT$ have been fixed, the tensor factors $\psi_{0,u}$ and $\psi_{0,v}$ each have unique representatives as in the above formula by the fact that $(\widehat{\scR}_{ij}, \widehat{R})$ and $(\widehat{\scL},\widehat{R})$ are coherent.
An analogous statement holds true for the outgoing state vector  $\psi_1^\vee$;  can be written as a linear combination of tensor products of pure states (compare~\eqref{eq:outgoing state})
\begin{equation}
\label{eq:reps for pure dual out vectors}
        \psi_1^\vee = \bigotimes_{u=1}^{n_1} \psi_{1,u}^\vee
        \ \otimes \
        \bigotimes_{v=1}^{m_0} \psi_{1,v}^\vee
        = \bigotimes_{u=1}^{n_1} [[\gamma_{c_{1,u}}^*\CT], z_{c_{1,u}}]
        \ \otimes \
        \bigotimes_{v=1}^{m_1} [\gamma_{s_{1,v}}^*\CT, \psi_{s_{1,v}}]\,.
\end{equation}

We now proceed to define the surface amplitude of a scattering diagram whose states $\psi_0$ and $\psi_1^\vee$ are pure vectors in the tensor products~\eqref{eq:incoming state} and~\eqref{eq:outgoing state}.
The fully general scattering amplitude is then defined via multi-linear extension.
For any $x \in \Sigma$, let $\iota_x \colon \pt \hookrightarrow \Sigma$ denote the inclusion of the point at $x$.
Similarly, any connected face $c \subset \partial \Sigma$ of $\Sigma$ can be viewed as a submanifold of $\Sigma$ with inclusion $\iota_c \colon c \hookrightarrow \Sigma$.
For each corner $x$ of $\Sigma$, we set $E_x \coloneqq \Delta (\CE_{i(x)|\sigma(x)}, \CT_{|x})$, which is a hermitean vector bundle over a point and hence a finite-dimensional Hilbert space.
For each connected face $b \subset \partial_2 \Sigma$ let $E_b \coloneqq \Delta (\sigma_{|b}^*\CE_{i(b)}, \CT_{|b})$, which is a hermitean vector bundle with connection over $b$.

For each connected component $c$ of $\partial \Sigma$ we produce a number $z_c\in \FC$ in the following way.
If $c$ is diffeomorphic to $\bbS^1$, it is already a connected face in $\partial_0 \Sigma$, $\partial_1 \Sigma$, or $\partial_2 \Sigma$.
Otherwise, $c$ is a union of connected faces, each of which is diffeomorphic to $[0,1]$.
Thus, we have to treat the following cases:

\begin{enumerate}[leftmargin=1.3cm, label=(SA\arabic*)]
        \item
        $c \cong \bbS^1$ and $c$ is some connected face $c \subset \partial_0 \Sigma$ or $c \subset \partial_1 \Sigma$:
        in this case, we have $c = c_{\epsilon,u}$ for some $\epsilon \in \{0,1\}$ and $u \in \{1, \ldots, n_\epsilon\}$ (in the notation of~\eqref{eq:reps for pure in vectors} and \eqref{eq:reps for pure dual out vectors}).
        Further, we have chosen a parameterisation $\gamma_{c_{\epsilon,u}}$ of $c$ (for $\epsilon = 0$ or $1$) and an element $[[\gamma_{c_{\epsilon,u}}^*\CT], z_{c_{\epsilon,u}}] \in \scL_{|\gamma_{c_{\epsilon,u}}}$.
        We set
        \begin{equation}
                z_c \coloneqq z_{c_{\epsilon,u}}\,.
        \end{equation}
        
        \item
        $c \cong \bbS^1$ and $c = b$ for some connected face $b \subset \partial_2 \Sigma$:
        set
        \begin{equation}
                z_c \coloneqq \tr \big( \hol(E_b) \big)
        \end{equation}
        This is well-defined because the trace of the holonomy of a vector bundle is independent of the choices of a base point and of a parameterisation.
        
        \item
        $c$ is a union of connected faces $c_1,\ldots,c_n \subset \partial \Sigma$:
        by definition of a $\<3\>^*$-manifold we can order these faces in such a way that there are corners $x_0,\ldots,x_n \in c$ with $x_n=x_0$ and $\partial c_a = \{x_{a-1}, x_a\}$, where $x_{a-1}$ is the initial point and $x_a$ is the end point of $c_a$.
        We define the following linear maps $\lambda_a \colon E_{x_{a-1}} \to E_{x_a}$:
        \begin{enumerate}
                \item 
                If $c_a \subset \partial_2 \Sigma$ is brane boundary, then the hermitean vector bundle $E_{c_a}$ over $c_a$ comes with a connection, whose parallel transport yields an isomorphism
                \begin{equation}
                        \lambda_a \coloneqq pt^{E_{c_a}}_{c_a} \colon E_{x_{a-1}} = E_{c_a|x_{a-1}} \arisom E_{c_a|x_a} = E_{x_a}\,.
                \end{equation}
                
                \item
                If $c_a \subset \partial_0 \Sigma$ or $c_a \subset \partial_1 \Sigma$, then $c_a = s_{\epsilon,v}$ for some $\epsilon = 0,1$ and $v \in \{1, \ldots, m_\epsilon\}$ (in the notation of~\eqref{eq:reps for pure in vectors} and \eqref{eq:reps for pure dual out vectors}).
                It thus comes with the element $\psi_{s_{\epsilon,u}}$ chosen as part of the incoming and outgoing states (compare~\eqref{eq:reps for pure in vectors} and~\eqref{eq:reps for pure dual out vectors}).
                Under the canonical isomorphism
                \begin{equation}
                        \scR_{s_{\epsilon,u} |\sigma \circ \gamma_{s_{\epsilon,u}}}
                        \cong \scR_{s_{\epsilon,u} |\sigma \circ \gamma_{s_{\epsilon,u}}}(\CT)
                        \cong \Hom (E_{x_{a-1}}, E_{x_{a}})\,,
                \end{equation}
                we see that $\lambda_a \coloneqq \psi_{s_{\epsilon,u}}$ defines a morphism of vector spaces $\lambda_a \colon E_{x_{a-1}} \to E_{x_a}$.
        \end{enumerate}
        We thus obtain a linear map
        \begin{equation}
                \xymatrix{
                E_{x_0} \ar@{->}[r]^-{\lambda_1} & E_{x_1} \ar@{->}[r]^-{\lambda_2} & \ldots \ar@{->}[r]^-{\lambda_n} & E_{x_n} = E_{x_0}
                }\,,
        \end{equation}
        and we define
        \begin{equation}
                z_c \coloneqq \tr ( \lambda_n \circ \ldots \circ \lambda_1)\,.
        \end{equation}
        By the cyclicity of trace, this expression is invariant under cyclic permutations of the labels $c_1, \ldots, c_n$ of the connected faces, compatible with the orientation on $\partial \Sigma$.
\end{enumerate}

\begin{definition}
\label{def:parameterised_amplitude}
\label{eq:surface_amplitude}
Let $(\CG,\CE) \in \TBG(M,Q)$ be a target space brane geometry, and let $(\Sigma,\sigma, \psi_1^\vee, \psi_0)$ be a scattering diagram for $(\CG,\CE)$.
The \emph{surface amplitude} of $(\Sigma,\sigma, \psi_1^\vee, \psi_0)$ is defined by
\begin{equation}
        \CA^{\CG,\CE}(\Sigma, \sigma, \psi_1^\vee, \psi_0) \coloneqq \exp \bigg( - \int_\Sigma \rho \bigg) \prod_{c\, \in \pi_0(\partial \Sigma)} z_c
        \quad \in \FC\,.
\end{equation}
\end{definition}

\begin{example}
If $\partial \Sigma = \emptyset$, the surface amplitude $\CA^{\CG,\CE}(\Sigma, \sigma, 1,1)$ recovers the surface holonomy of $\CG$ over $(\Sigma,\sigma)$ and thus the closed WZW amplitude (see e.g.~\cite{Murray:Bundle_gerbs}).
In the case of only closed brane boundary  it coincides with the \quot{holonomy on D-branes} of~\cite{CJM--Holonomy_on_D-branes}.
\qen
\end{example}

In the following two lemmata we prove that the surface amplitude $\CA^{\CG,\CE}(\Sigma, \sigma, \psi_1^\vee, \psi_0)$ is defined independently of the choice of auxiliary data.

\begin{lemma}
\label{st:sh_independent_of_trivialisations}
The surface amplitude $\CA^{\CG,\CE}(\Sigma, \sigma, \psi_1^\vee, \psi_0)$ is independent of the trivialisation $\CT$ of $\sigma^*\CG$.
\end{lemma}

\begin{proof}
Suppose \smash{$\CT' \colon \sigma^* \CG \to \CI_{\rho'}$} is another trivialisation.
Define a hermitean line bundle with connection $J \coloneqq \Delta (\CT',\CT)$ over $\Sigma$.
Its curvature satisfies $\curv(J) = \rho' - \rho$.
All quantities used above will be written with a prime when constructed from $\CT'$ instead of $\CT$.
The isomorphisms~\eqref{eq:structural isos for Delta} yield connection-preserving isomorphisms
\begin{equation}
        U_x \colon J^\vee_{|x} \otimes E_x \arisom E_x'\,, \quad
        U_b \colon J^\vee_{|b} \otimes E_b  \arisom E_b'
\end{equation}
for every connected face $b \subset \partial_2 \Sigma$ and for every corner $x$ of $\Sigma$.
Consequently, for $c \subset \partial_2 \Sigma$ a closed 1-manifold, i.e.\ in case (SA2), we obtain that $z'_c = \hol(J,c)^{-1} \cdot z_c$.

For the other cases, we have to represent the states $\psi_0$ and $\psi_1^\vee$ with respect to the new trivialisation $\CT'$ (cf.~(SD3)).
Using the parameterisations $\gamma$, chosen as part of the auxiliary data, and the coherent structure $\widehat{R}$ of $\widehat{\scR}$ and $\widehat{\scL}$, as well as the respective formulae for changes of trivialisations of $\CG$, we obtain
\begin{alignat}{2}
        \big[ [\gamma_c^*\CT], z_c \big] &= \big[ [\gamma_c^*\CT'],\, \hol( J, s)^{-1} \cdot z_c \big]\ \in \scL_{|\sigma \circ \gamma_c}\,,
        & \qquad &\text{for} \quad c \cong \bbS^1\,,
        \\
        \big[ \gamma_c^*\CT, \psi_c \big] &= \big[ \gamma_c^*\CT',\, r_{\CT'_{|s}, \CT_{|s}} (\psi_c) \big]\ \in \scR_{s|\sigma \circ \gamma_c}\,,
        &\qquad &\text{for} \quad c \cong [0,1]\,.
\end{alignat}
Thus, we infer that $z'_c = \hol( J, c)^{-1} \cdot z_c$ for closed 1-manifolds $c \subset \partial_0 \Sigma$ or $c \subset \partial_1 \Sigma$.
From our construction~\eqref{eq:def of r_(CT,CT')} of the isomorphisms $r_{\CT',\CT}$ we obtain the following commutative diagram:
\begin{equation}
\xymatrix@C=3cm{  E'_{x_{a-1}} \ar[r]^{r_{\CT',\CT} (\psi_c)} \ar[d]_{U_{x_{a-1}}} & E'_{x_a} \ar[d]^{U_{x_a}}
        \\
        J^\vee_{|x_{a-1}} \otimes E_{x_{a-1}} \ar[r]_{pt^{J^\vee}_{c_a} \otimes\ \psi_c} & J^\vee_{|x_a} \otimes E_{x_a}
}
\end{equation}
This holds because the 2-isomorphisms $\psi$ used to construct $r_{\CT',\CT}$ are connection-preserving.
Therefore, with the above conventions for $\lambda'_a$, we find that
\begin{equation}
\begin{aligned}
        z'_c &= \tr (\lambda'_n \circ \ldots \circ \lambda'_1)
        \\
        &= \tr \big( U_{x_n} \circ (pt^{J^\vee}_{c_n} \otimes \lambda_n) \circ U_{x_{n-1}}^{-1} \circ \ldots \circ U_{x_1} \circ (pt^{J^\vee}_{c_1} \otimes \lambda_1) \circ U_{x_0}^{-1} \big)
        \\
        &= \tr \big( (pt^{J^\vee}_{c_n} \otimes \lambda_n) \circ \ldots \circ (pt^{J^\vee}_{c_1} \otimes \lambda_1) \big)
        \\
        &= \hol(J,c)^{-1} \cdot z_c\,.
\end{aligned}
\end{equation}
Putting everything together while keeping track of orientations, we obtain
\begin{align*}
        \exp \bigg( \int_\Sigma - \rho' \bigg) \prod_{c\, \in \pi_0(\partial \Sigma)} \Big( \hol(J,c)^{-1} \cdot z_c \Big)
        &= \exp \bigg( \int_\Sigma - \rho' \bigg)\ \hol(J,\partial \Sigma)^{-1}\, \prod_{c\, \in \pi_0(\partial \Sigma)} z_c
        \\*
        &= \exp \bigg( \int_\Sigma - \rho' \bigg)\ \exp \bigg( \int_\Sigma \curv(J) \bigg) \prod_{c\, \in \pi_0(\partial \Sigma)} z_c
        \\*
        &= \exp \bigg( \int_\Sigma - \rho \bigg) \prod_{c\, \in \pi_0(\partial \Sigma)} z_c\,,
\end{align*}
as stated.
The integral of $\curv(J)$ comes with a positive sign because Stokes' Theorem holds true for the orientation on the boundary induced by an \emph{outward-pointing} vector field, but the holonomy is taken around the boundary with the opposite orientation.
\end{proof}

\begin{lemma}
\label{st:sh_reparameterisation_invariance}
The surface amplitude is invariant under orientation-preserving reparameterisation of the string boundary components:
\begin{enumerate}[(1)]
\item
If $s \subset \partial_0 \Sigma$ or $s \subset \partial_1 \Sigma$ is a closed connected face and $\tau \in \Diff^+(\bbS^1)$, then the surface amplitude is invariant under changing $\gamma_s$ to $\gamma_s \circ \tau$.
        
\item 
If $s \subset \partial_0 \Sigma$ or $s \subset \partial_1 \Sigma$ is an open connected face and $\tau \in \Diff^+([0,1])$, then the surface amplitude is invariant under changing $\gamma_s$ to $\gamma_s \circ \tau$.
\end{enumerate}
\end{lemma}

\begin{proof}
In case (1), the change of parameterisation leads to a new representation of the state \smash{$\psi_0 \in \scL^s_{|\sigma_{|s}}$} as an element of $\scL$, using the coherence of $\widehat{\scL}$, according to
\begin{equation}
        \scL_{|\sigma \circ \gamma_s} \ni \big[ [\gamma_s^*\CT], z_s \big]
        \ \longmapsto \ R_\tau \big[ [\gamma_s^*\CT], z_s \big]
        = \big[ [ (\gamma_s \circ \tau)^*\CT], z_s \big] \in \scL_{|\sigma \circ \gamma_s \circ \tau}\,,
\end{equation}
and similarly for $\psi_1^\vee$.
Analogously, in case (2) the image of the state \smash{$\psi_0 \in \scR^s_{|\sigma_{|s}}$} under $\widehat{R}$ changes as
\begin{equation}
        \scR_{s|\sigma \circ \gamma_s} \ni [\gamma_s^*\CT, \psi_s]
        \ \longmapsto \ R_\tau \big( [\gamma_s^*\CT, \psi_s] \big)
        = \big[ (\gamma_s \circ \tau)^*\CT, \psi_s \big]
        \in \scR_{s|\sigma \circ \gamma_s \circ \tau}\,,
\end{equation}
and accordingly for $\psi_1^\vee$.
Both claims then follow from the fact that the action of reparameterisations leaves the representing elements $z_s$ and $\psi_s$ unchanged -- see~\eqref{eq:Diff-action on L} and~\eqref{eq:Diff action on R_ij}.
\end{proof}

It remains to prove that the surface amplitude is well-defined on equivalence classes of scattering diagrams.

\begin{lemma}
\label{st:shol_is_diffeomorphism-invariant}
The surface amplitude is invariant under diffeomorphisms of scattering diagrams as in Definition~\ref{def:parameterised_scattering_digram}(2).
\end{lemma}

\begin{proof}
Let $t \colon \Sigma \to \Sigma'$ be a diffeomorphism as in Definition~\ref{def:parameterised_scattering_digram}(2).
Note that $\sigma' \circ \gamma'_{t(c)} =\allowbreak \sigma' \circ t \circ \gamma_c =\allowbreak \sigma \circ \gamma_c$.
Thus, we have
\begin{equation}
        \scL_{|\sigma' \circ \gamma'_{t(s)}} = \scL_{|\sigma \circ \tau_s}
        \qandq
        \scR_{t(s)|\sigma' \circ \gamma_{t(s)}} = \scR_{s|\sigma \circ \gamma_s}\,.
\end{equation}
The statement then follows after choosing a trivialisation $\CT \colon \sigma^*\CG \to \CI_\rho$ to compute $\CA^{\CG,\CE}(\Sigma,\sigma, \psi_1^\vee, \psi_0)$ and using the pullback $t^*\CT$ to trivialise $\sigma'^*\CG$ and to compute $\CA^{\CG,\CE}(\Sigma',\sigma', \psi_1^\vee, \psi_0)$.
\end{proof}

Combining Lemmas~\ref{st:sh_independent_of_trivialisations}, \ref{st:sh_reparameterisation_invariance}, and~\ref{st:shol_is_diffeomorphism-invariant}, we obtain:

\begin{proposition}
\label{prop:surfaceamplitude}
The surface amplitude is well-defined and depends only on the equivalence class of a scattering diagram.
That is, we can define
\begin{equation}
        \CA^{\CG,\CE} [\Sigma, \sigma, \psi_1^\vee, \psi_0] \coloneqq \CA^{\CG,\CE} (\Sigma, \sigma, \psi_1^\vee, \psi_0)\,.
\end{equation}
\end{proposition}

\subsection{Properties of the surface amplitude}
\label{sect:amplitudes-properties}

In this section we investigate further the properties of the surface amplitude of Definition \ref{eq:surface_amplitude}.
For a smooth map $h \in \Mfd(N,M)$, denote its differential by $h_* \colon TN \to TM$.

\begin{definition}
\label{def:thin homotopy}
Let $M$ be a manifold, $N$ a manifold with faces, and let $f_0, f_1 \colon N \to M$ be smooth maps.
A \emph{thin homotopy} from $f_0$ to $f_1$ is a smooth homotopy $h \colon [0,1] \times N \to M$ between $f_0$ and $f_1$ such that
\begin{equation}
        \rank(h_{*|(t,x)}) \leq \dim(N) - \ind(x)
\end{equation}
for all $t \in [0,1]$ and $x \in N$.
In the case of $f_{0|\partial N} = f_{1|\partial N}$, a \emph{(thin) homotopy rel boundary} is a (thin) homotopy such that $h(t,x) = f_0(x)$ for all $t \in [0,1]$ and $x \in \partial N$.
\end{definition}

Let $(\Sigma,\sigma_0, \psi_1^\vee, \psi_0)$ and $(\Sigma,\sigma_1, \psi_1^\vee, \psi_0)$ be  scattering diagrams for a target space brane geometry $(\CG,\CE)$, with the same underlying  $\<3\>^*$-manifold $\Sigma$, such that $\sigma_0$ and $\sigma_1$ agree on $\partial\Sigma$. A \emph{(thin) homotopy of scattering diagrams} between  $(\Sigma,\sigma_0, \psi_1^\vee, \psi_0)$ and $(\Sigma,\sigma_1, \psi_1^\vee, \psi_0)$  is a (thin) homotopy $h \colon [0,1] {\times} \Sigma \to M$ rel boundary from $\sigma_0$ to $\sigma_1$ that preserves the brane labels picked in  (SD1) and (SD2).

\begin{proposition}
\label{st:amplitude and homotopies}
If $h$ is a homotopy between scattering diagrams $(\Sigma,\sigma_0, \psi_1^\vee, \psi_0)$ and $(\Sigma,\sigma_1, \psi_1^\vee, \psi_0)$, the corresponding surface amplitudes satisfy
\begin{equation}
\label{eq:amplitude and homotopies}
        \CA^{\CG,\CE} [\Sigma, \sigma_1, \psi_1^\vee, \psi_0]
        = \exp \bigg( \int_{[0,1] \times \Sigma} - h^*\curv(\CG) \bigg)\
        \CA^{\CG,\CE} [\Sigma, \sigma_0, \psi_1^\vee, \psi_0]\,.
\end{equation}
In particular, if $h$ is a thin homotopy, the surface amplitudes coincide,
\begin{equation}
\CA^{\CG,\CE} [\Sigma, \sigma_1, \psi_1^\vee, \psi_0]= \CA^{\CG,\CE} [\Sigma, \sigma_0, \psi_1^\vee, \psi_0]\text{.}
\end{equation} 
\end{proposition}

Proposition~\ref{st:amplitude and homotopies} will be a consequence of the following general lemma.

\begin{lemma}
\label{st:descent and pullback of field strengths}
Let $\CG \in \Grb^\nabla(M)$ be a bundle gerbe with connection, defined over a surjective submersion $\pi \colon Y \to M$ and with curving $B \in \Omega^2(Y)$.
\begin{enumerate}[(1)]
\item Consider a morphism $\CE \colon \CG \to \CI_\omega$, with an hermitean vector bundle $E$ with connection over a surjective submersion $\zeta \colon Z \to Y$.
The 2-form $\curv(E) + B \in \Omega^2(Z, \End(E))$ descends to a 2-form
\begin{equation}
        \Desc \big( \curv(E) + B \big)
        \ \in \Omega^2 \big( M, \Delta(\CE,\CE) \big)\,.
\end{equation}

\item Let $X$ be a manifold with corners, and let $f \colon X \to M$ be a smooth map such that there exists a trivialisation $\CT \colon f^*\CG \to \CI_\rho$ for some $\rho \in \Omega^2(X)$.
We have the identity
\begin{equation}
\label{eq:field strengths and thin homotopies}
        \curv \big( \Delta(f^*\CE,\CT) \big)
        = f^* \Desc \big( \curv(E) + B \big) - \rho\,.
\end{equation}
\end{enumerate}
\end{lemma}

\begin{proof}
Ad (1):
Recall that if $B \in \Omega^2(Y)$ is the curving 2-form of $\CG$ and if $L \to Y^{[2]}$ its defining hermitean line bundle with connection, then we have $\curv(L) = d_0^*B - d_1^*B$ over $Y^{[2]}$. We let $\alpha$ be the vector bundle morphism over $Z^{[2]}$ which is part of $\CE$.
Since it is connection-preserving, the curvature of $E$ and the curving 2-forms satisfy the identity
\begin{equation}
\label{eq:descent eq. for twisted field strength}
        \alpha \circ d_1^*\big( \curv(E) - (\omega - B) \big) \circ \alpha^{-1}
        = d_0^*\big( \curv(E) - (\omega - B) \big)
\end{equation}
as 2-forms over $Z^{[2]} = Z \times_M Z$.
Consequently, the combination $\curv(E) + B$ descends to a bundle-valued 2-form
\begin{equation}
        \Desc \big( \curv(E) + B) \big)
        \
        \in \Omega^2 \big( M, \Delta(\CE,\CE) \big)\,.
\end{equation}

Ad (2):
Let $f \colon X \to M$ be a smooth map such that there exists a trivialisation $\CT \colon f^*\CG \to \CI_\rho$ for some $\rho \in \Omega^2(X)$.
The map $f$ pulls back the sequence of surjective submersions $Z \to Y \to M$ to a sequence $\hat{f}_Y^*Z \to f^*Y \to X$ of surjective submersions.
Here $\hat{f}_Y \colon f^*Y \to Y$ is the map over $f$ induced by the pullback of $X \to M \leftarrow Y$.
Let $\hat{f}_Z \colon f^*Z \to Z$ be the analogous induced map over $\hat{f}_Y$.
Finally, let $\xi_\CT \colon Z_\CT \to f^*Y$ be the surjective submersion in the data of the trivialisation $\CT$.

We compute the curvature of the bundle $\Delta(f^*\CE,\CT)$:
the morphism $\CE \circ \CT^*$ is defined over the surjective submersion $\xi \colon Z_\CT \times_{f^*Y} f^*Z \to f^*Y$, and if $T$ is the bundle defining $\CT$, the bundle underlying $\CE \circ \CT^*$ is given by $\hat{f}_Z^*E \otimes T^\vee$.
Omitting other pullbacks, we have
\begin{align}
        \curv(\hat{f}_Z^*E \otimes T^\vee)
        &= \hat{f}_Z^* \curv(E) + \hat{f}_Z^*B - \xi^*\rho
        = \hat{f}_Z^* \big( \curv(E) + B \big) - \xi^*\rho\,.
\end{align}
Since the functor $\sfR$ (see Example~\ref{eg:sfR}) is built from the descent functor for hermitean vector bundles with connection~\cite{Waldorf--More_morphisms,Waldorf--Thesis}, this yields
\begin{align}
        \curv \big( \Delta(f^*\CE, \CT) \big)
        &=\curv \big( \sfR(\CE \circ \CT^*) \big)
        \\
        &= \Desc \big( \curv(\hat{f}_Z^*E \otimes T^\vee) \big)
        \\
        &= \Desc \big( \hat{f}_Z^*\curv(E) + B \big) - \rho
        \\
        &= f^* \Desc \big( \curv(E) + B \big) - \rho\,,
\end{align}
as claimed.
\end{proof}

\begin{proof}[Proof of Proposition~\ref{st:amplitude and homotopies}]
First, we choose orientation-preserving parameterisations $\gamma$ of the connected components of the string boundary of $\Sigma$ as in Section~\ref{sect:amplitudes-definition}.
Consider the pullback bundle gerbe $h^* \CG$ over $[0,1] \times \Sigma$.
Since $\rmH^3([0,1] \times \Sigma, \RZ) \cong \rmH^3(\Sigma, \RZ) = 0$, we can choose a trivialization $\CT \colon h^* \CG \to \CI_{\rho}$.
We have
\begin{equation}
        \dd \rho =\curv(h^*\CG) = h^* \curv(\CG)\,.
\end{equation}
Under the coherence isomorphism $\widehat{R}$, a pure vector $\psi_0$ then corresponds to a tensor product
\begin{align}
\label{eq:change of classes for thho invar}
        \bigotimes_{u=1}^{n_0} \psi_{0,u}
        \ \otimes \
        \bigotimes_{v=1}^{m_0} \psi_{0,v}
        &= \bigotimes_{u=1}^{n_0} \big[ [\gamma_{c_{0,u}}^* \CT_{|\{0\} \times c_{0,u}}], z_{c_{0,u}} \big]
        \ \otimes \
        \bigotimes_{v=1}^{m_0} \big[ \gamma_{s_{0,v}}^* \CT_{|\{0\} \times s_{0,v}}, \psi_{s_{0,v}} \big]
        \\[0.1cm]
        &= \bigotimes_{u=1}^{n_0} \big[ [\gamma_{c_{0,u}}^* \CT_{|\{1\} \times c_{0,u}}], z'_{c_{0,u}} \big]
        \ \otimes \
        \bigotimes_{v=1}^{m_0} \big[ \gamma_{s_{0,v}}^* \CT_{\{1\} \times s_{0,v}}, \psi'_{s_{0,v}} \big]\,,
\end{align}
with
\begin{align}
        z'_{c_{0,u}} &= z_{c_{0,u}} \, \exp \bigg( \int_{[0,1] \times c_{0,u}} - \rho \bigg)
        \qquad \text{and}
        \\
        \psi'_{s_{0,v}} &= pt^{E_y}_{[0,1] \times \{y\}} \circ \psi_{c_{0,u}} \circ \big( pt^{E_x}_{[0,1] \times \{x\}} \big)^{-1}
        \cdot \exp \bigg( \int_{[0,1] \times c_{0,u}} - \rho \bigg)\,.
\end{align}
Here, $x$ is the initial point of $s_{0,v}$ and $y$ is its endpoint.
Further, we have set
\begin{equation}
        E_x \coloneqq \Delta(h_{|[0,1] \times \{x\}}^*\CE_{i(x)}, \CT_{|[0,1] \times c_{0,u}})\,,
\end{equation}
and similarly for $E_y$.
With these choices of $z'_{c_{0,u}}$ and $\psi'_{s_{0,v}}$, the identity~\eqref{eq:change of classes for thho invar} holds true because the vectors in the second and the third lines are related by parallel transports in $\scL$ and $\scR_{ij}$ along \emph{constant} paths (compare~\eqref{eq:pt in R_ij} and~\cite{Waldorf--Trangression_II,BW:Transgression_and_regression_of_D-branes}), which are identity maps.

If $b \subset \partial_2 \Sigma$ is an open connected face, set
\begin{equation}
        E_b \coloneqq \Delta(h_{|[0,1] \times b}^*\CE_{i(b)}, \CT_{[0,1] \times b})\,.
\end{equation}
Its contribution to the the amplitude at $t = 1$ is the parallel transport $\lambda_b \coloneqq pt^{E_b}_{\{1\} \times b}$, which enters in the trace term (cf.~(SA3)(i) and Definition \ref{def:parameterised_amplitude}).
Let $x$ and $y$ denote the initial and endpoint of $b$, respectively.
We have
\begin{align}
        pt^{E_b}_{\{1\} \times b}
        &= pt^{E_b}_{[0,1] \times \{y\}} \circ pt^{E_b}_{\{0\} \times b} \circ \big( pt^{E_b}_{[0,1] \times \{x\}} \big)^{-1}
        \circ \hol \big( E_b, \partial( [0,1] \times b ) \big)
        \\
        &= pt^{E_b}_{[0,1] \times \{y\}} \circ pt^{E_b}_{\{0\} \times b} \circ \big( pt^{E_b}_{[0,1] \times \{x\}} \big)^{-1}
        \cdot \exp \bigg( \int_{[0,1] \times b} - \rho \bigg)\,,
\end{align}
where we have used equation~\eqref{eq:descent eq. for twisted field strength} and that $h_{|[0,1] \times b}$ has rank one.
By an analogous argument, if $b \subset \partial_2\Sigma$ is a closed connected face, its contribution to the surface amplitude at $t=1$ differs from that at $t=0$ by the factor $\exp(\int_{[0,1] \times b} - \rho)$.

Inserting these findings into the expression for the surface amplitude and applying Stokes'~Theorem, we obtain the identity~\eqref{eq:amplitude and homotopies}.
\end{proof}

\begin{proposition}
\label{st:A^G and thin maps to M}
Let $(\Sigma,\sigma, \psi_1^\vee, \psi_0)$ and $(\Sigma,\sigma', \psi_1^\vee, \psi_0)$ be  scattering diagrams for a target space brane geometry $(\CG,\CE)$, with the same underlying  $\<3\>^*$-manifold $\Sigma$, such that $\sigma$ and $\sigma'$ agree on $\partial\Sigma$. If  $\sigma, \sigma' \colon \Sigma \to M$ are thin maps, i.e.~ $\rank(\sigma_{*|x}), \rank(\sigma'_{*|x}) < 2$ for all $x \in \Sigma$, then we have
\begin{equation}
        \CA^{\CG,\CE}[\Sigma, \sigma, \psi_1^\vee, \psi_0]
        = \CA^{\CG,\CE}[\Sigma, \sigma', \psi_1^\vee, \psi_0]\,.
\end{equation}
\end{proposition}

\begin{proof}
By~\cite[Theorem~C.1]{BW:Transgression_and_regression_of_D-branes} one can find trivialisations $\CT \colon \sigma^*\CG \to \CI_0$ and $\CT' \colon \sigma'{}^*\CG \to \CI_0$.
In this case, the exponential terms in the surface amplitudes are trivial.
Further, since $\sigma$ and $\sigma'$ agree on $\partial \Sigma$ the representatives of the states agree upon changing the trivialisations on $\partial \Sigma$, so that all remaining terms in the amplitudes agree as well.
\end{proof}

Summarising Propositions \ref{st:amplitude and homotopies} and \ref{st:A^G and thin maps to M}, the surface amplitude $\CA^{\CG,\CE}$ has the following properties:
\begin{enumerate}[(1)]
\item $\CA^{\CG,\CE}[\Sigma, \sigma_0, \psi_0, \psi_1^\vee] = \CA^{\CG,\CE}[\Sigma, \sigma_1, \psi_0, \psi_1^\vee]$ if $\sigma_0$ and $\sigma_1$ are thin homotopic rel boundary.

\item $\CA^{\CG,\CE}[\Sigma, \sigma_0, \psi_0, \psi_1^\vee] = \CA^{\CG,\CE}[\Sigma, \sigma_1, \psi_0, \psi_1^\vee]$ whenever  $\sigma_0$ and $\sigma_1$ are thin maps that agree on $\partial \Sigma$.
\end{enumerate}
Eventually, we refer to  these properties by saying that the surface amplitude $\CA^{\CG,\CE}$ is \emph{superficial}.

\section{Smooth open-closed functorial field theories}
\label{sect:Smooth TFTs}

We turn to the construction of a smooth open-closed bordism category suitable for our purposes and to the notion of smooth open-closed FFTs.
Subsequently, we assemble the amplitudes defined in Section~\ref{sect:amplitudes} into a smooth FFT in Section~\ref{sect:TFT_construction}.

\subsection{Smooth open-closed bordisms on a target space}
\label{sect:families of bordisms}

Let $(M,Q)$ be a target space.
In this section, we define a smooth version of the open-closed bordism category in dimension $d \geq 2$.
Roughly speaking, its objects are smooth families of compact $(d{-}1)$-dimensional manifolds $Y$, possibly with boundary, and its morphisms are smooth families of compact $d$-dimensional  bordisms with corners.
Additionally, objects and bordisms carry smooth maps to $M$ that map brane boundaries to the submanifolds $Q_{i}$.
Our model for smooth bordism categories borrows from that in~\cite{ST:SuSy_FTs_and_generalised_coho}, modified by adding boundaries, smooth maps to a target space, and collars.

Let $\Cart$ denote the category of cartesian spaces, whose objects are  submanifolds $U$ of some $\FR^n$ that are diffeomorphic to some $\FR^m$, and whose morphisms are smooth maps between these submanifolds.
Our goal is to assemble bordisms into a presheaf of symmetric monoidal categories on $\Cart$, which we will denote by $\OCBord_d(M,Q)$. 

Let $U \in \Cart$.
If $Y$ is a $(d{-}1)$-manifold with boundary and if $f \colon U \times Y \to M$ is a smooth map, we say that $f$ \emph{has fibre-wise sitting instants} if there exists an open neighbourhood $V$ of $\partial Y$ in $Y$ and a diffeomorphism $g \colon V \to \partial Y \times \FR_{\geq 0}$ such that the following diagram commutes:
\begin{equation}
\xymatrix@C=2cm{
        U \times V \ar[r]^{f} \ar[d]_{1_U \times g}
        & M
        \\
        U \times \partial Y \times \FR_{\geq 0} \ar[r]_-{\pr_{U \times \partial Y}}
        & U \times \partial Y \ar[u]_{f}
}
\end{equation}

\begin{definition}
\label{def: family of objects}
An object of the category $\OCBord_{d}(M,Q)(U)$ is a quintuple $(Y,f,b, \ori_T,\ori_\FR)$ of the following data:
\begin{itemize}
        \item $Y$ is a compact $(d{-}1)$-manifold with an orientation $\ori_T$ of its stabilised tangent bundle $\underline{\FR} \oplus TY$ and a map $\ori_\FR \colon \pi_0(Y) \to \{+,-\}$,
        
        \item $b \colon \pi_0(\partial Y) \to I$ is a map to the set $I$ of  D-brane labels,
        
        \item $f \colon U \times Y \to M$ is a smooth map with fibre-wise sitting instants and such that for each $y \in \partial Y$ and $x \in U$ we have $f(x,y) \in Q_{b[y]}$, where $[y] \in \pi_0(\partial Y)$ denotes the connected component of $y$ in $\partial Y$.
\end{itemize}
\end{definition}

Here we regard the symbols $\{+,-\}$ as the two possible orientations on $\FR$, so that the map $\ori_{\mathbb{R}}$ assigns to each connected component of $Y$ an orientation on $\FR$ (with $+$ amounting to the standard orientation).
We induce an orientation on $Y$ by saying that a local frame $(e_1, \ldots, e_{d-1})$ on $Y$ at a point $y \in Y$ is oriented if the tuple $(\ori_{\mathbb{R}}[y], e_1, \ldots, e_{d-1})$ is oriented with respect to $\ori_T$ as a local frame for the stabilised tangent bundle $\underline{\FR} \oplus TY$, where $[y] \in \pi_0(Y)$ is the connected component containing $y$.
An important difference between our bordism category and other versions is that the orientations of the manifolds $Y$ are not part of the data of the object; the orientation of $Y$ is determined from the orientations $\ori_{\mathbb{R}}$ and $\ori_T$.
This is essential for the definition of reflection structures in Section~\ref{sect:variants of OCFFTS}.

We will generally omit the map $b$ and the orientations $\ori_T, \ori_\FR$ from the notation and abbreviate $(Y,f,b,\ori_T,\ori_\FR)$ by just writing $(Y,f)$.
Moreover, for $(Y,f)$ an object in $\OCBord_{d}(M,Q)(U)$ with $Y$ a connected manifold, we let $C_\pm Y$ denote the product $\FR_{\pm \ori_\FR} \times Y$, where $\FR_+ \coloneqq [0, \infty)$ and $\FR_- \coloneqq (-\infty, 0]$.
This is endowed with the orientation induced from $\ori_T$ under the canonical identification of a frame $(r, e_1, \ldots, e_{d-1})$ of $C_\pm Y$ at $(t,y) \in C_\pm Y$ with a frame of $\ul{\FR} \oplus TY$ at $y \in Y$.
Note that the datum $\ori_{\mathbb{R}}$ does not enter in the definition of the orientation of $C_\pm Y$; it only enters in the definition of the underlying manifold.
For non-connected $Y$, we define $C_\pm Y$ as the disjoint union of the component-wise construction.

Next, we define the morphisms of $\OCBord_{d}(M,Q)(U)$ over  $U \in \Cart$.
To that end, we consider tuples 
\begin{equation}
\label{eq:morphism family in OCBord}
        (\Sigma, \sigma) \coloneqq \big( (Y_0,f_0), (Y_1,f_1), \Sigma, W_0, W_1, w_0, w_1, \sigma, \ell \big)
\end{equation}
of the following data:
\begin{itemize}
\item The pairs $(Y_a, f_a)$ are objects of $\OCBord_{d}(M,Q)(U)$ for $a = 0,1$.

\item $\Sigma$ is a compact, oriented $d$-dimensional $\<3\>^*$-manifold with a map $\ell \colon \pi_0(\partial_2 \Sigma) \to I$ and a smooth map $\sigma \colon U \times \Sigma \to M$ such that for every $x \in U$ and $y \in \partial_2 \Sigma$ we have that $\sigma(x,y) \in Q_{\ell[x]}$.

\item $W_0$ is an open neighbourhood of $\{0\} \times Y_0$ in $C_+ Y_0$, and $W_1$ is an open neighbourhood of $\{0\} \times Y_1$ in $C_- Y_1$.
We equip $W_0$ and $W_1$ with the orientations induced from the orientations on $CY_0$ and $CY_1$.

\item The $w_a$ are smooth, orientation-preserving embeddings $w_a \colon W_a \to \Sigma$.
These maps must have disjoint images, satisfy that $\partial_a \Sigma = w_a(\{0\} \times Y_a)$, for $a = 0,1$, that $\ell[w_a(y,t)] = b_a[y]$ for all $(t,y) \in W_a$ with $y \in \partial Y_a$.
Furthermore, they must restrict to embeddings
\begin{equation}
(\FR \times \partial Y_a) \cap W_a\ \hookrightarrow\ \partial_2 \Sigma\,.
\end{equation}

\item We demand that the following diagram commutes:
\begin{equation}
\label{eq:sitting instants on bordisms}
\xymatrix@C=1,5cm{U \times W_a \ar@{^(->}[d]_{1_U \times w_a} \ar@{^(->}[r]
& U \times \FR \times Y_a \ar[r]^-{\pr_{U \times Y_a}} & U \times Y_a \ar[d]^{f_a}
        \\
        U \times \Sigma \ar[rr]_{\sigma} & & M
}
\end{equation}
One can view this diagram as a sitting instant condition on the map $\sigma$ in a direction normal to the image of $U \times Y_a$ in $U \times \Sigma$.
\end{itemize}
We then define an equivalence relation on the set of tuples:
consider two tuples
\begin{align}
        (\Sigma,\sigma) &=\big( (Y_0,f_0), (Y_1,f_1), \Sigma, W_0,W_1, w_0, w_1, \sigma, \ell \big)
\\
 (\Sigma',\sigma') &=\big( (Y_0,f_0), (Y_1,f_1), \Sigma', W'_0, W'_1, w_0', w_1', \sigma', \ell' \big)\,
\end{align}
with the same objects $(Y_0,f_0)$ and $(Y_1,f_1)$. 
We say that $(\Sigma,\sigma)$ and $(\Sigma',\sigma')$ are \emph{equivalent} and write 
\begin{equation}
\label{eq:equivalence in OCBord_1}
(\Sigma,\sigma)\sim(\Sigma',\sigma')
\end{equation}
if there exists a triple $(V_0, V_1, \Psi)$ consisting of a diffeomorphism $\Psi \colon U \times \Sigma \to U \times \Sigma'$ that commutes with the projections to $U$ and that preserves brane labels and the orientation on each fibre, as well as common refinements $V_a \subset W_a \cap W'_a$ of the open neighbourhoods of $U \times \{0\} \times Y_a$ in $C_+Y_0$ and $C_-Y_1$, respectively.
These data have to satisfy
\begin{equation}
\label{eq: diffeos of bordisms}
        \sigma' \circ \Psi = \sigma
        \qandq
        (1_U \times w'_{a|V_a}) = \Psi \circ (1_U \times w_{a|V_a})
        \qquad \text{for} \quad a = 0,1\,.
\end{equation}

\begin{definition}
\label{def: family of bordisms}
A morphism in $\OCBord_{d}(M,Q)(U)$ is an equivalence class $[\Sigma, \sigma]$ of a tuple $(\Sigma, \sigma)$.
\end{definition}

Composition of two morphisms $[\Sigma, \sigma] \colon (Y_0,f_0) \to (Y_1,f_1)$ and $[\Sigma', \sigma'] \colon (Y_1,f_1) \to (Y_2,f_2)$ in $\OCBord_d(M,Q)(U)$ is defined by gluing the manifolds $\Sigma$ and $\Sigma'$ using the collars $(W_1, w_1)$ and $(W'_1, w'_1)$ of $(Y_1,f_1)$, just as in the ordinary bordism category.
Note that the maps $\sigma \colon \Sigma \to M$ and $\sigma' \colon \Sigma' \to M$ glue smoothly due to the sitting instant condition~\eqref{eq:sitting instants on bordisms}.
As we are using collars as part of the data of a bordism, gluing of representatives is associative up to a canonical diffeomorphism (see e.g.~\cite{Kock2003}).
Moreover, this diffeomorphism induces an equivalence of tuples:
the preservation of brane labels and orientations is immediate, and the compatibility with the maps to $M$ follows from the fact that on the level of the sets underlying the bordisms, gluing amounts to forming a pushout; the universal property of the pushout then guarantees compatibility with maps out of the glued manifolds.

The identity morphism of an object $(Y,f,b,\ori_T,\ori_\FR=+)$ reads as
\begin{equation}
\label{eq:id bordism}
\begin{split}
        1_{(Y,f)} = &\big[ (Y,f), (Y,f),\, U \times [0,1] \times Y,\, [0,\epsilon) \times Y,\, (-\epsilon,0] \times Y,
        \\
        &\quad \iota_{[0,\epsilon)} \times 1_Y,\, (sh_1 \circ \iota_{(-\epsilon,0]}) \times 1_Y,\, f \circ \pr_{U \times Y},\, b \circ \pr_{\partial Y} \big]
\end{split}
\end{equation}
for some $0 < \epsilon < \frac{1}{2}$.
The bordism $[0,1] \times Y$ is endowed with the orientation induced from $\ori_T$ under the inclusion $[0,1] \times Y \hookrightarrow \FR \times Y$.
Here we have used the canonical embeddings $\iota_{[0,\epsilon)} \colon [0,\epsilon) \hookrightarrow \FR$ of an interval into the real line and the shift map $sh_s \colon \FR \to \FR$, $t \mapsto t+s$ for $s \in \FR$.
That is, the identity bordisms are defined like in the ordinary bordism category, but multiplied by the test space $U$, endowed with the constant extension of $f$ in the direction normal to $\{0\} \times Y$.
Finally, note that shrinking the collar neighbourhoods $W_a$ of a general morphism $[\Sigma, \sigma]$ to smaller collar neighbourhoods while restricting the collar embeddings $w_a$ has no effect on the class $[\Sigma, \sigma]$.
Hence, our definition of the identity bordism $[\Sigma, \sigma]$ does not depend on $\epsilon$ and is neutral with respect to composition.

So far, we have defined a category $\OCBord_d(M,Q)(U)$ for each cartesian space $U\in\Cart$. Under the disjoint union of manifolds this becomes a symmetric monoidal category in a straightforward way. If $g:V \to U$ is a smooth map between cartesian spaces, we define a functor $\OCBord_d(M,Q)(U) \to \OCBord_d(M,Q)(V)$ by setting  $g^*(Y,f) \coloneqq (Y, f \circ (g \times 1_Y))$ for objects $(Y,f)$ and
 \begin{equation*}
g^*[\Sigma, \sigma] \coloneqq \big[ g^*(Y_0,f_0), g^*(Y_1,f_1),\, \Sigma,\, W_0,\, W_1,\,
        w_0,\, w_1,\, \sigma \circ (g \times 1_\Sigma),\, \ell \big]
\end{equation*}  
for morphisms. This turns $\OCBord_d(M,Q)$ into a presheaf of categories on $\Cart$, the presheaf of \emph{smooth $d$-dimensional, oriented open-closed bordisms on $(M,Q)$}.

\begin{remark}
We have chosen a particularly simple dependence on the test space $U \in \Cart$ here, in that we only allow the maps $\sigma$ to the target manifold $M$ to vary over $U$.
The full picture would be to consider bundles of $\<3\>^*$-manifolds $\sfSigma \to U$ with maps $\varsigma \colon \sfSigma \to M$, which can be done by a straightforward adaptation of~\cite{ST:SuSy_FTs_and_generalised_coho}.
However, since we work over contractible test spaces, any such bundle is trivialisable, and the trivial bundles $\sfSigma = U \times \Sigma$ capture all possible smooth families in this sense.
Another reason to not work with the bordism category from~\cite{ST:SuSy_FTs_and_generalised_coho} is that we wish to make contact with earlier literature on open-closed field theories such as~\cite{LP--Open-closed_TQFTs,MS--2DTFTs,Lazaroiu:OCFFT_in_2D}.
The bordism categories used in these references differ in some important features from the model of~\cite{ST:SuSy_FTs_and_generalised_coho}; most strikingly, the identity bordisms in~\cite{ST:SuSy_FTs_and_generalised_coho} are degenerate, and it is somewhat tedious to work out a precise relation between both versions of the bordism category.
\qen
\end{remark}

\begin{remark}
The fact that we do not allow $\Sigma$ to vary over $U$ prohibits us from allowing its collars to vary over $U$: otherwise, composition of families of bordisms would yield families of bordisms that do vary over $U$.
Thus, we keep the collars $(W_a,w_a)$ constant over $U$, but that requires restrictions on the smooth maps $\sigma \colon U \times \Sigma \to M$:
all maps $\sigma_{|\{x\} \times \Sigma} \colon \Sigma \to M$, for $x \in U$, must have sitting instants on these collars, as we impose in \eqref{eq:sitting instants on bordisms}.
However, in our setting this is not actually restrictive:
consider some $U \in \Cart$, some $\<3\>^*$-manifold $\Sigma$, and suppose we are given any smooth map $\sigma \colon U \times \Sigma \to M$.
Suppose further, we have a $(d{-}1)$-manifold $Y$ with a orientations $\ori_T$ and $\ori_\FR$ as in the definition of an object of $\OCBord_d(M,Q)(U)$, as well as an orientation-preserving embedding $w_0 \colon W_0 \to \Sigma$, where $W_0$ is an open neighbourhood of $\{0\} \times Y$ in $C_+Y$.
Let us assume that $Y$ is connected and that $\ori_\FR = +$; the following argument extends to the general case in a straightforward way.
Then, we can choose a $\delta > 0$ such that $[0, \delta) \times Y \subset W_0$ and we can choose a smooth function $\varepsilon \colon [0, \delta) \to [0, \delta)$ that is constantly zero on $[0,\delta/3)$ and which is the identity on $(2\delta/3, \delta)$.
Then, we can replace $\sigma$ by a smooth map $\sigma'$ such that $\sigma' \circ (1_U \times w_0) = \sigma \circ (1_U \times w_0) \circ (1_U \times \varepsilon \times 1_Y)$ on $[0,\delta) \times Y$.
The new map $\sigma'$ now has sitting instants on the collar neighbourhood $[0,\delta/3) \times Y$ and for each $x \in U$ the map $\sigma'_{|\{x\} \times \Sigma}$ is thin homotopic to $\sigma_{|\{x\} \times \Sigma}$.
That is, up to fibre-wise thin homotopy, we still capture all families of smooth maps $U \times \Sigma \to M$.
This is analogous to the treatment of concatenation of smooth paths in~\cite[Paragraph~5.5]{Iglesias-Zemmour:Diffeology}.
\qen
\end{remark}

\subsection{Smooth open-closed functorial field theories}
\label{sec:smoothopenclosedffts}

Let $\VBdl$ denote the sheaf of complex finite-dimensional vector bundles; regarded as a sheaf of symmetric monoidal categories on $\Cart$. 

\begin{definition}
\label{def:smooth field theories}
A \emph{smooth $d$-dimensional open-closed functorial field theory (OCFFT) on a target space $(M,\scQ)$} is a morphism
\begin{equation}
        \scZ \colon \OCBord_d(M,Q) \rightarrow \VBdl
\end{equation}
of presheaves of symmetric monoidal categories on $\Cart$.
\end{definition}

The smooth OCFFTs we will construct in Section \ref{sect:TFT_definition} have certain properties that we  describe next. The first is related to the rank of the vector bundles in its image.

\begin{definition}
\label{def:invertible}
A smooth OCFFT $\scZ$ is called \emph{invertible}, if for all $U \in\Cart$  the following conditions are satisfied:  
 \begin{itemize}
\item 
For all objects $(Y,f)$ in $\OCBord_d(M,Q)(U)$ with $Y$ closed, the vector bundle $\scZ(Y,f)$ is of rank one.

\item
For all morphisms $[\Sigma,\sigma]$ in $\OCBord_{d}(M,Q)(U)$ without brane boundary, i.e., $\partial_2 \Sigma = \emptyset$, the vector bundle morphism $\scZ[\Sigma,\sigma]$ is an isomorphism. 

\end{itemize}
\end{definition}

We define three new equivalence relations on the set of tuples \eqref{eq:morphism family in OCBord}, for every test space $U\in\Cart$. First, two tuples $(\Sigma,\sigma)$ and $(\Sigma,\sigma')$
with the same underlying surface $\Sigma$ are \emph{homotopy equivalent} if there exists a homotopy between $\sigma$ and $\sigma'$ rel boundary that restricts to a constant homotopy on the images of $W_a$ in $\Sigma$ for $a = 0,1$ and on $\partial_2 \Sigma$. Secondly, the tuples $(\Sigma,\sigma)$ and $(\Sigma,\sigma')$ are called \emph{thin homotopy equivalent} if they are homotopy equivalent via a homotopy whose restriction to $\{x\} \times \Sigma$ is thin for all $x\in U$. Thirdly, the tuples $(\Sigma,\sigma)$ and $(\Sigma,\sigma')$ are \emph{superficially equivalent} if $\sigma$ and $\sigma'$ agree on $U \times \partial \Sigma$ and if for every $x \in U$ the maps $\sigma$ and $\sigma'$ restrict to thin maps on the fibre $\{x\} \times \Sigma$.

\begin{definition}
\label{def:thinhomotopyinvariantocfft}
A smooth OCFFT $\scZ$ on $(M,Q)$ is called:
\begin{enumerate}[(1)]

\item
\emph{homotopy invariant} if $\scZ[\Sigma,\sigma]=\scZ[\Sigma,\sigma']$ whenever $(\Sigma,\sigma)$ and $(\Sigma',\sigma')$ are  homotopy equivalent.

\item
\emph{thin homotopy invariant} if $\scZ[\Sigma,\sigma]=\scZ[\Sigma,\sigma']$ whenever $(\Sigma,\sigma)$ and $(\Sigma',\sigma')$ are thin homotopy equivalent.

\item
\emph{superficial} if it is thin homotopy invariant and satisfies $\scZ[\Sigma,\sigma]=\scZ[\Sigma',\sigma']$ whenever $(\Sigma,\sigma)$ and $(\Sigma',\sigma')$ are superficially equivalent.

\end{enumerate} 
\end{definition}

Since smooth OCFFTs are defined as morphisms in a 2-category, they come naturally organised into a category $\OCFFT_d(M,\scQ)$. Superficial OCFFTs, homotopy invariant OCFFTs, and thin homotopy invariant OCFFTs form, respectively, full subcategories    
\begin{equation}
        \OCFFT_{d}^{\sf}(M,Q)\ ,\ \OCFFT_d^{h}(M,Q) \ \subset \ \OCFFT_{d}^{th}(M,Q)\ \subset \ \OCFFT_d(M,\scQ)\,.
\end{equation}
Invertibility will be denoted by the symbol $()^{\times}$ in all cases, and refers again to the \emph{full subcategory} on all invertible OCFFTs.

It is possible to combine the equivalence relation \eqref{eq:equivalence in OCBord_1} and (thin) homotopy invariance on tuples into a single equivalence relation, in such a way that the composition of bordisms is well-defined on equivalence classes. This results in  new presheaves $\OCBord_{d}^{h}(M,Q)$ and $\OCBord_{d}^{th}(M,Q)$ of symmetric monoidal categories, together with  quotient morphisms 
\begin{equation*} 
\OCBord_d(M,Q) \to \OCBord_d^{th}(M,Q)\to \OCBord_d^{h}(M,Q).
\end{equation*}
We have the following obvious result.

\begin{proposition}
\label{prop:thinhomotopyinvariantocfft}
A smooth OCFFT is (thin) homotopy invariant if and only if it factors through the quotients  $\OCBord_d^{h}(M,Q)$ and $\OCBord_d^{th}(M,Q)$, respectively.
\end{proposition}

\begin{remark}
A similar treatment of \emph{superficial} OCFFTs is not easy to obtain. The reason is that two superficially equivalent bordisms may be composed separately with a fixed third bordism, whose target space map $\sigma$ is not thin. Then, the two separate composites will in general not be superficially equivalent; superficial equivalence is not preserved under composition. 

In principle, it is possible to construct a presheaf of symmetric monoidal categories $\OCBord_d^{\sf}(M,Q)$ by starting with the free symmetric monoidal categories $\OCBord_d^{th}(M,Q)(U)$ and imposing the relation that superficially equivalent morphisms be identified, somewhat in the spirit of Gabriel-Zisman localisation.
Superficial smooth OCFFTs would then equivalently be morphisms $\OCBord_d^\sf(M,Q) \to \VBdl$.
However, we feel that such a construction would take us to far away from geometrical and  physical intuition. Therefore, we decided to treat superficiality, and then, for the sake of consistency, also (thin) homotopy invariance as  additional conditions on functors defined on $\OCBord_d(M,Q)$, rather than working with functors from $\OCBord_d^{th}(M,Q)$, $\OCBord_d^{h}(M,Q)$, or $\OCBord_d^{\sf}(M,Q)$.
\qen
\end{remark}

The presheaf $\OCBord_d^{th}(M,Q)$ is interesting for another reason as well, though, namely for the study of \emph{path bordisms}. 
Consider two objects $(Y,f_a)\in \OCBord_d(M,Q)(U)$, $a = 0,1$, with the same underlying connected manifold $Y$.
We assume that $\ori_\FR(Y) = +$ for both objects and consider a morphism $(Y,f_0) \to (Y,f_1)$  represented by
\begin{equation}
\label{eq:path bordism}
\begin{split}
        ([0,1] \times Y, \sigma) &= \big( (Y,f_0), (Y,f_1),\, [0,1] \times Y,\, [0,\epsilon) \times Y,\, (-\epsilon,0] \times Y,
        \\
        &\qquad \iota_{[0,\epsilon)} \times 1_Y,\, (sh_1 \circ \iota_{(-\epsilon,0]}) \times 1_Y,\, \sigma,\, b \circ \pr_Y \big)
\end{split}
\end{equation}
for some smooth map $\sigma \colon U \times [0,1] \times Y \to M$.
As in the construction of the identity bordism~\eqref{eq:id bordism}, the bordism $[0,1] \times Y$ is endowed with the orientation induced from $\ori_T$ under the inclusion $[0,1] \times Y \hookrightarrow \FR \times Y$.
We call such a bordism a \emph{path bordism} since it represents a smooth family of paths in the diffeological mapping space $M^Y$ with sitting instants.
This is easily extended to objects with non-connected $Y$.
The following is essentially the statement that that smooth paths in $M$ are invertible only up to thin homotopy.

\begin{lemma}
\label{lem:pathbordismsareinvertible}
Path bordisms are invertible in $\OCBord_d^{th}(M,Q)$.
\end{lemma}   

\begin{proof}
We can restrict ourselves to connected $Y$.
We provide an explicit inverse path bordism in the case where $\ori_\FR = +$; in the case of $\ori_\FR = -$, the construction adapts straightforwardly.
The inverse of the bordism in~\eqref{eq:path bordism} is given by
\begin{equation}
\begin{split}
        \big[ [0,1] \times Y, \sigma \big]^{-1}
        &=\big[ (Y,f_1), (Y,f_0),\, [0,1] \times Y,\, [0,\epsilon) \times Y,\, (-\epsilon,0] \times Y,
        \\*[0.1cm]
        &\qquad \iota_{[0,\epsilon)} \times 1_Y,\, (sh_1 \circ \iota_{(-\epsilon,0]}) \times 1_Y,\, \sigma \circ (1_U \times \sfr_{\frac{1}{2}} \times 1_Y),\, b \circ \pr_Y \big]\,,
\end{split}
\end{equation}
where $\sfr_s \colon \FR \to \FR$, $t \mapsto 2s - t$ is the reflection of $\FR$ at $s \in \FR$.
This is the reversal of the path in the mapping space $M^Y$ that is defined by the bordism $[[0,1] \times Y, \sigma]$.
\end{proof}

Note that we need to employ a thin homotopy of paths in the mapping space $M^Y$ in order to obtain the identity $[[0,1] \times Y, \sigma]^{-1} \circ [[0,1] \times Y, \sigma] = 1_{(Y,f)}$; that is, path bordisms are not invertible in $\OCBord_d(M,Q)$ in general, but they are always invertible in $\OCBord_d^{th}(M,Q)$.

\subsection{Duals, opposites, and hermitean structures for  bordisms}
\label{sect:Duals and Opposites}

In this section we work over a fixed test space $U \in \Cart$, and thus with the symmetric monoidal category $\OCBord_{d}(M,Q)(U)$. We  note that this monoidal category  has (left) duals;
moreover, there is a canonical choice of duality data
for every object in $\OCBord_d(M,Q)(U)$, turning $\OCBord_d(M,Q)$ into a symmetric monoidal category \emph{with fixed duals}; see Appendix~\ref{app:duals in monoidal cats}.
Indeed, if $(Y,f) \in \OCBord_d(M,Q)(U)$ is an object, duality data $((Y,f)^\vee, \ev_Y, \coev_Y)$ is fixed as follows:
\begin{equation}
        \big( Y, f, b, \ori_T, \ori_\FR \big)^\vee = \big( Y, f, b, \ori_T, -\ori_\FR \big)\,.
\end{equation}
Observe that the orientation on $\underline{\FR} \oplus TY$ is unchanged, so that the orientation induced on $Y$ is reversed, as a consequence of reversing the orientation $\ori_\FR(Y)$.
The evaluation $\ev_{Y}$ and coevaluation $\coev_Y$ are given by the standard evaluation and coevaluation bordisms represented by the manifold $\Sigma = [0,1] \times Y$ and its standard collars, multiplied by $U$ and decorated with the maps $\sigma = f \circ \pr_{U \times Y}$.
In the following we treat  this choice of fixed duals in terms of (twisted) involutions. 

\begin{definition}
\label{def:involution}
Let $\scC$ be a symmetric monoidal category.
An \emph{involution on $\scC$} is a pair $(d, \delta)$ of a symmetric monoidal functor $d \colon \scC \to \scC$ and a  monoidal natural isomorphism \smash{$\delta \colon d \circ d \to 1_\scC$} satisfying
\begin{equation}
        1_{d} \circ \delta = \delta\circ 1_{d}
\end{equation}
as natural transformations $d \circ d \circ d \to d$.
A \emph{twisted involution on $\scC$} is a pair $(d, \delta)$ of a symmetric monoidal functor $d \colon \scC^\opp \to \scC$ and a monoidal natural isomorphism \smash{$\delta \colon d \circ d^\opp \to 1_{\scC}$} satisfying 
\begin{equation}
1_d \circ (\delta^\opp)^{-1} = \delta \circ 1_{d}
\end{equation}
as natural transformations $d \circ d^\opp \circ d \to d$.
A (twisted) involution $(d, \delta)$ is called \emph{strict} if  $\delta$ is the identity. 
\end{definition}

In general, in a symmetric monoidal category with fixed duals, the assignment $x\mapsto x^{\vee}$ extends canonically to a symmetric monoidal functor $d_{\scC}: \scC^{\opp} \to \scC$ and further to a  twisted involution $(d_{\scC},\delta_{\scC})$, which we then call the \emph{duality involution} of $\scC$.
In the case of the symmetric monoidal category $\OCBord_d(M,Q)(U)$ the duality involution $d_{\OCBord}$ is in fact strict, and it is straightforward to check that its action on morphisms is explicitly given by  
\begin{equation}
        \big[ (Y_0,f_0), (Y_1,f_1), \Sigma, W_0,W_1, w_0, w_1, \sigma, \ell \big]^\vee
        = \big[ (Y_1,f_1)^\vee, (Y_0, f_0)^\vee, \Sigma, W_1, W_0, w_1, w_0, \sigma, \ell \big]\,.
\end{equation}

There is another interesting operation on $\OCBord_d(M,Q)(U)$, given on objects and morphisms by the following definitions:
\begin{equation}
\begin{split}
        \overline{\big( Y,f,b,\ori_T,\ori_\FR \big)} &\coloneqq \big( Y, f, b, -\ori_T, \ori_\FR \big)\,,
        \\
        \overline{\big[ (Y_0,f_0), (Y_1,f_1), \Sigma, W_0,W_1, w_0, w_1, \sigma, \ell \big]}
        &\coloneqq \big[ \overline{(Y_0,f_0)}, \overline{(Y_1,f_1)}, \overline{\Sigma}, W_0, W_1, w_0, w_1, \sigma, \ell \big]\,.
\end{split}
\end{equation}
In our shorthand notation, we call $\overline{(Y,f)}$ the \emph{opposed bordism} of $(Y,f)$.
It is straightforward to check that opposition defines a strict involution  
\begin{equation*}
op_{\OCBord} \colon \OCBord_d(M,Q)(U) \to \OCBord_d(M,Q)(U)\text{.}
\end{equation*}

The duality involution $d_{\OCBord}$ and the involution $op_{\OCBord}$ are in fact  related. This relation is described by  a so-called \emph{hermitian structure}. We first provide the relevant definitions in a more general context, following \cite{FH:Reflection_positivity}. 

Let $\scC$ be a symmetric monoidal category  with fixed duals, whose duality data we denote by $(x^{\vee},\ev_x,\coev_x)$,  and with corresponding duality involution $(d_{\scC},\delta)$. We suppose that  $(op_{\scC},\gamma)$ is some involution on $\scC$. If $x \in \scC$ is an object, we will also write  $\overline{x}$ as a shorthand for $op_{\scC}(x)$.

\begin{definition}
\label{def:pre-herm structure}
A \emph{pre-hermitean structure on an object $x \in \scC$} is an isomorphism $\flat_x \colon \overline{x}\to x^{\vee}$.
We denote the inverse of $\flat_x$ by $\sharp_x \coloneqq \flat_x^{-1} \colon x^{\vee} \to \overline{x}$.
\end{definition}

A pre-hermitean structure $\flat_x$ on an object $x \in \scC$ defines a \emph{pre-hermitean pairing} on $x$, by which we mean the morphism
\begin{equation}
\label{eq:hermitean pairing from hermitean structure}
        h_x \colon \overline{x} \otimes x \to u\,,
        \quad
        h_x \coloneqq  \ev_x \circ (\flat_x \otimes 1_x)\,,
\end{equation}
where $u$ is the monoidal unit in $\scC$.
This pairing is non-degenerate in the sense that we have a \emph{pre-hermitean co-pairing}
\begin{equation}
        \check{h}_x \colon u \to x \otimes \overline{x}\,,
        \quad
        \check{h}_x \coloneqq (\sharp_x \otimes 1_{x^\vee}) \circ \coev_x\,,
\end{equation}
so that the usual triangle identities are satisfied.

\begin{definition}
\label{def:herm structure}
A pre-hermitean structure $\flat_x$ on an object $x \in \scC$ is called \emph{hermitean structure} if its pre-hermitian pairing $h_x$ is \emph{symmetric} in the sense that the following diagram commutes
\begin{equation}
\xymatrix@C=2cm{        \overline{x} \otimes x \ar[rr]^{h_x} \ar[d]_{swap} & & u
        \\
        x \otimes \overline{x} \ar[r]_{\gamma_{x}^{-1} \otimes 1_{\overline{x}}} & \overline{\overline{x}} \otimes \overline{x} \ar[r] & \overline{ \overline{x} \otimes x }. \ar[u]_{\overline{h_x}}
}
\end{equation}
\end{definition}

\begin{remark}
We remark that a hermitean structure relates the values of functors with different domain categories; therefore, one cannot regard it as a natural transformation. Instead, compatibility with morphisms will be treated as a condition on the \emph{morphisms}, see Definition \ref{def:unitarity}. We also remark that a property \quot{positive-definite} cannot be defined in this general context; this will be done by hand for the category of vector bundles, see Definition \ref{def:positive reflection structure}. 
\qen
\end{remark}

\begin{definition}
\label{def:unitarity}
Let  $f \colon x \to y$ be a morphism between objects with hermitian structures $\flat_{x}$ and $\flat_y$, respectively, and induced hermitian pairings $h_x$ and $h_y$, respectively.
We call $f$ is \emph{isometric} if it satisfies the identity
\begin{equation}
        h_y\circ(\overline{f} \otimes f) = h_x\,.
\end{equation}
We say that $f$ is \emph{unitary} if it is an isometric isomorphism.
The \emph{adjoint of $f$} is the composition
\begin{equation}
\xymatrix{
        f^* \colon \overline{y} \ar[r]^{\flat_y} & y^\vee \ar[r]^{f^{\vee}} & x^\vee \ar[r]^{\sharp_x} & \overline{x}\,.
}
\end{equation}
\end{definition}

The name \emph{adjoint} is justified by the following lemma.

\begin{lemma}
For a morphism as in Definition~\ref{def:unitarity} we have
\begin{equation}
        h_y \circ (1_{\overline{y}} \otimes f) = h_x \circ (f^* \otimes 1_x)\,.
\end{equation}
\end{lemma}

\begin{proof}
We compute
\begin{align}
        h_y \circ (1_y \otimes f) &= \ev_y \circ (\flat_y \otimes 1_y) \circ (1_{\overline{y}} \otimes f)
        \\
        &= \ev_y \circ (1_{y^\vee} \otimes f) \circ (\flat_y \otimes 1_x)
        \\
        &= \ev_x \circ (f^\vee \otimes 1_x) \circ (\flat_y \otimes 1_x)
        \\
        &= \ev_x \circ (\flat_x \otimes 1_x) \circ \big( ( \sharp_x \circ f^\vee \circ \flat_y) \otimes 1_x \big)\,,
\end{align}
where we have used Lemma~\ref{st:dual behaves like transpose} in the third step.
\end{proof}

\begin{proposition}
\label{st:unitarity for isos}
Let  $f \colon x \to y$ be an isomorphism between objects with hermitian structures $\flat_{x}$ and $\flat_y$. Then, $f$ is unitary if and only if
\begin{equation}
        f^* = \overline{f}^{-1}\,.
\end{equation}
\end{proposition}

\begin{proof}
We need to show that $f$ is isometric if and only if $f^* = \overline{f}^{-1}$.
To that end, we compute
\begin{align}
        h_y \circ (\overline{f} \otimes f) &= \ev_y \circ \big( (\flat_y \circ \overline{f}) \otimes f \big)
        \\
        &= \ev_x \circ \big( (f^\vee \circ \flat_y \circ \overline{f}) \otimes 1_x \big)
        \\
        &= \tau^{x\, -1}_{\overline{x}, u} \big( f^\vee \circ \flat_y \circ \overline{f} \big)\,,
\end{align}
where $\tau^{x}_{\overline x,u}$ is part of the adjunction of Proposition~\ref{prop:adjunction}. 
On the other hand, using this adjunction, we may write~\eqref{eq:hermitean pairing from hermitean structure} as $h_x = \tau^{x\, -1}_{\overline{x}, u} (\flat_x)$.
Since $\tau^{x\, -1}_{\overline{x}, u}$ is bijective, we thus have that
\begin{align}
        h_y \circ (\overline{f} \otimes f) = h_x
        \quad \Leftrightarrow \quad
        f^\vee \circ \flat_y \circ \overline{f} = \flat_x
                \quad \Leftrightarrow \quad
        \overline{f}^{-1} = \sharp_x \circ f^\vee \circ \flat_y .
\end{align}
With the definition of the adjoint morphism $f^{*}$ this proves the assertion.
\end{proof}

\begin{example}
\label{ex:fixed duals in Vect}
The symmetric monoidal category $\Vect$ of finite-dimensional complex vector spaces has its usual fixed duals:
for a vector space $V \in \Vect$, we set
\begin{alignat}{3}
        V^\vee &\coloneqq \Hom_\FC(V, \FC)\,, \qquad && &&
        \\
        \ev_V &\colon V^\vee \otimes V \to \FC\,, &&
        \psi \otimes v &&\mapsto \psi(v)
        \\
        \coev_V &\colon \FC \to V^\vee \otimes V\,, &&
        \quad z &&\mapsto z \cdot 1_V\,.
\end{alignat}
We consider the monoidal strict involution $op_{\Vect}$ sending $V$ to the complex conjugate vector space $\overline{V}$.
Then, a hermitean structure on $V \in \Vect$ in the sense of Definition~\ref{def:herm structure} is the same as a hermitian metric on $V$ (not  necessarily positive definite).
The notions of isometric morphisms and adjoints reproduce precisely the classical ones. Analogous statements hold for the category $\VBdl(M)$ of vector bundles over a manifold $M$. 
\qen
\end{example}

We return to the concrete situation of the symmetric monoidal category $\OCBord_{d}(M,Q)(U)$ of open-closed bordisms over a test space $U \in \Cart$.
On each object $(Y,f)$ of $\OCBord_{d}(M,Q)(U)$ we identify a canonical pre-hermitean structure 
\begin{equation*}
\flat_{(Y,f)} \colon \overline{(Y,f)} \longrightarrow (Y, f)^\vee
\end{equation*}
with respect to the involution  $op_{\OCBord}$.
It is given by  
\begin{equation}
\label{eq:canonical hermitean structure}
\begin{split}
        \flat_{(Y,f)} &\coloneqq \big[ \overline{(Y,f)}, (Y, f)^\vee,\, [0,1] \times Y,\, [0,\epsilon) \times Y,\, [0,\epsilon) \times Y,
        \\
        &\qquad (sh_1 \circ \sfr_0 \circ \iota_{[0,\epsilon)}) \times 1_Y,\, \iota_{[0,\epsilon)} \times 1_Y,\, f \circ \pr_{U \times Y},\, b \circ \pr_Y \big]\,.
\end{split}
\end{equation}
Here we have made use of the involution $\sfr_s \colon \FR \to \FR$, $t \mapsto 2s - t$ that reflects the real line at an element $s \in \FR$.
Observe that because of this reflection, the inclusion $w_0$ in $\flat_{(Y,f)}$ is indeed orientation-preserving if the product $[0,1] \times Y$ is endowed with the orientation induced from $\ori_T$.
It is straightforward to check that the pre-hermitean structure $\flat_{(Y,f)}$ of \eqref{eq:canonical hermitean structure} is hermitean. 

For the following argument we consider the quotient category $\OCBord_d^{th}(M,Q)(U)$ where thin homotopy equivalent morphisms are identified. It is easy to see that  the involutions $d_{\OCBord}$ and $op_{\OCBord}$, as well as the hermitean structures $\flat_{(Y,f)}$ descend to this quotient category.

\begin{proposition}
\label{st:path bordisms are unitary}
Path bordisms $([0,1] \times Y, \sigma) \colon (Y,f_0) \to (Y,f_1)$ are unitary  in the symmetric monoidal category $\OCBord_2^{th}(M,Q)(U)$, with respect to the hermitean structures $\flat_{(Y,f_0)}$ and $\flat_{(Y,f_1)}$.
\end{proposition}

\begin{proof}
By Lemma \ref{lem:pathbordismsareinvertible}, path bordisms are invertible in $\OCBord_d^{th}(M,Q)(U)$.
Recall that the bordism $[0,1] \times Y$ in a path bordism has its orientation induced from $\ori_T$.
We calculate:
\begin{align}
        op_{\OCBord} \big[ [0,1] \times Y, \sigma \big]^{-1}
        &= op_{\OCBord} \big[ (Y,f_1), (Y,f_0),\, [0,1] \times Y,\, [0,\epsilon) \times Y,\, (-\epsilon,0] \times Y,
        \\*[0.1cm]
        &\qquad \iota_{[0,\epsilon)} \times 1_Y,\, (sh_1 \circ \iota_{(-\epsilon,0]}) \times 1_Y,\, \sigma \circ (1_U \times \sfr_{\frac{1}{2}} \times 1_Y),\, b \circ \pr_Y \big]
        \\[0.2cm]
        &= \big[ \overline{(Y,f_1)}, \overline{(Y,f_0)},\, \overline{[0,1] \times Y},\, [0,\epsilon) \times Y,\, (-\epsilon,0] \times Y,
        \\*[0.1cm]
        &\qquad \iota_{[0,\epsilon)} \times 1_Y,\, (sh_1 \circ \iota_{(-\epsilon,0]}) \times 1_Y,\, \sigma \circ (1_U \times \sfr_{\frac{1}{2}} \times 1_Y),\, b \circ \pr_Y \big]
        \\[0.2cm]
        &= \big[ \overline{(Y,f_1)}, \overline{(Y,f_0)},\, [0,1] \times Y,\, [0,\epsilon) \times Y,\, (-\epsilon,0] \times Y,
        \\*[0.1cm]
        &\qquad (sh_1 \circ \sfr_0 \circ \iota_{[0,\epsilon)}) \times 1_Y,\, (\sfr_0 \circ \iota_{(-\epsilon, 0]}) \times 1_Y,\, \sigma,\, b \circ \pr_Y \big]\,,
        \\[0.2cm]
        &= \sharp_{(Y,f_0)} \circ \big[ [0,1] \times Y, \sigma \big]^\vee \circ \flat_{(Y,f_1)}
        \\*[0.1cm]
        &= \big[ [0,1] \times Y, \sigma \big]^*\,.
\end{align}
Here, the third identity arises from the orientation-preserving diffeomorphism
\begin{equation}
        1_U \times \sfr_{\frac{1}{2}} \times 1_Y \colon\ U \times \overline{[0,1] \times Y} \longrightarrow U \times [0,1] \times Y\,.
\end{equation}
Observe that this change of orientation on $[0,1] \times Y$ is accounted for by changing the embeddings $w_0$ and $w_1$ accordingly.
The following identity arises from another diffeomorphism and a thin homotopy to express the composition $\sharp_{(Y,f_0)} \circ [[0,1] \times Y, \sigma]^\vee \circ \flat_{(Y,f_1)}$ by a representative that relates directly to the standard representative of a path bordism.
(The diffeomorphism in question amounts to a rescaling $[0,3] \cong [0,1]$.)
Now the claim follows from Proposition~\ref{st:unitarity for isos}.
\end{proof}

Finally, we want to consider again $\OCBord_d(M,Q)$ as a presheaf over $\Cart$. Definition~\ref{def:involution} generalizes in a straightforward way to (twisted) involutions on presheaves of symmetric monoidal categories, by requiring that the functors $d$ and the natural isomorphism $\delta$ are morphisms and 2-morphisms of presheaves, respectively. 
It is straightforward to check that in case of $\OCBord_d(M,Q)$ the duality involution $d_{\OCBord}$ as well as the involution $op_{\OCBord}$ are indeed morphisms of presheaves of symmetric monoidal categories over $\Cart$. Similarly, the duality involution $d_{\VBdl}$ and the involution $op_{\VBdl}$ on the category $\VBdl(U)$ described in Example \ref{ex:fixed duals in Vect} become morphisms of presheaves of symmetric monoidal categories over $\Cart$.

\subsection{Positive reflection structures on OCFFTs}
\label{sect:variants of OCFFTS}

In this section we describe how a smooth OCFFT $\scZ$ relates the involutions $d_{\OCBord}$ and $op_{\OCBord}$ on the presheaf $\OCBord_d(M,Q)$ to the  involutions $d_{\VBdl}$ and $op_{\VBdl}$ on the presheaf $\VBdl$ of vector bundles.

Again, we first discuss a more general setting. Let $\scC$ and $\scD$ be symmetric monoidal categories with fixed duality data and associated duality involutions $(d_{\scC},\delta_{\scC})$ and $(d_{\scD},\delta_{\scD})$. Then, for any symmetric monoidal functor $F:\scC \to \scD$, there exists a unique  natural isomorphism 
\begin{equation}
\label{eq:canonicaldualityiso}
\beta \colon F \circ d_{\scC} \to d_\scD \circ F^{\opp}
\end{equation}
that is compatible with evaluation and coevaluation, and which makes the diagram
\begin{equation}
\xymatrix@C=2cm{F \circ d_{\scC} \circ d_{\scC}^{\opp} \ar[r]^-{\beta \circ 1} \ar[d]_{1 \circ \delta_{\scC}} & d_\scD \circ F^{\opp} \circ d_\scC^{\opp} \ar[d]^{1 \circ (\beta^{\opp})^{-1}}
        \\
        F & d_\scD \circ d_{\scC}^{\opp} \circ F \ar[l]^{\delta_{\scD} \circ 1}
}
\end{equation}
commutative, see Proposition~\ref{st:duality data and monoidal functors}.
For arbitrary other (non-twisted) involutions
$(op_{\scC},\gamma_{\scC})$ and  $(op_{\scD},\gamma_{\scD})$, such a structure is not automatic; Proposition~\ref{st:duality data and monoidal functors} relies crucially on the triangle identities for duality data.

\begin{definition}
\label{def: hofpt str on general functor}
Let $\scC$ and $\scD$ be symmetric monoidal categories endowed with involutions $(op_\scC, \gamma_\scC)$ and $(op_\scD, \gamma_\scD)$, respectively.
Let $F \colon \scC \to \scD$ be a symmetric monoidal functor.
A monoidal natural isomorphism 
\begin{equation*}
\alpha \colon F \circ op_\scC \to op_\scD \circ F
\end{equation*}
such that the diagram
\begin{equation}
\xymatrix@C=2cm{F \circ op_\scC^2 \ar[r]^-{\alpha \circ 1_{op_\scC}} \ar[d]_{1_F \circ \gamma_\scC} & op_\scD \circ F \circ op_\scC \ar[d]^{1_{op_\scD} \circ \alpha}
        \\
        F & op_\scD^2 \circ F \ar[l]^{\gamma_\scD \circ 1_F}
}
\end{equation}
commutes is called a \emph{homotopy fixed point structure on} $F$.
\end{definition}

Given a homotopy fixed point structure $\alpha$ on a functor $F$ and an object of $x$ of $\scC$ is equipped with a (pre-)hermitean structure $\flat_x:op_{\scC}(x) \to d_{\scC}(x)$, setting
\begin{equation}
\label{eq:hermiteanstructureimage}
\flat_{F,x} := \beta_x \circ F(\flat_x) \circ \alpha^{-1}_{x}\text{,}
\end{equation}
defines a (pre-)hermitean structure on the image $F(x)$.
Here $\beta$ is the canonical natural isomorphism \eqref{eq:canonicaldualityiso}.
We have the following result.
\begin{lemma}
\label{lem:isometricimage}
Suppose $f:x \to y$ is an isometric morphism between objects with hermitean structures $\flat_x$ and $\flat_y$, respectively. Then, $F(f)$ is isometric with respect to the hermitean structures $\flat_{F,x}$ and $\flat_{F,y}$ defined by \eqref{eq:hermiteanstructureimage}. 
\end{lemma}

\begin{proof}
For $x,y \in \scC$, let $f_{x,y} \colon Fx \otimes Fy \to F(x \otimes y)$ denote the isomorphism that renders $F$ a symmetric monoidal functor.
For simplicity, we denote again the involutions $op_{\scC}$ and $op_{\scD}$ by $x \mapsto \overline{x}$. Then, we calculate
\begin{align*}
h_{F(y)}\circ (\overline{F(f)} \otimes F(f)) 
&=\ev_{F(y)}\circ \big( \big( \beta_y \circ F(\flat_y)\circ \alpha_y^{-1} \circ \overline{F(f)} \big) \otimes F(f) \big) 
\\&=\ev_{F(y)}\circ ( \beta_y \otimes 1_y)  \circ (F(\flat_y) \otimes 1) \circ  (F(\overline{f})  \otimes F(f))\circ (\alpha_x^{-1} \otimes 1) 
\\&=F(\ev_{y}) \circ f_{y^\vee,y} \circ (F(\flat_y) \otimes 1) \circ  (F(\overline{f})  \otimes F(f))\circ (\alpha_x^{-1} \otimes 1) 
\\&=F(\ev_{x}) \circ f_{x^\vee,x} \circ (F(\flat_x) \otimes 1)\circ (\alpha_x^{-1} \otimes 1)
\\&=\ev_{F(x)} \circ ( \beta_x \circ F(\flat_x) \circ \alpha^{-1}_{x}\otimes 1)
\\&=h_{F(x)}
\end{align*}
Here we have used the naturality of $\alpha$, the fact that $(F,f)$ is a monoidal functor, the compatibility of $\beta_x$ with the evaluations $\ev_x$ and $\ev_{Fx}$, see Proposition~\ref{st:duality data and monoidal functors}. 
\end{proof}

It is again straightforward to generalise the definition of a homotopy fixed point structure to morphisms of presheaves of symmetric monoidal categories. Applying this to OCFFTs we obtain the following definition, which is an adaption of the formalism of \cite{FH:Reflection_positivity} from FFTs to  smooth FFTs, see \cite[Def.~4.17]{FH:Reflection_positivity}.

\begin{definition}
\label{def:reflection structure}
A \emph{reflection structure} on a smooth OCFFT $\scZ \colon \OCBord_d(M,Q) \to \VBdl$  is a homotopy fixed point structure $\alpha$ on $\scZ$ with respect to the involutions $op_\OCBord$ and $op_\VBdl$.
\end{definition}

We recall that every object $(Y,f)$ of $\OCBord_d(M,Q)(U)$, for every $U \in \Cart$, carries a canonical hermitean structure \eqref{eq:canonical hermitean structure}. Hence, by~\eqref{eq:hermiteanstructureimage} we obtain hermitean structures $\flat_{\scZ,(Y,f)}$ on the vector bundles $\scZ(Y,f)$ over $U$. As explained in Example~\ref{ex:fixed duals in Vect}, these are nothing but hermitean bundle metrics in the ordinary sense.
Now the following definition makes sense:

\begin{definition}
\label{def:positive reflection structure}
Let $\scZ \colon \OCBord_d(M,Q) \to \VBdl$ be a smooth OCFFT.
A reflection structure $\alpha$ on $\scZ$ is called \emph{positive} if for every cartesian space $U \in \Cart$ and every object $(Y,f) \in \OCBord_d(M,Q)(U)$ the hermitean structure $\flat_{\scZ,(Y,f)}$ on the vector bundle $\scZ(Y,f)$ over $U$  is positive definite. A smooth OCFFT with a positive reflection structure will be called a \emph{reflection-positive smooth OCFFT}. 
\end{definition}

Reflection-positive $d$-dimensional smooth OCFFTs come naturally organised in a 2-category $\mathrm{RP}\text{-}\OCFFT_{d}(M,Q)$, with full subcategories  of invertible, superficial, thin homotopy invariant, and homotopy invariant  smooth OCFFTs.

We close this section with the following result, obtained by combining Lemma~\ref{lem:isometricimage} with Propositions~\ref{prop:thinhomotopyinvariantocfft} and~\ref{st:path bordisms are unitary}. We plan to use this result in future work.   

\begin{proposition}
Let $\scZ$ be a reflection-positive, thin homotopy invariant smooth OCFFT. Then, the image of any path bordism under $\scZ$ is a unitary vector bundle morphism. 
\end{proposition}

\addtocontents{toc}{\protect\pagebreak}
\section{Smooth OCFFTs from B-fields and D-branes}
\label{sect:TFT_construction}

In this section, we explicitly construct a 2-dimensional smooth OCFFT over a target space $(M,Q)$ from a target space brane geometry $(\CG, \CE) \in \TBG(M,Q)$ as defined in Section~\ref{sect:BGrbs}.
This construction employs the the coherent transgression vector bundles $\widehat{\scL}$ and $\widehat{\scR}_{ij}$ from Section~\ref{sect:background} and the surface amplitude $\CA^{\CG,\CE}$ from Section~\ref{sect:amplitudes}. We equip our smooth OCFFT with a positive reflection structure, and show that it depends functorially on the target space brane geometry.

\subsection{From surface amplitudes to smooth OCFFTs}
\label{sect:TFT_definition}

To begin with, we consider a cartesian space $U \in \Cart$ and an object $(Y, f) \in \OCBord_2(M,Q)(U)$.
For simplicity, let us first assume that $Y \cong [0,1]$, with brane labels $i,j \in I$ assigned to its initial and end point, respectively.
The map $f \colon U \times Y \to M$ is the adjoint of a plot $f^{\vdash} \colon U \to P^Y_{ij}M$ defined by $f^{\vdash}(x)(y) \coloneqq f(x,y)$ for  $x\in U$ and $y\in Y$, i.e.~$f=(f^{\vdash})^{\dashv}$.
Analogously, if $Y \cong \bbS^1$, we obtain a plot $f^{\vdash}\colon U \to L^Y M$.
Let $\sfU \colon \HVBdl^\nabla \to \VBdl$ be the morphism of sheaves of symmetric monoidal categories that forgets hermitean metrics and connections.
We set
\begin{equation}
\label{eq:Z_G on objects}
        \scZ_{\CG,\CE}(Y, f) \coloneqq
        \begin{cases}
                \sfU \big( (f^{\vdash})^{*} \scR^Y_{ij} \big)\ , & Y \cong [0,1]\,,
                \\[0.2cm]
                \sfU \big( (f^{\vdash})^{*} \scL^Y \big)\ , & Y \cong \bbS^1\,.
        \end{cases}
\end{equation}
We extend this assignment to empty $Y$ by setting
\begin{equation}
        \scZ_{\CG,\CE}(Y, f) \coloneqq U \times \FC
\end{equation}
whenever $(Y,f) \in \OCBord_2(M,Q)(U)$ with $Y = \emptyset$, and to non-connected $Y$ by mapping families of disjoint unions of manifolds to tensor products of vector bundles.

We now use the amplitudes defined in Section~\ref{sect:amplitudes} in order to define the field theory $\scZ_{\CG,\CE}$ on morphisms.
Let $[\Sigma, \sigma] \colon (Y_0,f_0) \to (Y_1,f_1)$ be a morphism in $\OCBord_2(M,Q)(U)$.
The two objects $(Y_a,f_a)$ define vector bundles $\scZ_{\CG,\CE}(Y_a, f_a)$ over $U$, for $a = 0,1$, via ~\eqref{eq:Z_G on objects}.
Let $\Psi_0 \in \Gamma(U, \scZ_{\CG,\CE}(Y_0, f_0))$ and $\Psi_1^\vee \in \Gamma(U, \scZ_{\CG,\CE}(Y_1, f_1)^\vee)$ be arbitrary smooth sections.
We think of $\Psi_0$ as a smoothly parameterised family of incoming states and of $\Psi_1^\vee$ as a smoothly parameterised family of duals of outgoing states. 
\begin{remark}
\label{rmk: use of beta}
Recall that under the canonical isomorphisms~\eqref{eq:beta_ij} and~\eqref{eq:varrho mp}, and their coherent extensions~\eqref{eq:whbeta whvarrho}, we can canonically identify the vector bundles $\scZ_{\CG,\CE}((Y,f)^\vee)$ and $(\scZ_{\CG,\CE} (Y,f))^\vee$.
We will use this identification throughout.
\qen
\end{remark}

The following lemma is an immediate consequence of our definition of scattering diagrams (Definition~\ref{def:parameterised_scattering_digram}) and morphisms in $\OCBord_d(M,Q)$ (Definition~\ref{def: family of bordisms}).

\begin{lemma}
Let $x\in U$. 
\begin{itemize}
\item 
If $(\Sigma,\sigma)$ is a representative for the morphism $[\Sigma,\sigma]$ under the equivalence relation  \eqref{eq:equivalence in OCBord_1}, then $(\Sigma, \sigma_{|\{x\} \times \Sigma}, \Psi_{1|x}^\vee, \Psi_{0|x})$ is a  scattering diagram. 

\item
If $(\Sigma',\sigma')$ is another representative, then $(\Sigma', \sigma'_{|\{x\} \times \Sigma}, \Psi_{1|x}^\vee, \Psi_{0|x})$ is an equivalent  scattering diagram.

\end{itemize} 
\end{lemma}

Since the surface amplitude is well-defined on equivalence classes of scattering diagrams by Proposition \ref{prop:surfaceamplitude}, we obtain a well-defined number
\begin{equation*}
\CA^{\CG,\CE} \big[ \Sigma, \sigma_{|\{x\} \times \Sigma}, \Psi_{1|x}^\vee, \Psi_{0|x} \big] \in \mathbb{C}
\end{equation*}  
for every $x\in U$. We prove first that this number depends smoothly on $x$.

\begin{lemma}
\label{st:amplitudes are smooth}
The function 
\begin{equation*}
U \to \FC\,,\qquad
x \longmapsto \CA^{\CG,\CE} \big[ \Sigma, \sigma_{|\{x\} \times \Sigma}, \Psi_{1|x}^\vee, \Psi_{0|x} \big] 
\end{equation*}
is smooth.
\end{lemma}

\begin{proof}
Since $U$ is diffeomorphic to $\FR^n$ for some $n \in \NN_0$ and since $\Sigma$ is 2-dimensional, we find a trivialisation $\CT \colon \sigma^*\CG \to \CI_\rho$.
By the coherence of $\widehat{\scL}$ and $\widehat{\scR}_{ij}$, after choosing parameterisations of the connected components of $Y_a$ (i.e.~orientation-preserving diffeomorphisms to either $[0,1]$ or $\bbS^1$), any smooth section $\Psi_a$ of $\scZ_{\CG,\CE}(Y_a, f_a)$ over $U$ can be written in a unique way as a $C^\infty(U, \FC)$-linear combination of tensor products of sections of the form $[[\CT], z_{cl_a}]$ for smooth functions $z_{cl_a} \colon U \to \FC$, where $cl_a$ runs over the connected components of $Y_a$ that are diffeomorphic to $\bbS^1$, and sections of the form $[\CT, \psi_{e_a}]$, where
\begin{equation}
        \psi_{e_a} \in \Hom \big( \Delta(\CE_{i(e_{a,0})|U \times e_{a,0}},\, \CT_{|U \times e_{a,0}}),\, \Delta(\CE_{i(e_{a,1})|U \times e_{a,1}},\, \CT_{|U \times e_{a,1}}) \big)\,.
\end{equation}
Here, $e_a$ runs over the connected components of $Y_a$ that are diffeomorphic to $[0,1]$.
For the sake of legibility we have omitted pullbacks, and we have written $e_{a,0}$ for the initial point of the oriented component of $Y_a$ labelled by $e_a$ and $e_{a,1}$ for its end point.
The indices $i(e_{a,0})$ and $i(e_{a,1}) \in I$ are the D-brane indices of the respective corners of $\Sigma$ (see Definition~\ref{def: family of objects}).
The reason that tensor products of these sections generate all smooth sections over $U$ under $C^\infty(U, \FC)$-linear combinations is that, by construction of the bundles $\scR_{ij}$ and $\scL$, these sections generate each fibre of their respective pullbacks (cf. Section~\ref{sect:bundles on path spaces}) that appear in $\scZ_{\CG,\CE}(Y_0,f_0)$ and $\scZ_{\CE,\CG}(Y_1,f_1)$ (cf.~\eqref{eq:Z_G on objects}).
From the explicit construction of the surface amplitude $\CA^{\CG,\CE}$ in Definition~\ref{def:parameterised_amplitude}, one can now see the claimed smoothness.
\end{proof}

Now, we define a morphism
\begin{equation}
                \scZ_{\CG,\CE}[\Sigma,\sigma] \colon \scZ_{\CG,\CE}(Y_0, f_0) \longrightarrow \scZ_{\CG,\CE}(Y_1, f_1)
\end{equation}
of vector bundles over $U$ by requiring that
\begin{equation}
\label{eq:smooth amplitude}
        \Psi_{1|x}^\vee \big( \scZ[\Sigma, \sigma]_{|x}\, (\Psi_{0|x} )\big)
        = \CA^{\CG,\CE} \big[ \Sigma, \sigma_{|\{x\} \times \Sigma}, \Psi_{1|x}^\vee, \Psi_{0|x} \big]
\end{equation}
for all $x \in U$ and for all smooth sections $\Psi_0 \in \Gamma(U,\scZ_{\CG,\CE}(Y_0,f_0))$ and $\Psi_1^\vee \in \Gamma(U, \scZ_{\CG,\CE}(Y_1, f_1)^\vee)$.
Since the evaluation pairing between a finite-dimensional vector space and its dual is non-degenerate, the expression~\eqref{eq:smooth amplitude}  defines a bundle morphism $\scZ_{\CG,\CE}[\Sigma, \sigma]$ which is smooth by Lemma \ref{st:amplitudes are smooth}.

\begin{theorem}
\label{st:BGrb plus D-branes yields OC TFT}
Let $(M,Q)$ be a target space.
For any target space brane geometry $(\CG, \CE) \in \TBG(M,Q)$, Equations~\eqref{eq:Z_G on objects} and~\eqref{eq:smooth amplitude} define a 2-dimensional, invertible, superficial, smooth OCFFT $\scZ_{\CG,\CE}$ on $(M, \scQ)$.
\end{theorem}

We will decompose the proof of Theorem~\ref{st:BGrb plus D-branes yields OC TFT} into a series of smaller assertions.
The compatibility with symmetric monoidal structures has been built into the definition of $\scZ_{\CG,\CE}$ on objects.
On morphisms it follows from the fact that the surface amplitude \smash{$\CA^{\CG,\CE}$} factorises accordingly for disjoint unions of surfaces, as can readily be seen from its definition (see Section~\ref{sect:amplitudes-definition}).
The main part is now to prove that $\scZ_{\CG,\CE}$ is a morphism of presheaves of categories on $\Cart$. By construction, the pullbacks of objects in $\OCBord_2(M,Q)(U)$ along morphisms $V \to U$ of cartesian spaces get mapped to pullbacks of vector bundles, and the pullback of morphisms pulls back due to its fibre-wise definition \eqref{eq:smooth amplitude}.  
Thus, it suffices to show the functoriality of $\scZ_{\CG,\CE}$ pointwise, i.e.\ we only have to check that
\begin{equation}
        \scZ_{\CG,\CE}(\pt) \colon \OCBord_2(M,Q)(\pt) \longrightarrow \VBdl(\pt) \cong \Vect
\end{equation}
is a functor.

\begin{lemma}
$\scZ_{\CG,\CE}(\pt)$ preserves identity morphisms.
\end{lemma}

\begin{proof}
Since $\scZ_{\CG,\CE}$ is respects the monoidal structure by construction, we can restrict ourselves to connected $Y$.
Assume further that $(Y,f)$ is an object with $\ori_\FR = +$; the case for $\ori_\FR = -$ is analogous.
Recall from~\eqref{eq:id bordism} that, for any object $(Y,f) \in \OCBord_2(M,Q)(\pt)$, we have
\begin{align}
        1_{(Y,f)} &= \big[ (Y,f), (Y,f),\, [0,1] \times Y,\, [0,\epsilon) \times Y,\, (-\epsilon,0] \times Y,
        \\*
        &\qquad \iota_{[0,\epsilon)} \times 1_Y,\, (sh_1 \circ \iota_{(-\epsilon, 0]}) \times 1_Y,\, f \circ \pr_Y,\, b \circ \pr_Y \big]\,.
\end{align}
Let us assume that $Y \cong [0,1]$.
We choose a parameterisation $\phi \colon [0,1] \to Y$ and a trivialisation $\CT_0 \colon (f \circ \phi)^*\CG \to \CI_0$ on $[0,1]$.
Because of the particularly simple form of the identity bordism, we can use the pullback of $\CT_0$ along the projection $[0,1]^2 \to [0,1]$ to trivialise the pullback of $\CG$ to $[0,1]^2 \cong [0,1] \times Y$.
For this choice of trivialisation, the surface integral in the amplitude $\CA^{\CG,\CE}([[0,1] \times Y,f \circ \pr_Y], (\psi_1^\vee, \psi_0))$ is trivial for any pair of incoming states $\psi_0$ and duals of outgoing states $\psi_1^\vee$.
Further, the incoming and outgoing string boundaries of $1_{(Y,f)}$ agree as manifolds with maps to the target space $M$.
Consequently, the amplitude is just the evaluation pairing between the finite-dimensional vector space $\scZ_{\CG,\CE|\pt}(Y,f)$ and its dual, and thus defines the identity operator
\begin{equation}
        \scZ_{\CG,\CE|\pt} (1_{(Y,f)}) = 1_{\scZ_{\CG,\CE|\pt}(Y,f)}\,,
\end{equation}
as claimed.
A completely analogous argument applies for $Y \cong \bbS^1$.
The argument then extends to non-connected compact 1-manifolds, and thus to all objects $(Y,f) \in \OCBord_2(M,Q)(*)$, since $\scZ$ is symmetric monoidal.
\end{proof}

\begin{lemma}
$\scZ_{\CG,\CE}(\pt)$ respects the composition of bordisms.
\end{lemma}

\begin{proof}
Let $[\Sigma, \sigma] \colon (Y_0, f_0) \to (Y_1, f_1)$ and $[\Sigma', \sigma'] \colon (Y_1, f_1) \to (Y_2, f_2)$ be two composable bordisms in $\OCBord_2(M,Q)(\pt)$.
Recall that their composition reads as
\begin{equation}
        [\Sigma', \sigma'] \circ [\Sigma, \sigma]
        \coloneqq \big[ (Y_0, f_0), (Y_2, f_2), \Sigma'', W_0, W_2, \tilde{w}_0, \tilde{w}_2, \sigma'' \big] \ \in \OCBord_2(M,Q)(\pt)\,.
\end{equation}
Here, $\Sigma'' \coloneqq \Sigma \cup_{Y_1} \Sigma'$, $\sigma'' \coloneqq \sigma \cup_{f_1} \sigma'$, and $\tilde{w}_a$ are the collar maps $W_a \to \Sigma''$ canonically induced by $w_a$ for $a = 0,2$.
Choose a trivialisation $\CT \colon \sigma^{\prime \prime *}\CG \to \CI_\rho$ over $\Sigma''$.
For $a = 0,1,2$, vectors in $\scZ_{\CG,\CE}(Y_a,f_a)$ are linear combinations of tensor products of elements of \smash{$\scL^{Y_{a,cl_a}}$} and \smash{$\scR^{Y_{a,e_a}}_{e_a}$} for the respective D-branes for $\CG$ (cf. Section~\ref{sect:amplitudes}), in the notation of the proof of Lemma~\ref{st:amplitudes are smooth}.
Consider an arbitrary incoming state $\psi_0 \in \scZ_{\CG,\CE}(Y_0, f_0)$, dual outgoing state $\psi_2^\vee \in (\scZ_{\CG,\CE}(Y_2, f_2))^\vee$, and parameterisations $\phi_a$ of the manifolds $Y_a$ for $a=0,1,2$.
For the amplitude of the composition we then have
\begin{equation}
\label{eq:glued amplitude}
\begin{aligned}
        \< \psi_2^\vee,\, \scZ_{\CG,\CE} \big( [\Sigma', \sigma'] \circ [\Sigma, \sigma] \big) \psi_0 \>_{\scZ_{\CG,\CE}(Y_2,f_2)}
        &= \exp \bigg( - \int_{\Sigma''} \rho \bigg)\ \prod_{c'' \in \pi_0 \Sigma''} z_{c''}
        \\
        &= \exp \bigg( - \int_{\Sigma''} \rho \bigg)\ \tr \bigg( \bigotimes_{c'' \in \pi_0 \Sigma''} \lambda_{c''} \bigg)\,,
\end{aligned}
\end{equation}
where $\lambda_{c''}$ are the compositions of morphisms in $\Vect$ associated to $c''$ in the construction of $\CA^{\CG,\CE}$ as in Definition~\ref{def:parameterised_amplitude}.
Moreover, we have used the notation $\<-,-\>_V$ for the evaluation pairing of a vector space $V$ and its dual $V^\vee$.

Every connected component $c'' \subset \partial \Sigma''$ which intersects the image of $Y_1$ in $\Sigma''$ is naturally decomposed into two parts,
\begin{equation}
        c'' = \bigg( \bigsqcup_{n=1}^{\sfn_{c''}} c'_n \bigg) \cup \bigg( \bigsqcup_{n=1}^{\sfn_{c''}} c_n \bigg)\,,
        \quad \text{where} \quad
        c'_n \subset \partial \Sigma',\ c_n \subset \partial \Sigma\,,
\end{equation}
and $c'_n$ and $c_n$ are diffeomorphic to $[0,1]$ for all $n \in \{1, \ldots, \sfn_{c''}\}$.
Note that $c''$ can only intersect the image of $Y_1$ at points belonging to the brane boundary of $\Sigma''$.
Similarly, every morphism $\lambda_{c''}$ naturally decomposes as
\begin{equation}
        \lambda_{c''} = \lambda_{c'_{\sfn_{c''}}} \circ \lambda_{c_{\sfn_{c''}}} \circ \ldots \circ \lambda_{c'_1} \circ \lambda_{c_1}\,.
\end{equation}
The reason is that $\lambda_{c''}$ is a composition of morphisms given by the morphisms $\psi_s$ from vectors $[\phi_s^*\CT, \psi_s] \in \scR_{s|f \circ \phi_s}(\CT_{|s})$ which label open components of the open string boundary of $c''$, and parallel transports in the bundles $E_b$ along brane boundary components of $c''$.
The decomposition of $\lambda_{c''}$ is then induced by decomposing the parallel transports in $E_b$ at points where $c''$ intersects the image of $Y_1$.
Consequently, we can rewrite the amplitude~\eqref{eq:glued amplitude} as
\addtocounter{equation}{1}
\begin{align*}
        &\hspace{-2em}\< \psi_2^\vee,\, \scZ_{\CG,\CE} \big( [\Sigma', \sigma'] \circ [\Sigma, \sigma] \big) \psi_0 \>_{\scZ_{\CG,\CE}(Y_2,f_2)}
        \\[0.2cm]
        &= \exp \bigg( \int_{\Sigma''} -\rho \bigg)\
        \tr \bigg( \bigotimes_{c''\, \cap Y_1 = \emptyset} \lambda_{c''}\ \otimes \bigotimes_{c''\, \cap Y_1 \neq \emptyset} \big( \lambda_{c'_{\sfn_{c''}}} \circ \lambda_{c_{\sfn_{c''}}} \circ \ldots \circ \lambda_{c'_1} \circ \lambda_{c_1} \big) \bigg)
        \\[0.2cm]
        &= \exp \bigg( \int_{\Sigma'} -\rho \bigg)\ \exp \bigg( \int_{\Sigma} -\rho \bigg) \tag{\theequation}
        \\
        &\hspace{1cm} \tr \bigg( \bigotimes_{c''\, \cap Y_1 = \emptyset} \lambda_{c''}\ \otimes \bigotimes_{c''\, \cap Y_1 \neq \emptyset} \big( \lambda_{c'_{\sfn_{c''}}} \circ \One \circ \lambda_{c_{\sfn_{c''}}} \circ \One \circ \ldots \circ \lambda_{c'_1} \circ \One \circ \lambda_{c_1} \circ \One \big) \bigg)\,.
\end{align*}
Here we omitted labelling the identities by the vector spaces which they act on in order to avoid unnecessarily heavy notation.
Let $x \in c'' \cap Y_1$, and let $\iota_x \colon \pt \hookrightarrow \Sigma''$ denote its inclusion into $\Sigma''$.
Then, the vector spaces in question are
\begin{equation}
        E_x = \Delta\big( (f_1 \circ \iota_x)^*\CE_{i(x)}, \iota_x^*\CT \big)\,,
\end{equation}
or their duals, depending on orientations.
Note that these are precisely those vector spaces which constitute $\scZ_{\CG,\CE}(Y_1,f_1)$, after forgetting the inner product structure.
Choosing bases $(e_{x,\mu_x})_{\mu_x = 1, \ldots, m_x}$ in each of these vector spaces $E_x$, with dual basis elements denoted $e^\vee_{x,\mu_x}$, we can rewrite each of the identities as
\begin{equation}
        \One _{E_x} = \sum_{\mu_x = 1}^{m_x} e_{x,\mu_x} \otimes e^\vee_{x,\mu_x}\,.
\end{equation}
Thus, the trace is broken up into a sum over products whose factors are of the form $\< e_{x,\mu_x}^\vee, \lambda_{c_n} (e_{y,\mu_y})\>_{E_x}$, where $x$ and $y$ are the initial and endpoint of $c_n$, and accordingly for paths $c'_n$.

Now we  reorder these products:
we  combine the factors involving the maps $\lambda_{c_i}$, and we  separately combine the factors containing the maps $\lambda_{c'_i}$.
The two ways of grouping the factors at a connected component $s \subset Y_1$ are illustrated in the following diagram:
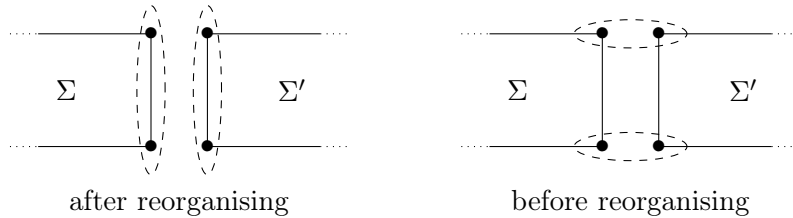
\begin{equation}
\begin{tikzpicture}[scale=0.75]
        \draw[dotted] (5,1) -- (5.5,1);
        \draw (3,1) -- (5,1);
        \draw[dotted] (5,-1) -- (5.5,-1);
        \draw (3,-1) -- (5,-1);
        
        \draw (3,1) -- (3,-1);

        \draw (2,1) -- (2,-1);
        \draw (0,1) -- (2,1);
        \draw (0,-1) -- (2,-1);
        \draw[dotted] (-.5,1) -- (0,1);
        \draw[dotted] (-.5,-1) -- (0,-1);
        
        \node at (0.5,0) {$\Sigma$};
        \node at (4.5,0) {$\Sigma'$};
        
        \node at (2,1) {$\bullet$};
        \node at (2,-1) {$\bullet$};
        \node at (3,1) {$\bullet$};
        \node at (3,-1) {$\bullet$};
        
        \draw[dashed] (2,0) ellipse (0.25cm and 1.5cm);
        \draw[dashed] (3,0) ellipse (0.25cm and 1.5cm);

        \draw[dotted] (13,1) -- (13.5,1);
        \draw (11,1) -- (13,1);
        \draw[dotted] (13,-1) -- (13.5,-1);
        \draw (11,-1) -- (13,-1);
        
        \draw (11,1) -- (11,-1);
        
        \draw (10,1) -- (10,-1);
        \draw (8,1) -- (10,1);
        \draw (8,-1) -- (10,-1);
        \draw[dotted] (7.5,1) -- (8,1);
        \draw[dotted] (7.5,-1) -- (8,-1);
        
        \node at (8.5,0) {$\Sigma$};
        \node at (12.5,0) {$\Sigma'$};

        \draw[dashed] (10.5,1) ellipse (1cm and 0.25cm);
        \draw[dashed] (10.5,-1) ellipse (1cm and 0.25cm);
        
        \node at (10,1) {$\bullet$};
        \node at (10,-1) {$\bullet$};
        \node at (11,1) {$\bullet$};
        \node at (11,-1) {$\bullet$};

        \node at (2.5,-2) {after reorganising};
        \node at (10.5,-2) {before reorganising};
                
\end{tikzpicture}
\end{equation}
For $x \in c'' \cap Y_1$, let $z$ be the unique point in $\partial \Sigma''$ such that $x$ and $z$ are joined by a connected component of the image of $Y_1$ in $\Sigma''$.
Using the fact that the collection $(e_{z,\mu_z} \otimes e_{x,\mu_x}^\vee)_{\mu_x = 1,\ldots,m_x,\, \mu_z = 1,\ldots, m_z}$ forms a basis for $E_z \otimes E_x^\vee \cong \scR_{s|f \circ \phi_s}(\CT_{|s})$, we thus arrive at
\begin{align*}
        &\< \psi_2^\vee,\, \scZ_{\CG,\CE} \big( [\Sigma', \sigma'] \circ [\Sigma, \sigma] \big)\, \psi_0 \>_{\scZ_{\CG,\CE}(Y_2,f_2)}
        \\
        &= \sum_k \< \psi_2^\vee,\, \scZ_{\CG,\CE} [\Sigma', \sigma']\, \psi_{1,k} \>_{\scZ_{\CG,\CE}(Y_2,f_2)} \cdot
                \< \psi_{1,k},\, \scZ_{\CG,\CE} [\Sigma, \sigma]\, \psi_0 \>_{\scZ_{\CG,\CE}(Y_1,f_1)}\,,
\end{align*}
where $(\psi_{1,k})$ is the basis for $\scZ_{\CG,\CE}(Y_1,f_1)$ formed by using the bases $(e_{z,\mu_z} \otimes e_{x,\mu_x}^\vee)$ in the individual tensor factors comprising $\scZ_{\CG,\CE}(Y_1,f_1)$.
Now, by the non-degeneracy of the evaluation pairing, it follows that
\begin{equation}
\begin{aligned}
        \< \psi_2^\vee,\, \scZ_{\CG,\CE} \big( [\Sigma', \sigma'] \circ [\Sigma, \sigma] \big)\, \psi_0 \>_{\scZ_{\CG,\CE}(Y_2,f_2)}
        = \< \psi_2^\vee,\, \scZ_{\CG,\CE} [\Sigma', \sigma'] \circ \scZ_{\CG,\CE}[\Sigma, \sigma]\, \psi_0 \>_{\scZ_{\CG,\CE}(Y_2,f_2)}\,,
\end{aligned}
\end{equation}
for all $\psi_0 \in \scZ_{\CG,\CE}(Y_0,f_0)$ and $\psi_2^\vee \in (\scZ_{\CG,\CE}(Y_2,f_2))^\vee$, which implies the claim.
\end{proof}

To complete the proof of Theorem~\ref{st:BGrb plus D-branes yields OC TFT} we recall that the amplitude $\CA^{\CG,\CE}$ is superficial (Proposition~\ref{st:amplitude and homotopies} and Proposition~\ref{st:A^G and thin maps to M}) which readily implies that $\scZ_{\CG,\CE}$ is a superficial smooth OCFFT. Finally, invertibility is immediately build in into the definition \eqref{eq:Z_G on objects}.

\subsection{Reflection positivity}
\label{sect: Reflection positivity of Z_G}

Here we show that the field theory $\scZ_{\CG,\CE}$ carries a canonical positive reflection structure in the sense of Section~\ref{sect:variants of OCFFTS}.
It turns out that the hermitean metrics on the bundles $\scZ_{\CG,\CE}(Y,f)$ induced from the reflection structures are precisely the hermitean metrics  on the bundles $\widehat{\scR}_{ij}$ and $\widehat{\scL}$ that have been forgotten in \eqref{eq:Z_G on objects}.

Consider an object $(Y,f) \in \OCBord_2(M,Q)(U)$ over a test space $U \in \Cart$.
For simplicity, let us first assume that $Y$ is connected and open, i.e.~that $Y \cong [0,1]$.
We recall from~\eqref{eq:Z_G on objects} that $\scZ_{\CG,\CE}(Y,f) = f^{\vdash *} \scR^Y_{ij}$; here, we omit displaying the forgetful functor $\sfU$.
Consequently, we have
\begin{equation}
        (\scZ_{\CG,\CE} \circ op_{\OCBord})(Y,f)
        =  f^{\vdash *} \scR^{\overline{Y}}_{ji} \,.
\end{equation}
On the other hand, we have
\begin{equation*}
(op_{\VBdl} \circ \scZ_{\CG,\CE})(Y,f)=f^{\vdash *} \overline{\scR^Y_{ij}}.
\end{equation*}
We define the vector bundle isomorphism
\begin{equation}
\label{eq:reflection structure alpha}
        \alpha_{(Y,f)} \coloneqq f^{\vdash *} \widehat{\alpha}_{ij}^{-1} \colon
        f^{\vdash *} r_{ij}^* \scR^{\overline{Y}}_{ji}
        \longrightarrow  f^{\vdash *} \overline{\scR^Y_{ij}} 
\end{equation}
over $U$,
with $r_{ij}$ as  in~\eqref{eq:def alpha hat} and $\widehat{\alpha}_{ij}$ as  in~\eqref{eq:def alpha widehat}; see Section~\ref{sect:pull-push and coherence} for details.
We show that this  yields a natural isomorphism
\begin{equation}
        \alpha \colon \scZ_{\CG,\CE} \circ op_{\OCBord_2} \longrightarrow op_{\VBdl} \circ \scZ_{\CG,\CE}\,.
\end{equation}
Consider a morphism $[\Sigma, \sigma]$ in $\OCBord_2(M,Q)(U)$ from $(Y_0,f_0)$ to $(Y_1, f_1)$.
Let $\Psi_0 \in \Gamma(U, \scZ_{\CG,\CE}(Y_0,f_0))$ be a smooth family of incoming states, and let $\Psi^\vee_1 \in \Gamma(U, \scZ_{\CG,\CE}(Y_1,f_1)^\vee)$ be a smooth family of outgoing dual states.
We consider the evaluation
\begin{align}
        &\hspace{-2em}\< \Psi^\vee_1,\, \alpha_{(Y_1,f_1)} \circ \big( \scZ_{\CG,\CE} \circ op_{\OCBord} [\Sigma, \sigma] \big) \circ \alpha_{(Y_0,f_0)}^{-1} (\Psi_0) \>_{\overline{\scZ_{\CG,\CE} (Y_1,f_1)}}
        \\*[0.1cm]
        &= \< \alpha_{(Y_1,f_1)}^\vee (\Psi^\vee_1),\, \big( \scZ_{\CG,\CE} \circ op_{\OCBord} [\Sigma, \sigma] \big) \circ \alpha_{(Y_0,f_0)}^{-1} (\Psi_0) \>_{\scZ_{\CG,\CE} \overline{(Y_1,f_1)}}
        \\[0.1cm]
        &= \overline{\< \Psi^\vee_1,\, \big( \scZ_{\CG,\CE} [\Sigma, \sigma] \big) (\Psi_0) \>_{\scZ_{\CG,\CE} (Y_1,f_1)}}
        \\*[0.1cm]
        &= \< \Psi^\vee_1,\, \big( op_\VBdl \circ \scZ_{\CG,\CE} [\Sigma, \sigma] \big) (\Psi_0) \>_{\overline{\scZ_{\CG,\CE} (Y_1,f_1)}}\,.
\end{align}
In the first step we have used nothing but the definition of the transpose of a morphism of vector bundles.
The second step is seen from the explicit construction of the amplitude $\CA^{\CG,\CE}$ in Section~\ref{sect:amplitudes-definition}:
using the coherent structure $\widehat{R}$ on $\widehat{\scL}$ and $\widehat{\scR}_{ij}$ to pass to the case where the connected components of $Y$ are given by $[0,1]$ and $\bbS^1$ (i.e.~by choosing parameterisations), the amplitude in the second line has all faces decorated with the hermitean adjoints of the linear maps that occur in the amplitude in the third line, and the integral is taken over $\overline{\Sigma}$.
The last identity is just making use of the relation between evaluation on a complex vector space $V$ and on its complex conjugate,
\begin{equation}
        \< v^\vee, w \>_{\overline{V}} = \overline{\<v^\vee,w \>_V}
        \qquad \forall\, v^\vee \in V^\vee,\, w \in V\,.
\end{equation}
This proves that $\alpha$ is natural in $(Y,f)$.
Furthermore, we see that $\alpha$ satisfies the coherence axiom from Definition~\ref{def: hofpt str on general functor} since $op_\OCBord^2 = 1_\OCBord$ and $op_\VBdl^2 = 1_\VBdl$, and since with respect to any choice of parameterisations of $Y$ the morphism $\alpha$ acts on tensor factors as the morphism $\alpha_{ij}$ from \eqref{def:alpha}; one can readily see that $\alpha_{ji} \circ \alpha_{ij} = 1$.
This proves

\begin{proposition}
\label{prop:reflectionstructure}
For any target space brane geometry $(\CG, \CE) \in \TBG(M,\scQ)$, Equation~\eqref{eq:reflection structure alpha} defines a reflection structure $\alpha$ on $\scZ_{\CG,\CE}$.
\end{proposition}

We recall from Section~\ref{sect:variants of OCFFTS} that for any test space $U \in \Cart$, we have given choices of fixed duals on both $\OCBord_d(M,Q)(U)$ and on $\VBdl(U)$, which are compatible with pullbacks along smooth maps $V \to U$ of cartesian spaces.
Such a choice of fixed duality data canonically defines duality functors (see Appendix~\ref{app:duals in monoidal cats}).
Further, Proposition~\ref{st:duality data and monoidal functors} implies that for any morphism $\scZ \colon \OCBord_d(M,Q) \to \VBdl$ the fixed duals induce a canonical monoidal natural isomorphism
\begin{equation}
        \beta \colon \scZ \circ d_{\OCBord} \rightarrow d_{\VBdl} \circ \scZ\,.
\end{equation}
For $Y \cong [0,1]$, this is the isomorphism
\begin{equation}
\label{eq:duality beta}
        \beta_{(Y,f)} \colon f^{\vdash*} \rev^* \scR_{ji}^{\overline{Y}}
        \longrightarrow
        \big( f^{\vdash*} \scR_{ij}^Y \big)^\vee\,,
\end{equation}
defined in~\eqref{eq:evaluation of beta_ij} and~\eqref{eq:beta_ij}.
If $Y \cong \bbS^1$ instead, $\beta_{(Y,f)}$ is given by the isomorphism $\tilde{\varrho}$ defined in~\eqref{eq:varrho mp}.
We extend $\beta$ to general manifolds $Y$ by sending disjoint unions to tensor products.
Using the canonical identification of a finite-dimensional vector space with its double dual, we have
\begin{align}
        &\hspace{-2em}\< \Psi_0,\, \beta_{(Y_0,f_0)} \circ \big( \scZ_{\CG,\CE} \circ d_{\OCBord} [\Sigma, \sigma] \big) \circ \beta_{(Y_1,f_1)}^{-1} (\Psi^\vee_1) \>_{(\scZ_{\CG,\CE} (Y_0,f_0))^\vee}
        \\*[0.1cm]
        &= \< \beta_{(Y_0,f_0)}^\vee (\Psi_0),\, \big( \scZ_{\CG,\CE} \circ d_{\OCBord} [\Sigma, \sigma] \big) \circ \beta_{(Y_1,f_1)}^{-1} (\Psi^\vee_1) \>_{\scZ_{\CG,\CE} (Y_0,f_0)^\vee}
        \\[0.1cm]
        &= \< \beta_{(Y_0,f_0)^\vee} (\Psi_0),\, \big( \scZ_{\CG,\CE} \circ d_{\OCBord} [\Sigma, \sigma] \big) \circ \beta_{(Y_1,f_1)}^{-1} (\Psi^\vee_1) \>_{\scZ_{\CG,\CE} (Y_0,f_0)^\vee}
        \\[0.1cm]
        &= \< \Psi^\vee_1,\, \scZ_{\CG,\CE} [\Sigma, \sigma] (\Psi_0) \>_{\scZ_{\CG,\CE} (Y_1,f_1)}
        \\[0.1cm]
        &= \< \big( \scZ_{\CG,\CE} [\Sigma, \sigma] \big)^\vee (\Psi^\vee_1),\, \Psi_0 \>_{\scZ_{\CG,\CE} (Y_0,f_0)}
        \\*[0.1cm]
        &= \< \Psi_0,\, \big( d_{\VBdl} \circ \scZ_{\CG,\CE} [\Sigma, \sigma] \big) (\Psi^\vee_1) \>_{(\scZ_{\CG,\CE} (Y_0,f_0))^\vee}\,.
\end{align}
Here, the first identity is just transposing $\beta_{(Y_0,f_0)}$.
The second identity uses that $\beta_{(Y,f)}^\vee = \beta_{(Y,f)^\vee}$ as morphisms $\scZ_{\CG,\CE} (Y,f) \to (\scZ_{\CG,\CE}(Y,f)^\vee)^\vee$; this follows right from the definition of $\beta_{ij}$ and $\tilde{\varrho}$ in Section~\ref{sect:equivar strs on R_ij} and Section~\ref{sect:LBdl_on_loop_space}, respectively.
The third identity follows from the explicit form of the amplitude $\CA^{\CG,\CE}$ (cf.~\eqref{eq:surface_amplitude}, again using the coherence $\widehat{R}$ to pass to the case where the connected components of $Y_i$ are copies of $[0,1]$ and of $\bbS^1$):
the exponential factors agree since the manifolds and maps to target space of the morphisms $[\Sigma, \sigma]$ and $[\Sigma, \sigma]^\vee$ agree.
Moreover, the insertions of $\beta$ and $\beta^{-1}$, together with the use of $\beta$ to define $\Psi_1^\vee$ (cf.~Remark~\ref{rmk: use of beta}), implies that the factors $\lambda_i$ in the amplitudes in the third and the fourth lines agree.
The last two steps are transposing the morphism $\scZ_{\CG,\CE}[\Sigma, \sigma]$ of vector bundles and rewriting the transpose as the dual of a morphism of vector bundles.

\begin{proposition}
\label{st:coincidence of hermitean structures}
Let $(\CG, \CE) \in \TBG(M,Q)$. 
For any test space $U \in \Cart$ and any object $(Y,f) \in \OCBord_2(M,Q)(U)$, the hermitean structure $\flat_{\scZ_{\CG,\CE},(Y,f)}$ on the vector bundle $\scZ_{\CG,\CE}(Y,f)$ over $U$ as defined in~\eqref{eq:hermiteanstructureimage} agrees with the hermitean structure on $\scR^Y_{ij}$ and $\scL^Y$ obtained in Section~\ref{sect:bundles on path spaces}.
\end{proposition}

\begin{proof}
Recall the definition of $\flat_{(Y,f)}$ form~\eqref{eq:canonical hermitean structure}.
By~\eqref{eq:hermiteanstructureimage} we have that
\begin{equation}
        \flat_{\scZ_{\CG,\CE} ,(Y,f)} = \beta_{(Y,f)} \circ \scZ_{\CG,\CE} (\flat_{(Y,f)}) \circ \alpha_{(Y,f)}^{-1}
        \colon \quad \overline{\scZ_{\CG,\CE}(Y,f)} \longrightarrow \big( \scZ_{\CG,\CE} (Y,f) \big)^\vee\,.
\end{equation}
Observe that because of $\sigma = f \circ \pr_Y$ we have $\sigma^* \curv(\CG) = 0$.
Since $U \times [0,1] \times Y$ is homotopy equivalent to $Y$, we have $\rmH^2_\dR(U \times [0,1] \times Y) = 0$, and we may therefore choose a trivialisation $\CT \colon \sigma^* \CG \to \CI_0$.
We let $\Psi_a \in \Gamma(U, \scZ_{\CG,\CE}(Y,f))$, for $a = 0,1$, be smooth sections and consider the amplitude
\begin{equation}
        \< \flat_{\scZ_{\CG,\CE} ,(Y,f)}(\Psi_1),\, \Psi_0 \>_{\scZ_{\CG,\CE}(Y,f)}
        = \< \scZ_{\CG,\CE} (\flat_{(Y,f)}) \circ \alpha_{(Y,f)}^{-1} (\Psi_1),\, \beta_{(Y,f)}^\vee(\Psi_0) \>_{\scZ_{\CG,\CE}(Y,f)}
        = \prod_{c \in \pi_0 Y} z_c\,,
\end{equation}
in the notation of Section~\ref{sect:amplitudes-definition}.
Note that the integral term in the surface amplitude is trivial here because of the special choice of trivialisation.
For every closed component $c$ of $Y$ and corresponding vectors $[[\CT_{|c}], z_{c,a}]$, we obtain $z_c = \overline{z_{c,1}}\, z_{c,0}$.
For every open component $c$ of $Y$ and corresponding vectors $[\CT_{|c},\psi_{c,a}]$ we have $z_c = \tr( \psi_{c,1}^*\, \psi_{c,0})$.
Consequently,
\begin{equation}
\label{eq:induced hermitean pairing}
        \< \flat_{\scZ_{\CG,\CE} ,(Y,f)} (\Psi_1),\, \Psi_0 \>_{\scZ_{\CG,\CE}(Y,f)}
        = h_{(Y,f)}( \Psi_1, \Psi_0)\,,
\end{equation}
where $h_{(Y,f)}$ is the hermitean metric on the vector bundle $\scZ_{\CG,\CE}(Y,f)$ induced by the hermitean metrics on the transgression bundles $\scL^c$ and $\scR^c_{ij}$, for the respective connected components $c$ of $Y$.
On the other hand, the left-hand side of~\eqref{eq:induced hermitean pairing} is equal to $\tilde{h}_{(Y,f)}(\Psi_1, \Psi_0)$, where $\tilde{h}_{(Y,f)}$ is the hermitean metric on $\scZ_{\CG,\CE}(Y,f)$ defined by the hermitean structure $\flat_{\scZ(Y,f)}$.
\end{proof}

Since the hermitian metrics on $\widehat{\scR}_{ij}$ and $\widehat{\scL}$ are positive definite, we obtain the following. 

\begin{corollary}
\label{co:positivity}
The reflection structure $\alpha$ on $\scZ_{\CG,\CE}$ is positive.
\end{corollary}

\subsection{Functorial dependence on the target space brane geometry}
\label{sect:Z_(-) as a functor}

We fix a target space $(M,Q)$ and investigate the dependence of the functorial field theory $\scZ_{\CG,\CE}$ on the target space brane geometry $(\CG,\CE) \in \TBG(M,Q)$.
For a 2-category $\scC$, let $\rmh_1 \scC$ denote its homotopy category.
It has the same objects as $\scC$, and its morphisms are the 2-isomorphism classes of 1-morphisms in $\scC$.

\begin{theorem}
\label{st:Z_(-) is a functor}
The assignment $(\CG,\CE) \mapsto \scZ_{\CG,\CE}$ defined in Sections~\ref{sect:TFT_definition} and \ref{sect: Reflection positivity of Z_G} extends to a functor
\begin{equation}
        \scZ \colon \rmh_1\TBG(M,\scQ) \rightarrow \rmRP\text{-}\OCFFT_{2}^{\sf}(M,\scQ)^{\times}\,.
\end{equation}
\end{theorem}

We carry out the proof of this theorem in the remainder of this section.
Let us first define $\scZ$ on 1-morphisms $(\CA,\xi) \colon (\CG, \CE) \to (\CG', \CE')$ in $\TBG(M,Q)$ (see Section~\ref{sect:BGrbs} for the definition).
We denote the coherent transgression bundles obtained from $(\CG',\CE')$ by $\widehat{\scR'}$ and $\widehat{\scL'}$.
Consider an arbitrary test space $U \in \Cart$, and let \smash{$(Y,f) \in \OCBord_2(M,Q)(U)$}.
Since $\scZ_{\CG,\CE}$ is symmetric monoidal, we may restrict ourselves to connected manifolds $Y$.

We first consider the case where $Y \cong [0,1]$.
The field theory $\scZ_{\CG,\CE}$ sends the object $(Y,f)$ to a vector bundle $\scZ_{\CG,\CE}(Y,f)$ over $U$.
We choose a parameterisation of $Y$ and write $\hat{f}$ for the composition $U \times [0,1] \to U \times Y \to M$.
Further, let $\iota_t \colon U \hookrightarrow U \times [0,1]$, $x \mapsto (x,t)$ for $t \in [0,1]$, and write
\begin{equation}
        \partial_a \hat{f} \coloneqq \hat{f} \circ \iota_a \colon U \to M
\end{equation}
for $a = 0,1$.
We choose a trivialisation $\CT \colon \hat{f}^*\CG \to \CI_\rho$ over $U \times [0,1]$.
Then there are canonical isomorphisms (cf.~\eqref{eq:CE via composition} and~\eqref{eq:Z_G on objects})
\begin{equation}
\label{eq:Z_A on open objects}
\begin{split}
        \scZ_{\CG,\CE} (Y,f)
        & \cong \Hom \big( \Delta \big( (\partial_0 \hat{f})^*\CE_i, \iota_0^*\CT \big),\, \Delta \big( (\partial_1 \hat{f})^*\CE_j, \iota_1^*\CT \big) \big)
        \\[0.1cm]
        &\cong  \Hom \big( \sfR \big( (\partial_0 \hat{f})^*\CE_i \circ \iota_0^*\CT^{-1} \big),\, \sfR \big( (\partial_1 \hat{f})^*\CE_j \circ \iota_1^*\CT^{-1} \big) \big)\,.
\end{split}
\end{equation}
In this representation of the bundle $\scZ_{\CG,\CE} (Y,f)$, we define isomorphisms
\begin{align}
\label{eq:zeta_A,i,0}
        &\xi_{\CA,i,0} \coloneqq \sfR \big( (\partial_0 \hat{f})^*\xi_i \circ 1_{\CT^{-1}} \big)
        \colon \sfR \big( (\partial_0 \hat{f})^*\CE_i \circ \iota_0^*\CT^{-1} \big)
        \rightarrow \sfR \big( (\partial_0 \hat{f})^*\CE'_i \circ \iota_0^*\CT'{}^{-1} \big)\,,
\end{align}
and we define $\xi_{\CA,j,1}$ analogously.
For $\psi \in \scZ_{\CG,\CE}(Y,f)$ we then define smooth isomorphisms
\begin{equation}
\label{eq:Z_A,T}
        \scZ_{\CA,\CT}(\psi) \coloneqq \xi_{\CA,j,1} \circ \psi \circ \xi_{\CA,i,0}^{-1}
\end{equation}
of vector bundles over $U$.
For $Y \cong \bbS^1$ with parameterisation $\phi \colon \bbS^1 \to Y$ and $\CT \colon (f \circ \phi)^*\CG \to \CI_0$ we set
\begin{equation}
        \scZ_{(\CA,\xi)|(Y,f)} \big[ [\CT], z \big] \coloneqq \big[ [\CT \circ (f \circ \phi)^*\CA^{-1}], z \big]\,.
\end{equation}
We claim that these morphisms are compatible with changes of trivialisations and with reparameterisations of $Y$, so that they induce a smooth bundle isomorphism
\begin{equation}
        \scZ_{(\CA,\xi)|(Y,f)} \colon \scZ_{\CG,\CE}(Y,f) \rightarrow \scZ_{\CG',\CE'}(Y,f)\,.
\end{equation}
Indeed, the compatibility with reparameterisations is seen readily from the construction; acting with an orientation-preserving diffeomorphism $g \colon [0,1] \to [0,1]$ sends $\scZ_{\CA,\CT}$ to $\scZ_{\CA, g^*\CT}$.
Thus, if we can show the compatibility with changes of trivialisations $\CT$, then, by the functoriality of the coherent pull-push construction $\Psi_*\Phi^*$, this isomorphism induces an isomorphism $\scZ_{\CG,\CE} (Y,f) \to \scZ_{\CG',\CE'} (Y,f)$.

We check the compatibility with changes of trivialisations fibre-wise over points $x \in U$:
we set $\jmath_x \colon [0,1] \to U \times [0,1]$, $t \mapsto (x,t)$, and we set
\begin{equation}
        \hat{f}_x \colon [0,1] \to M\,,
        \quad
        \hat{f}_x(t) = \hat{f}(x,t)
        \qandq
        f_x \colon Y \to M\,,
        \quad
        y \mapsto f(x,y)\,.
\end{equation}
Let $\CS \colon \hat{f}^*\CG \to \CI_{\rho'}$ be another trivialisation and set
\begin{equation}
        \CT' \coloneqq \CT \circ \hat{f}^*\CA^{-1}
        \qandq
        \CS' \coloneqq \CS \circ \hat{f}^*\CA^{-1}\,.
\end{equation}
We obtain the diagram
\begin{equation}
\label{eq:Z_A and trivialisations}
\xymatrix{        & \scR_{ij|\hat{f}_x}(\jmath_x^*\CT) \ar[dd]^{r_{\CS,\CT}} \ar[rr]^{\scZ_{\CA,\CT}} \ar[dl]^{r_\CT}
        & & \scR'_{ij|\hat{f}_x}(\jmath_x^*\CT') \ar[dd]_{r_{\CS',\CT'}} \ar[dr]^{r_{\CT'}} &
        \\
        \scR_{ij|\hat{f}_x} \ar[dd]_{\widehat{R}}^{\cong} & & & & \scR'_{ij|\hat{f}_x} \ar[dd]^{\widehat{R}}_{\cong}
        \\
        & \scR_{ij|\hat{f}_x}(\jmath_x^*\CS) \ar[rr]_{\scZ_{\CA,\CS}}  \ar[ul]^{r_\CS}
        & & \scR'_{ij|\hat{f}_x}(\jmath_x^*\CS') \ar[ur]^{r_{\CS'}} &
        \\
        (f^{\vdash *} \scR^Y_{ij})_{|x} \ar@{-->}[rrrr]^{\scZ_{(\CA,\xi)|x}} & & & & (f^{\vdash *} \scR'{}^Y_{ij})_{|x}
}
\end{equation}
The right-hand triangle commutes by the construction of the fibre of \smash{$\scR'_{ij|\gamma}$} as a colimit of the $\scR_{ij|\gamma}(\CT')$ (cf.~Equation~\eqref{eq:def R_ij}).
The same argument applies to the left-hand triangle.
We are thus left to show that the upper square in~\eqref{eq:Z_A and trivialisations} commutes.
For every $x \in U$ there exists a 2-isomorphism $\zeta_x \colon \jmath_x^*\CT \to \jmath_x^*\CS$, since such 2-isomorphisms are in one-to-one correspondence with parallel unit-length section of $\jmath_x^*\Delta(\CS,\CT) \in \HLBdl^\nabla([0,1])$.
We also introduce
\begin{equation}
        \zeta'_x \coloneqq \zeta_x \circ 1_{f_x^*\CA}\,.
\end{equation}
Then we can use $\zeta_x$ and $\zeta'_x$ to make the canonical isomorphisms $r_{\CS,\CT}$ and $r_{\CS', \CT'}$ explicit (recall~\eqref{eq:def of r_(CT,CT')}).
For $\psi \in \scZ_{\CG,\CE}(Y,f)$ and omitting pullbacks to $\{x\} \times [0,1]$, we obtain the commutative diagram
\begin{equation}
\xymatrix@C=4em{\sfR (\CE'_i \circ \CT_0'{}^{-1}) \ar[rd]^{\sfR(\xi_{i,0} \circ 1)^{-1}} \ar[rrr]^{\scZ_{\CA,\CT}(\psi)} & & & \sfR (\CE'_j \circ \CT_1'{}^{-1}) \ar[ddd]^{\sfR(1 \circ \zeta_{x,1}'{}^{(-1)})^{-1}}
        \\
        & \sfR (\CE_i \circ \CT_0^{-1}) \ar[r]^{\psi} & \sfR (\CE_j \circ \CT_1^{-1}) \ar[ru]^{\sfR(\xi_{i,1} \circ 1)} \ar[d]^{\sfR(1 \circ \zeta_{x,1}^{(-1)})^{-1}} &
        \\
        & \sfR (\CE_i \circ \CS_0^{-1}) \ar[u]^{\sfR(1 \circ \zeta_{x,0}^{(-1)})} \ar[r]_{r_{\CS,\CT}(\psi)} & \sfR (\CE_j \circ \CS_1^{-1}) \ar[rd]^{\sfR(\xi_{i,1} \circ 1)} &
        \\
        \sfR (\CE'_i \circ \CS_0'{}^{-1}) \ar[ru]^{\sfR(\xi_{i,0} \circ 1)^{-1}} \ar[rrr]_{\scZ_{\CA,\CS} \circ r_{\CS,\CT}(\psi)} \ar[uuu]^{\sfR(1 \circ \zeta_{x,0}'{}^{(-1)})} & & & \sfR (\CE'_j \circ \CS_1'{}^{-1})
}
\end{equation}
Here we denote the \emph{horizontal} inverse%
\footnote{This is defined in analogy to how one defines the dual of a morphism in a monoidal category, see Appendix~\ref{app:duals in monoidal cats}, using isomorphisms $A^{-1} \circ A \to 1$ and $1 \to A \circ A^{-1}$ in place of the evaluation and coevaluation.}
of a (not necessarily invertible) 2-morphism $\psi$ by $\psi^{(-1)}$.
Observe that by~\eqref{eq:def of r_(CT,CT')} the commutativity of the outer square is the desired identity
\begin{equation}
        r_{\CS',\CT'} \circ \scZ_{\CA,\CT} = \scZ_{\CA,\CS} \circ r_{\CS,\CT}\,.
\end{equation}
First of all, the inner square is just the definition of $r_{\CS,\CT}$, so it commutes.
The top and bottom squares are the definitions of $\scZ_{\CA,\CT}$ and of $\scZ_{\CA,\CS} \circ r_{\CS,\CT}(\psi)$, respectively, whence they commute as well.
Finally, the left and right squares commute as a consequence of the interchange law in a 2-category and the fact that the 2-morphisms on the vertical and horizontal arrows act on different factors in the respective compositions of 1-morphisms.
Thus, the representatives $\scZ_{\CA,\CT}$, which depend on choices of trivialisations, induce a well-defined isomorphism
\begin{equation}
        \scZ_{(\CA,\xi)|(Y,f)} \colon \scZ_{\CG,\CE}(Y,f) \rightarrow \scZ_{\CG',\CE'}(Y,f)\,.
\end{equation}
of vector bundles on $U$.

\begin{lemma}
Let $(\CA, \xi) \colon (\CG,\CE) \to (\CG', \CE')$ be a 1-isomorphism in $\TBG(M,\scQ)$.
The vector bundle isomorphisms $\scZ_{(\CA,\xi)|(Y,f)}$ defined above form an isomorphism of smooth OCFFTs
\begin{equation}
        \scZ_{(\CA,\xi)} \colon \scZ_{\CG,\CE} \rightarrow \scZ_{\CG',\CE'}\,.
\end{equation}
\end{lemma}

\begin{proof}
In order to see that $\scZ_{(\CA,\xi)}$ is natural, we consider an arbitrary smooth family of bordisms $[\Sigma, \sigma] \colon (Y_0,f_0) \to (Y_1,f_1)$, parameterised over a cartesian space $U \in \Cart$.
Let $\Psi_0 \in \Gamma(U, \scZ_{\CG,\CE}(Y_0,f_0))$ and $\Psi_1^\vee \in \Gamma(U, \scZ_{\CG,\CE}(Y_1,f_1)^\vee)$ be smooth sections.
Choose parameterisations for $\partial_0 \Sigma$ and $\partial_1 \Sigma$, let $\CT \colon \sigma^*\CG \to \CI_\rho$ be a trivialisation, and set $\CT' \coloneqq \CT \circ \sigma^*\CA^{-1} \colon \sigma^*\CG' \to \CI_\rho$.
The value of the scattering amplitude at a point $x \in U$ reads as
\begin{equation}
\label{eq:amplitude before iso}
\begin{split}
        \< \Psi_1^\vee,\, \scZ_{\CG,\CE}[\Sigma, \sigma]\, \Psi_0 \>_{|x}
        &= \CA^{\CG,\CE} \big[\Sigma, \sigma_{| \{x\} \times \Sigma}, \Psi_{1|x}^\vee, \Psi_{0|x} \big]= \exp \bigg( - \int_{\{x\} \times \Sigma} \rho \bigg)\ \prod_{c \in \pi_0(\partial \Sigma)} z_c\,,
\end{split}
\end{equation}
where $\<-,-\>$ is the evaluation pairing on $\scZ_{\CG,\CE}(Y_1,f_1)$.
We compare the amplitude~\eqref{eq:amplitude before iso} to the amplitude
\begin{equation}
\label{eq:amplitude after iso}
\begin{split}
        &\< (\scZ_{(\CA,\xi)})_{|(Y_1,f_1)}^{-\vee} (\Psi_1^\vee),\, \scZ_{\CG',\CE'}[\Sigma, \sigma] \circ \scZ_{(\CA,\xi)|(Y_0,f_0)} (\Psi_0) \>_{|x}
        \\[0.1cm]
        &\quad = \CA^{\CG,\CE} \big[ \Sigma, \sigma_{| \{x\} \times \Sigma}, (\scZ_{(\CA,\xi)})_{|(Y_1,f_1)}^{-\vee} (\Psi_1^\vee)_{|x},\, \scZ_{(\CA,\xi)|(Y_0,f_0)} (\Psi_0)_{|x} \big]
        \\[0.1cm]
        &\quad = \exp \bigg( - \int_{\{x\} \times \Sigma} \rho \bigg)\ \prod_{c \in \pi_0(\partial \Sigma)} z'_c\,,
\end{split}
\end{equation}
where we use the evaluation pairing on $\scZ_{\CG',\CE'}(Y_1,f_1)$.
Observe that because the target bundle gerbes of the trivialisations $\CT$ and $\CT'$ agree, the exponential factors in~\eqref{eq:amplitude before iso} and~\eqref{eq:amplitude after iso} coincide.
We now go through the list (SA1) to (SA3) in Section~\ref{sect:amplitudes-definition} in order determine how the factors $z_c$ differ from the factors $z'_c$.

Regarding (SA1), observe that since on objects $(Y,f)$ with $Y \cong \bbS^1$, the morphism $\scZ_{(\CA,\xi)|(Y,f)}$ sends the element $[[\CT],z]$ to $[[\CT'],z]$, we have $z_c = z'_c$ for boundary components $c$ of this type.

For (SA2), i.e.~if $c$ is a closed loop in the brane boundary $\partial_2 \Sigma$, getting mapped to a D-brane $Q_i$, the 2-isomorphism $\xi_i$ induces an isomorphism of hermitean vector bundles with connection
\begin{equation}
        \sfR \big( (\sigma_{|c})^*\CE_i \circ \CT_{|c}^{-1} \big)
        \cong   \sfR \big( (\sigma_{|c})^*\CE'_i \circ \CT_{|c}'{}^{-1} \big)\,,
\end{equation}
as in~\eqref{eq:zeta_A,i,0}.
Since this isomorphism is connection-preserving, the traces of the holonomies of the source and target bundles around $c$ are equal; hence, $z_c = z'_c$ in this case.

Finally, for (SA3) consider an object $(Y,f)$ with $Y \cong [0,1]$, such that $U \times \partial_0 Y$ gets mapped to some $Q_i$ and such that $U \times \partial_1 Y$ gets mapped to some $Q_j$.
From~\eqref{eq:Z_A,T} we see that the trace in the expression for $z_c$ in (SA3) stays unaffected; the additional terms $\xi_{\CA,i,0}$ and $\xi_{\CA,j,1}$ from the individual factors cancel each other under the trace.

Thus, we see that the two amplitudes~\eqref{eq:amplitude before iso} and~\eqref{eq:amplitude after iso} agree factorwise, and hence that $\scZ_{(\CA,\xi)}$ is natural.
Further, it is symmetric monoidal by construction.
\end{proof}

\begin{lemma}
For every 2-isomorphism $\varphi \colon (\CA, \xi) \to (\CA', \xi')$ in $\TBG(M,Q)$ we have
\begin{equation}
        \scZ_{(\CA,\xi)} = \scZ_{(\CA', \xi')}\,.
\end{equation}
\end{lemma}

\begin{proof}
It suffices to prove the statement for the case where $Y_0,Y_1$ are disjoint unions of copies of $[0,1]$ and $\bbS^1$, since we have constructed $\scZ_{(\CA,\xi)}$ coherently with respect to parameterisations of $Y_0$ and $Y_1$.
Using the same notation as above, we additionally define $\CT'' \coloneqq \CT \circ \hat{f}^*\CA^{\prime -1}$.
First, we consider the case $Y = [0,1]$, with initial point mapped to $Q_i$ and end point mapped to $Q_j$, and show that the diagram
\begin{equation}
\xymatrix{    & \scR'_{ij|\hat{f}_x}(\jmath_x^*\CT') \ar[dd]^{r_{\CT'',\CT'}} \ar[dr]^{r_{\CT'}} &
        \\
        \scR_{ij|\hat{f}_x}(\jmath_x^*\CT) \ar[ur]^{\scZ_{\CA,\CT}} \ar[dr]_{\scZ_{\CA',\CT}} & & \scR'_{ij|\hat{f}_x}
        \\
        & \scR'_{ij|\hat{f}_x}(\jmath_x^* \CT'') \ar[ur]_{r_{\CT''}} &
}
\end{equation}
commutes.
In other words, the composition $r_{\CT'} \circ \scZ_{\CA,\CT}$ defines a morphism from $\scR_{ij|\hat{f}_x}(\jmath_x^*\CT)$ to the colimit \smash{$\scR'_{ij|\hat{f}_x}$}.
The commutativity is seen as follows.
The  connection-preserving 2-isomorphism $\varphi \colon \CA \to \CA'$ that defines the 2-isomorphism $(\CA,\xi) \to (\CA', \xi')$ induces a  connection-preserving isomorphism
\begin{equation}
        \tilde{\varphi} \coloneqq 1_\CT \circ \hat{f}^*\varphi^{-1 (-1)} \colon \CT' \to \CT''
\end{equation}
of trivialisations of $\hat{f}^*\CG$.
We use this  connection-preserving isomorphism in order to represent the isomorphism $r_{\CT'',\CT'}$ (cf.~\eqref{eq:def of r_(CT,CT')}).
For any
\begin{equation}
                \psi \in \scR_{ij|\hat{f}_x}(\CT)
                \cong \Hom \big( \sfR \big( (\partial_0 \hat{f}_x)^*\CE_i \circ \iota_0^*\jmath_x^*\CT^{-1} \big),\, \sfR \big( (\partial_1 \hat{f})^*\CE_j \circ \iota_1^*\jmath_x^*\CT^{-1} \big) \big)
\end{equation}
we obtain a commutative diagram
\begin{equation}
\xymatrix{        \sfR (\CE'_i \circ \CT_0'{}^{-1}) \ar[dr]^{\sfR(\xi_{i,0} \circ 1)^{-1}} \ar[rrr]^{\scZ_{\CA,\CT}(\psi)} & & & \sfR (\CE'_j \circ \CT_1'{}^{-1}) \ar[dd]^{\sfR(1 \circ \tilde{\varphi}_1^{(-1)})^{-1}}
        \\
        & \sfR (\CE_i \circ \CT_0^{-1}) \ar[r]^{\psi} & \sfR (\CE_j \circ \CT_1^{-1}) \ar[ur]^{\sfR(\xi_{j,1} \circ 1)}  \ar[dr]^{\sfR(\xi_{j,1}' \circ 1)}&
        \\
        \sfR (\CE'_i \circ \CT_0''{}^{-1}) \ar[ur]^{\sfR(\xi'_{i,0} \circ 1)^{-1}} \ar[rrr]_{\scZ_{\CA',\CT'}(\psi)} \ar[uu]^{\sfR(1 \circ \tilde{\varphi}_0^{(-1)})} & & & \sfR (\CE'_j \circ \CT_1''{}^{-1})
}
\end{equation}
The top and bottom squares commute by definition of $\scZ_{\CA, \CT}$ and $\scZ_{\CA',\CT}$, respectively, while the left and right triangles commute by definition of 2-morphisms in $\TBG(M,\scQ)$.
Consequently, we have that
\begin{equation}
        r_{\CT'} \circ \scZ_{\CA,\CT} = r_{\CT''} \circ \scZ_{\CA',\CT}
\end{equation}
for every choice of trivialisation $\CT$.
The morphisms $r_{\CT'} \circ \scZ_{\CA,\CT}$ for different choices of trivialisations $\CS$ of $\hat{f}^*\CG$ are compatible as we have already checked in diagram~\eqref{eq:Z_A and trivialisations}.
This implies that the compositions $r_{\CT'} \circ \scZ_{\CA,\CT}$ induce a morphism of the colimits
\begin{equation}
        \scZ_{[\CA,\xi]|(Y,f)} \colon \scZ_{\CG,\CE}(Y,f) \rightarrow \scZ_{\CG',\CE'}(Y,f)
\end{equation}
that only depends on the 2-isomorphism class $[\CA, \xi]$ in $\TBG(M,Q)$ of the 1-morphism $(\CA,\xi)$.
\end{proof}

\begin{proposition}
For every morphism $(\CA,\xi) \colon (\CG, \CE) \to (\CG', \CE')$ in $\TBG(M,\scQ)$, the morphism $\scZ_{[\CA,\xi]} \colon \scZ_{\CG,\CE} \to \scZ_{\CG', \CE'}$ intertwines the reflection structures on $\scZ_{\CG,\CE}$ and $\scZ_{\CG', \CE'}$.
\end{proposition}

\begin{proof}
This statement holds true because transgression maps morphisms in $\TBG(M,\scQ)$ to  bundle isomorphisms~\cite[Lemma~4.8.1]{BW:Transgression_and_regression_of_D-branes}, and by Proposition~\ref{st:coincidence of hermitean structures} the hermitean structure on the bundle $\scZ_{\CG,\CE}(Y,f)$ obtained from the reflection structure agrees with the hermitean structure obtained from transgression.
\end{proof}

This completes the proof of Theorem~\ref{st:Z_(-) is a functor}.

\begin{remark}
We remark that by the construction of the OCFFT $\scZ_{\CG,\CE}$, the pullback of the parallel transports on the transgression bundles over loop and path spaces agree with the images under $\scZ_{\CG,\CE}$ of the associated path bordisms.
That is, for example, if $Y = [0,1]$ and $\Gamma \colon U \times [0,1] \to P_{ij}M$ is a smooth family of smooth paths with sitting instants, this gives rise to a path bordism $[[0,1]^2, \Gamma^\dashv]$.
The pullback of the parallel transport in $\scR_{ij}$ along $\Gamma$ and the bundle morphism $\scZ_{\CG,\CE}[[0,1]^2, \Gamma^\dashv]$ agree as morphisms of vector bundles $\scZ_{\CG,\CE}(Y,\Gamma_0) \to \scZ_{\CG,\CE}(Y, \Gamma_1)$.
In fact, if we did not know the connection on the transgression bundles, this would reproduce it, and Proposition~\ref{st:path bordisms are unitary} would  imply that this connection is unitary with respect to the canonically induced hermitean structures $\flat_{\scZ_{\CG,\CE}(Y,f)}$.
\qen
\end{remark}

\section{Subsectors, flat gerbes, and TQFTs}
\label{sect:Subsectors and kinetic terms}

In this section, we investigate several variants of the field theories $\scZ_{\CG,\CE}$ constructed in Section~\ref{sect:TFT_construction}.
We consider its closed subsectors with and without branes, and we treat the special case of flat gerbes.
We conclude by showing that in the case where $M$ is a single point our formalism is compatible with the classification of open-closed topological quantum field theories (TQFTs) in terms of coloured, knowledgeable Frobenius algebras.

\subsection{Closed subsectors}

\label{sec:closedsubsector}

We consider an arbitrary target space $(M,Q)$. 
We let $\CBord_{d}(M,Q)$ denote the full sub-presheaf of $\OCBord_{d}(M,Q)$ whose objects over a test space $U \in \Cart$ are those objects \smash{$(Y,f) \in \OCBord_{d}(M,Q)(U)$} with $\partial Y = \emptyset$.
We emphasise that the bordisms in $\CBord_d(M,Q)$ are still allowed to have non-trivial brane boundary $\partial_2 \Sigma \neq \emptyset$, but now every connected component of $\partial_2 \Sigma$ is necessarily diffeomorphic to $\bbS^1$.
The involutions $op_{\OCBord_d}$ and $d_{\OCBord_d}$ restrict to involutions of $\CBord_d(M,Q)$. 

\begin{definition}
A \emph{smooth $d$-dimensional FFT on a target space $(M,\scQ)$} is a morphism
\begin{equation}
        \scZ \colon \CBord_d(M,Q) \rightarrow \VBdl
\end{equation}
of presheaves of symmetric monoidal categories on $\Cart$.
\end{definition}

Invertibility, thin homotopy invariance, superficiality, and reflection structures are defined exactly as for smooth OCFFTs. Every smooth  OCFFT can be restricted to a smooth  FFT, under preservation of all additional properties and reflection structures.
The restriction of our OCFFT $\scZ_{\CG,\CE}$ to $\CBord_d(M,Q)$ will be denoted by $\scZ_{\CG,\CE}^{cl}$.

\begin{proposition}
\label{th:Zcl}
Let $(M,Q)$ be a target space and $(\CG,\CE) \in \TBG(M,Q)$ be a target space brane geometry. Then, 
\begin{equation*}
\scZ^{cl}_{\CG,\CE}:\CBord_2(M,Q) \to \VBdl
\end{equation*}
is a 2-dimensional, invertible, reflection-positive, superficial smooth FFT on $(M,Q)$ with the following properties:
\begin{itemize}

\item 
Its values on bordisms $(\Sigma,\sigma)$ without boundary agree with the usual surface holonomy of the bundle gerbe $\CG$ around $\sigma:\Sigma \to M$. 

\item
Its values on bordisms  without ingoing or outgoing string boundary agree with the surface holonomy in the presence of D-branes defined in~\cite{CJM--Holonomy_on_D-branes,Waldorf--Thesis}.

\end{itemize} 
In other words, our construction extends these surface holonomies to a  fully fledged smooth FFT.
\end{proposition}

\begin{proof}
The statement about the values on bordisms is a straightforward comparison of our definition of Section \ref{sect:amplitudes} with the literature. 
\end{proof}

We define the presheaf
\begin{equation}
        \Bord_{d}(M) \coloneqq \CBord_{d}(M,\emptyset) = \OCBord_{d}(M,\emptyset).
\end{equation}
It has the same objects as $\CBord_{d}(M,Q)$, but only those morphisms $[\Sigma,f]$ with $\partial_2 \Sigma = \emptyset$.
For any target space $(M,Q)$ there are canonical inclusion morphisms
\begin{equation}
\xymatrix{\Bord_{d}(M) \ar@{^(->}[r] & \CBord_{d}(M,Q) \ar@{^(->}[r] & \OCBord_{2}(M,Q)\,.
}
\end{equation}
The first morphism is surjective on objects and faithful over every $U \in \Cart$, but in general not full, whereas the second morphism is fully faithful for every $U \in \Cart$, but in general not essentially surjective.
The involutions $op_{\OCBord_d}$ and $d_{\OCBord_d}$ restrict again to involutions of $\Bord_d(M)$.
\begin{definition}
A \emph{smooth $d$-dimensional FFT on a smooth manifold $M$} is a morphism
\begin{equation}
        \scZ \colon \Bord_d(M) \rightarrow \VBdl
\end{equation}
of presheaves of symmetric monoidal categories on $\Cart$.
We let $\FFT_d(M)$ denote the category of $d$-dimensional smooth FFTs on $M$ with values in $\VBdl$.
\end{definition}

Again, invertibility, thin homotopy invariance, superficiality, and reflection structures are defined exactly as for smooth FFTs on a target space  $(M,Q)$. Every smooth  FFT on $(M,Q)$ can be restricted to a smooth  FFT on $M$. We write $\scZ_{\CG} := \scZ^{cl}_{\CG,\emptyset}$ for this restriction. We obtain the following. 

\begin{proposition}
\label{st:closed smooth TFT from BGrb}
For any bundle gerbe $\CG$ with connection over $M$, 
\begin{equation*}
\scZ_\CG \colon \Bord_d(M) \to \VBdl
\end{equation*}
is a 2-dimensional, invertible, reflection-positive, superficial smooth FFT on $M$, whose values on closed bordisms coincide with the surface holonomy of $\CG$.
\end{proposition}

\begin{remark}
A smooth FFT on a manifold $M$ is invertible if and only if it factors through the inclusion $\LBdl \hookrightarrow \VBdl$, where $\LBdl$ is the sheaf of groupoids of line bundles. In particular, this holds true for the smooth FFT $\scZ_{\CG}$ of Proposition~\ref{st:closed smooth TFT from BGrb}.
\end{remark}

We conclude this subsection with remarks on the comparison of the smooth field theories considered here and those considered in~\cite{BTW--Gerbes_and_HQFTs}.
The field theories of Bunke-Turner-Willerton are symmetric monoidal functors from a thin-homotopy surface \emph{category} $\CT_M$ of $M$ to the category $\Vect^\times = \LBdl(*)$ of complex lines.
From the definitions in~\cite{BTW--Gerbes_and_HQFTs} we observe that there is an inclusion of symmetric monoidal categories $\iota \colon \CT_M \hookrightarrow \Bord_2^{th}(M)(*)$ which is an equivalence of categories.

\begin{proposition}
Given a gerbe $\CG$ with connection on $M$, the field theory $E^\CG \colon \CT_M \to \Vect^\times$ constructed in~\cite{BTW--Gerbes_and_HQFTs} coincides with the composition
\begin{equation}
	\CT_M \overset{\iota}{\hookrightarrow} \Bord_2^{th}(M)(*) \overset{\scZ_{\CG|*}}{\longrightarrow} \Vect^\times\,,
\end{equation}
where $\scZ_{\CG|*} \colon \Bord_2^{th}(M)(*) \longrightarrow \LBdl(*)$ is the component over $* \in \Cart$ of our smooth FFT $\scZ_\CG$.
\end{proposition}

\begin{proof}
Indeed, the value of both functors on simple objects $(\bbS^1,f) \in \Bord_2^{th}(M)(*)$ is the fibre $\CT \CG_{|f}$ of the transgression line bundle of $\CG$ at the loop $f \colon \bbS^1 \to M$.
The value on morphisms is essentially the surface holonomy of the gerbe.
\end{proof}

Apart from the above relation between both constructions, the categories of field theories used in~\cite{BTW--Gerbes_and_HQFTs} and here are very different.
However, let $\mathrm{TIFT}(M)$ be the set of isomorphism classes of thin-invariant smooth two-dimensional field theories as defined in~\cite[Prop.~3.3]{BTW--Gerbes_and_HQFTs}.
We expect that there is an isomorphism of groups
\begin{equation}
	\mathrm{TIFT}(M) \cong \pi_0 \big( \FFT_2^{sf}(M)^\times \big)
\end{equation}
between $\mathrm{TIFT}(M)$ and the group of isomorphism classes of two-dimensional superficial, invertible smooth FFTs.
Indeed, Bunke-Turner-Willerton show that $\mathrm{TIFT}(M)$ is in bijection with isomorphism classes of gerbes with connection on $M$~\cite[Thm.~6.3]{BTW--Gerbes_and_HQFTs}, and we conjecture that our construction induces a bijection between $\pi_0(\FFT_2^{sf}(M)^\times)$ and isomorphism classes of gerbes with connection on $M$ as well (see also the end of Section~\ref{sec:intro}).

Let us comment on why we find that our definition of a smooth field theory is more closely aligned with direct physical intuition (see Section~\ref{sec:intro}).
In~\cite{BTW--Gerbes_and_HQFTs} a rather strict notion of smoothness is used for field theories, which requires the existence of a certain closed 3-form $H$ on $M$ such that for any 3-manifold $V$ with smooth map $v \colon V \to M$ the field theory $E \colon \CT_M \to \Vect^\times$ satisfies
\begin{equation}
	E(\partial V, v_{|\partial V}) = \exp \bigg( \int_V v^*H \bigg)\,.
\end{equation}
For any gerbe $\CG$ with connection on $M$, the restriction $\scZ_{\CG|*} \circ \iota$ of our smooth FFT satisfies this property, where $H = \curv(\CG)$ is the curvature 3-form of $\CG$.
(This follows from a well-known property of the surface amplitude.)

This smoothness condition is powerful, but it does not encode the intuition that the values of the field theory functor should depend smoothly on the input data.
Further, the smoothness property in~\cite{BTW--Gerbes_and_HQFTs} implies additional properties such as invariance under thin 3-bordism (see~\cite[Thm.~3.5]{BTW--Gerbes_and_HQFTs}) which a priori are not related to the smoothness of a field theory.
In our formalism, any additional properties of field theories are treated separately from smoothness and are implemented in the explicit notion of superficiality.
For instance, this implies the invariance under thin three-bordism; we sketch this for the torus:
given a solid torus $V$ with a smooth rank-two map $v \colon V \to M$, we can interpret these data as a thin homotopy between the bordism $[\partial V, v_{|\partial v}] \colon \emptyset \to \emptyset$ and a bordism $[\partial V, \sigma] \colon \emptyset \to \emptyset$, where $\sigma$ is a rank-one map.
If $\scZ$ is a superficial smooth FFT, it now follows that $\scZ [\partial V, v_{|\partial v}] = \scZ [\partial V, \sigma]$ as complex numbers (since $\scZ$ is thin-homotopy invariant).
By the superficiality of $\scZ$ (see Definition~\ref{def:thinhomotopyinvariantocfft} and the preceding paragraph) it now follows that $\scZ [\partial V, \sigma] = 1$.

The presheaf formalism we are using here encodes directly the smooth behaviour of the field theory on both objects and morphisms, and it fits with the modern treatment of functorial field theories~\cite{ST:SuSy_FTs_and_generalised_coho}.
Further, our formalism easily extends to field theories valued in targets other than $\VBdl$, whereas the implementation of smoothness in~\cite{BTW--Gerbes_and_HQFTs} makes specific use of the target category (in particular the fact that it is linear).

Finally, we work with the full category of smooth FFTs rather than with their isomorphism classes, and our construction depends functorially on the input gerbe.
We further believe that this is why in~\cite{BTW--Gerbes_and_HQFTs} there was no need to consider reflection structures; we refer to an upcoming paper for details on this.

\subsection{Flat gerbes and homotopy  field theories}

\label{sec:flatgerbes}

Recall the definition of a  homotopy invariant smooth OCFFT from Definition \ref{def:thinhomotopyinvariantocfft}. 

\begin{theorem}
\label{th:flat}
Let $(\CG, \CE) \in \TBG(M,\scQ)$ be a target space brane geometry on a target space $(M,\scQ)$.
The following are equivalent:
\begin{enumerate}[(1)]
\item The smooth OCFFT $\scZ_{\CG,\CE}$ is homotopy invariant.

\item The connection on the bundle gerbe $\CG$ is flat; that is, $\curv(\CG) = 0$.
\end{enumerate}
\end{theorem}

\begin{proof}
If $\CG$ is flat, Proposition~\ref{st:amplitude and homotopies} implies that $\scZ_{\CG,\CE}[\Sigma,\sigma]$ is invariant under homotopies of $\sigma$ that are constant homotopies on $U \times \partial_2 \Sigma$ and the collar neighbourhoods of $\partial_a \Sigma$ for $a=0,1$.

In the other direction, let $\Sigma = \bbD^3$ be the three-dimensional closed disc, and let $h \colon \bbD^3 \to M$ be any smooth map.
We may regard $h$ as a smooth homotopy $t \mapsto h_t$ of maps $h_t \colon \bbS^2 \to M$ for $t \in [0,1]$, with $h_0$ a constant map.
Proposition~\ref{st:amplitude and homotopies} and the homotopy invariance of $\scZ_{\CG,\CE}$ imply that
\begin{equation}
        \exp \Big( \int_{\bbD^3} h^*\curv(\CG) \Big) = \scZ_{\CG,\CE}[\bbD^3, h_1]^{-1}\, \scZ_{\CG,\CE}[\bbD^3,h_0]  = 1\,.
\end{equation}
Since $h$ was arbitrary, this implies that $\curv(\CG) = 0$.
\end{proof}

Thus, field theories arising from flat gerbes are homotopy invariant FFTs, similar to those  studied by Turaev~\cite{Turaev:HQFT}, and for topological (rather than smooth) models of gerbes, in particular, in~\cite{MW:PT_of_flat_gerbes_and_HQFT}.
Our smooth OCFFTs refine these results in that they are smooth, work in the non-flat case, and include D-branes.

\subsection{Open-closed TQFTs and smooth OCFFTs on the point}
\label{sect:Reduction to point}

In~\cite{ST:What_is_an_elliptic_object} the paradigm has been introduced that a functorial field theory with a target space $M$ should be viewed as a \emph{classical field theory on $M$}, whereas a field theory over the point $(M=\pt)$ should be understood as a \emph{quantum field theory}.

In this section we explain how our definition of a smooth OCFFT reduces to the common definition of an open-closed topological quantum field theory (OCTQFT). To this end, we consider the target space $(M,Q)$ where $M$ is a single point and where the collection $\scQ = \{Q_i\}_{i \in I}$ is a family of copies of the point indexed by $I$. We write $I$ as a shortcut for this target space. We consider the symmetric monoidal category
\begin{equation}
        \OCBord_d^{I} \coloneqq \OCBord_d(I)(\pt)\,,
\end{equation}
where $* \in \Cart$ is the point, seen as a zero-dimensional cartesian space.
The symmetric monoidal category $\OCBord_d^{I}$ inherits the two involutions $d_{\OCBord}$ and $op_{\OCBord}$. In analogy with Definitions~\ref{def:smooth field theories} and \ref{def:reflection structure} we set up the following definition.

\begin{definition}
A \emph{$d$-dimensional OCTQFT} is a symmetric monoidal functor
\begin{equation}
        \scZ \colon \OCBord_d^I \rightarrow \Vect.
\end{equation}
\end{definition}

An OCTQFT is called \emph{invertible} if the vector spaces assigned to closed objects $Y\cong S^1$ 
are 1-dimensional and all bordisms with empty brane boundary are sent to isomorphisms. Reflection structures for OCTQFTs and their positivity are defined analogously to the ones for smooth OCFFTs.
Morphisms of (reflection-positive) OCTQFTs are symmetric monoidal transformations (intertwining the reflection structures).
The category  of  $d$-dimensional OCTQFTs with D-brane labels $I$ is denoted by $\mathrm{OCTQFT}_d^I$. We use the prefix \quot{$\rmRP\text{-}$} for reflection-positivity, and the notation $(..)^{\times}$ for invertibility.

The relation between smooth OCFFTs over $I$ and OCTQFTs with D-brane labels $I$ is the evaluation of a morphism of presheaves on $\pt\in\Cart$, which is a functor
\begin{equation}
\label{eq:evaluationfunctor}
        \ev_* \colon \OCFFT_d(I) \longrightarrow \OCTQFT_d^I\,,
        \qquad
        \scZ \mapsto \scZ(\pt).
\end{equation}
We remark that every smooth OCFFT over $I$ is automatically thin homotopy invariant, superficial, and homotopy invariant, so that these properties need not to be discussed here. We describe now the surprising result hat the functor \eqref{eq:evaluationfunctor} is an equivalences of categories.
This is remarkable, since it may seem that smooth FFTs on the point have vastly more structure than TQFTs, as TQFTs do not encode any smooth families of bordisms.
However, functoriality in the test space, combined with our careful construction of the smooth bordism categories $\OCBord_d$, leads to the this result:

\begin{theorem}
\label{st:ev_* is equivalence}
The functor $\ev_*$ is an equivalence of symmetric monoidal categories,
\begin{equation}
\OCFFT_d(I) \cong \OCTQFT_d^{I},
\end{equation}
between the category of  $d$-dimensional,  smooth OCFFTs on $I$ and the category of $d$-dimensional OCTQFTs with D-brane labels $I$. The functor $\ev_{*}$ and this equivalence result extends to the full subcategories of reflection-positive and invertible field theories.
\end{theorem}

We remark that Theorem \ref{st:ev_* is equivalence} holds for an empty index set $I=\emptyset$, and hence in particular for smooth FFTs and TQFTs. 
For the proof of Theorem \ref{st:ev_* is equivalence} we require the following lemma.
Let $U \in \Cart$ be arbitrary, and let $c \colon U \to *$ be the collapse map in $\Cart$.
Since $\OCBord_d(I)$ is a presheaf of symmetric monoidal categories on $\Cart$, this induces a symmetric monoidal functor
\begin{equation}
        c^* \colon \OCBord_d^{I} \to \OCBord_d(I)(U)\,.
\end{equation}

\begin{lemma}
\label{st:c* is bijective}
For any $U \in \Cart$, the functor $c^*$ is bijective on objects and morphisms.
\end{lemma}

\begin{proof}
The bijectivity on objects follows directly from Definition~\ref{def: family of objects} and the fact that $M=\pt$.
Let $[\Sigma] \colon Y_0 \to Y_1$ be a morphism in $\OCBord_d(I)(U)$.
It readily follows from Definition~\ref{def: family of bordisms} that any representative $\Sigma$  arises as the pullback of a unique representative $\Sigma_*$ of a morphism $[\Sigma_*] \colon Y_{0,*} \to Y_{1,*}$ in $\OCBord_2^I$.
It remains to show that $\Sigma$ and $\Sigma'$ are two representatives of the same bordism in $\OCBord_2(I)(U)$ if and only if $\Sigma_*$ and $\Sigma'_*$ represent the same bordism in $\OCBord_2^I$.
To see this, let $\Psi \colon U \times \Sigma \to U \times \Sigma'$ be a fibre-wise diffeomorphism that establishes the equivalence $\Sigma \sim \Sigma'$ in $\OCBord_2(I)(U)$ (compare~\eqref{eq:equivalence in OCBord_1}).
Then, restricting $\Psi$ to the fibre over any point $x \in U$ establishes the equivalence $\Sigma_* \sim \Sigma'_*$ in $\OCBord_2^I$.
Conversely, given any diffeomorphism $\Psi_* \colon \Sigma_* \to \Sigma'_*$ that establishes $\Sigma_* \sim \Sigma'_*$, the fibre-wise diffeomorphism $1_U \times \Psi_* \colon U \times \Sigma \to U \times \Sigma'$ establishes that $\Sigma \sim \Sigma'$.
\end{proof}

\begin{proof}[Proof of Theorem~\ref{st:ev_* is equivalence}]
We consider the evaluation functor
\begin{equation}
        \ev_* \colon \OCFFT_d(I) \rightarrow \OCTQFT_d^{I}\,,
        \qquad
        \scZ \mapsto \scZ(\pt)
\end{equation}
that evaluates smooth field theories (and their morphisms) on the cartesian space $* \in \Cart$.
The functor $\ev_*$ is surjective on objects:
let $Z \in \OCTQFT_d^{I}$ be arbitrary.
We construct a smooth FFT $\scZ \in \OCFFT_d(I)$ with $\ev_*(\scZ) = Z$.
First, we define $\scZ$ on objects.
Let $U \in \Cart$ be arbitrary.
By Lemma~\ref{st:c* is bijective}, any object \smash{$Y \in \OCBord_d(I)(U)$} is actually constant, i.e.~it arises as a pullback of a unique object \smash{$Y_* \in \OCBord_d^I$} along the collapse map $U \to *$.
We set
\begin{equation}
        \scZ(U)(Y) \coloneqq Z(Y_*)\,.
\end{equation}
Also by Lemma~\ref{st:c* is bijective}, every morphism $[\Sigma]$ in $\OCBord_d(I)(U)$ is the pullback of a unique morphism $[\Sigma_*] \in \OCBord_d^I$, and we set
\begin{equation}
        \scZ(U)[\Sigma] \coloneqq Z[\Sigma_*] \colon \scZ(U)(Y_{0,*}) \to \scZ(U)(Y_{1,*})\,.
\end{equation}
The assignment $\scZ$ is functorial and symmetric monoidal since $Z$ is so, and by construction it satisfies $\ev_* \scZ = Z$.
Hence, $\ev_*$ is surjective on objects.

We now show that $\ev_*$ is fully faithful.
To that end, let $\scZ, \scZ' \colon \OCBord_d(I) \rightarrow \VBdl$ be OCFFTs, and let $\zeta \colon \scZ \to \scZ'$ be a natural transformation of presheaves of categories.
The naturality of $\zeta$ in $U \in \Cart$ implies that we have commutative diagrams
\begin{equation}
\label{eq:mps of FFTs on M=*}
\xymatrix{\scZ(U)(c^*Y_*) \ar[r]^{\scZ_c}\ar[d]_{\zeta(U)_{c^*Y_*}} & c^* \big( \scZ(\pt)(Y_*) \big) \ar[d]^{\zeta(\pt)_{Y_*}}
        \\
        \scZ'(U)(c^*Y_*) \ar[r]_{\scZ'_c} & c^* \big(\scZ'(\pt)(Y_*) \big)
}
\end{equation}
where the horizontal vector bundle isomorphisms over $U$ are part of the structure that makes $\scZ$ and $\scZ'$ morphisms of presheaves of categories.
Combining diagram~\eqref{eq:mps of FFTs on M=*} with Lemma~\ref{st:c* is bijective} now shows that any morphism $\zeta$ is entirely determined by its value on the terminal object $* \in \Cart$.
The consideration of reflection structures and invertibility is straightforward.
\end{proof}

Thus, in our formalism, smooth  OCFFTs on the one-point target space $I$ are equivalent to ordinary  OCTQFTs with D-brane labels $I$.

\subsection{Classification of 2-dimensional open-closed TQFTs}

Two-dimensional open-closed OCTQFTs have been investigated in~\cite{MS--2DTFTs,LP--Open-closed_TQFTs,Lazaroiu:OCFFT_in_2D}, for example. 
Our definition of the category $\OCTQFT_2^I$ coincides precisely with the one given in~\cite{LP--Open-closed_TQFTs}.
There, and also in \cite{MS--2DTFTs,Lazaroiu:OCFFT_in_2D}, it has been shown that 2-dimensional OCTQFTs with D-brane labels $I$ are in equivalence to so-called $I$-coloured knowledgeable Frobenius algebras, as defined in~\cite[Section~2]{LP--Open-closed_TQFTs}, see Theorem \ref{th:classkfrob} below. We first recall the definition here for convenience.

\begin{definition}
Let $I$ be a set.
An \emph{$I$-coloured knowledgeable Frobenius algebra} is a septuple $(\scL,\scR, \chi, \epsilon, \theta, \iota, \iota^*)$ of the following data:
\begin{enumerate}[(1)]
        \item $\scL$ is a commutative Frobenius algebra over $\FC$, with trace denoted by $\vartheta$.
        
        \item $\scR = \{\scR_{ij}\}_{i,j \in I}$ is a family of finite-dimensional complex vector spaces.
        
        \item $\chi = \{\chi_{ijk}\}_{i,j,k \in I}$ is a collection $\chi_{ijk} \colon \scR_{jk} \otimes \scR_{ij} \to \scR_{ik}$ of linear maps, which satisfies an associativity condition for quadruples of elements in $I$.
        
        \item $\epsilon = \{\epsilon_i\}_{i \in I}$ is a family of elements $\epsilon_i \in \scR_{ii}$ that is neutral with respect to $\chi_{iii}$.
        In particular, $(\scR_{ii}, \chi_{iii}, \epsilon_i)$ is an algebra for every $i \in I$.
        
        \item $\theta = \{\theta_i\}_{i \in I}$ is a collection of linear maps $\theta_i \colon \scR_{ii} \to \FC$.
        
        \item $\iota = \{\iota\}_{i \in I}$ is a family of linear maps $\iota_i \colon \scL \to \scR_{ii}$, central in the sense that
        \begin{equation}
                \chi_{iij} \big( v \otimes \iota_i(\ell) \big) = \chi_{ijj} \big( \iota_j(\ell) \otimes v \big)
        \end{equation}
        for all $\ell \in \scL$ and $v \in \scR_{ij}$.
        
        \item $\iota^* = \{\iota^*_i\}_{i \in I}$ is a family of linear maps $\iota^*_i \colon \scR_{ii} \to \scL$, adjoint to $\iota$ in the sense that
        \begin{equation}
                \vartheta \big( \ell \cdot \iota^*_i(c) \big) = \theta_i \big( \chi_{iii}(\iota_i(\ell) \otimes v) \big)
        \end{equation}
        for all $v \in \scR_{ii}$ and $\ell \in \scL$.
\end{enumerate}
For $i,j \in I$, let $\sigma_{ij}$ denote the pairing
\begin{equation}
\xymatrix{
        \sigma_{ij} \colon \scR_{ji} \otimes \scR_{ij} \ar[r]^(.675){\chi_{iji}} & \scR_{ii} \ar[r]^{\theta_i} & \FC\,.
}
\end{equation}
The above data are subject to the following conditions:
\begin{itemize}
        \item $\sigma_{ij}$ is non-degenerate; that is, it induces an isomorphism $\Phi_{ij} \colon \scR_{ij} \to \scR_{ji}^\vee$ for every $i,j \in I$.
        
        \item $\sigma = \{\sigma_{ij}\}_{i,j \in I}$ is symmetric, meaning that $\sigma_{ij}(a \otimes b) = \sigma_{ji}(b \otimes a)$ for all $a \in \scR_{ji}$ and $b \in \scR_{ij}$.
        
        \item If $(v_1, \ldots, v_n)$ is a basis of $\scR_{ij}$ with dual basis $(v^1, \ldots, v^n)$ of $\scR_{ji}$ under the isomorphism $\Phi_{ij}$, then for any $v \in \scR_{ii}$ we have that
        \begin{equation}
                (\iota_j \circ \iota^*_i)(v) = \sum_{k = 1}^n \chi_{jij} \big( \chi_{iij} (v_k \otimes v) \otimes v^k \big)\,.
        \end{equation}
\end{itemize}
A \emph{morphism}
\begin{equation}
        (\varphi, \xi) \colon (\scL,\scR, \chi, \epsilon, \theta, \iota, \iota^*) \to (\scL',\scR', \chi', \epsilon', \theta', \iota', \iota'{}^*)
\end{equation}
of  $I$-coloured knowledgeable Frobenius algebras consists of a linear map $\varphi \colon \scL \to \scL'$ and a family $\xi = \{\xi_{ij}\}_{i,j \in I}$ of linear maps $\xi_{ij} \colon \scR_{ij} \to \scR_{ij}'$ that respect the products, units and traces, and satisfy
\begin{equation}
        \iota'_i \circ \varphi = \xi_{ii} \circ \iota_i\,,
        \quad
        \iota'_i{}^* \circ \xi_{ii} = \varphi \circ \iota_i^*
\end{equation}
for all $i,j \in I$.
\end{definition}

This defines a category of  $I$-coloured knowledgeable Frobenius algebras, which we denote by
$\KFrob^I$. The classification result obtained in \cite{LP--Open-closed_TQFTs,MS--2DTFTs,Lazaroiu:OCFFT_in_2D} is the following.

\begin{theorem}
\label{th:classkfrob}
There is an equivalence of categories
\begin{equation*}
\scF: \mathrm{OCTQFT}_2^I \to \KFrob^I\text{.}
\end{equation*}
\end{theorem}

We briefly recall the definition of this functor. If $\scZ$ is a 2-dimensional OCTQFT with brane labels $I$, then the Frobenius algebra $\scL$ of $\scF(\scZ)$ has the underlying vector space $\scZ(S^1)$, the product is given by the value of $\scZ$ on a closed pair of pants, and the trace $\vartheta$ is obtained from the cap bordism $S^1 \to \emptyset$. The vector space $\scR_{ij}$ is $\scZ([0,1])$, and the linear maps $\chi_{ijk}$ are obtained by evaluating $\scZ$ on an open pair of pants. The elements $\epsilon_i$ and the linear maps $\theta_i$ are obtained from the open cap bordisms $\emptyset \to [0,1]$ and $[0,1] \to \emptyset$, respectively. The linear maps $\iota_i$ and $\iota_i^{*}$ are obtained from the unzip bordism $S^1 \to [0,1]$ and the zip bordism $[0,1] \to S^1$, respectively.     

The inclusion of reflection structures is this picture is straightforward. First we recall the  definition of reflection structures on knowledgeable Frobenius algebras according to \cite[Def.\ 3.3.3, Def.\ 3.3.7]{BW:Transgression_and_regression_of_D-branes}.

\begin{definition}
\label{def:RPKFrob}
A \emph{reflection structure} on an $I$-coloured knowledgeable Frobenius algebra $(\scL,\scR, \chi, \epsilon, \theta, \iota, \iota^*)$ is a pair $(\tilde{\lambda}, \tilde{\alpha})$ of an involutive algebra isomorphism $\tilde{\lambda} \colon \scL \to \overline{\scL}$ and a family $\tilde{\alpha} = \{\tilde{\alpha}_{ij}\}_{i,j \in I}$ of involutive (meaning $\tilde{\alpha}_{ji} \circ \tilde{\alpha}_{ij} = 1$ for all $i,j \in I$), anti-multiplicative isomorphisms $\tilde{\alpha}_{ij} \colon \scR_{ij} \to \overline{\scR_{ji}}$.
These have to satisfy the conditions
\begin{equation}
        \vartheta \big( \tilde{\lambda}(\ell) \big) = \overline{\vartheta(\ell)},\,
        \quad
        \tilde{\alpha}_{ii}(\epsilon_i) = \epsilon_i\,,
        \quad
        \theta_i \big( \tilde{\alpha}_{ii}(v) \big) = \overline{\theta_i(v)}\,,
        \qandq
        \tilde{\alpha}_{ii} \circ \iota_i = \iota_i \circ \tilde{\lambda}
\end{equation}
for all $i \in I$, $v \in \scR_{ii}$ and $\ell \in \scL$.
The reflection structure $(\tilde{\lambda}, \tilde{\alpha})$ is called \emph{positive} if the pairings 
\begin{equation}
        (v,w) \mapsto \sigma_{ij} \big( \tilde{\alpha}_{ji}^{-1}(v) \otimes w \big)
        \qandq
        (\ell, \ell') \mapsto \vartheta \big( \tilde{\lambda}^{-1}(\ell) \cdot \ell')
\end{equation}
are positive definite for all $i,j \in I$.
An $I$-coloured knowledgeable Frobenius algebra with a positive reflection structure is called \emph{reflection-positive}.
 A morphism
$(\varphi, \xi)$ of $I$-coloured knowledgable Frobenius algebras respects reflection structures $(\tilde{\lambda}, \tilde{\alpha})$ and $(\tilde{\lambda}', \tilde{\alpha}')$  
if
\begin{equation}
         \tilde{\lambda}' \circ \varphi = \varphi \circ \tilde{\lambda}
        \qandq
        \tilde{\alpha}'_{ij} \circ \xi_{ij} = \xi_{ji} \circ \tilde{\alpha}_{ij}
\end{equation}
for all $i,j \in I$.
\end{definition}

Definition~\ref{def:RPKFrob} gives rise to a category of reflection-positive $I$-coloured knowledgeable Frobenius algebras, which we denote by
$\RPKFrob^{I}$. The following extension of Theorem \ref{th:classkfrob} is straightforward to deduce. 

\begin{proposition}
\label{prop:classfunctorRP}
The functor $\scF$ extends to an equivalence of categories
\begin{equation*}
\RP\text{-}\mathrm{OCTQFT}_2^I \cong \RP\text{-}\KFrob^I\text{.}
\end{equation*}
\end{proposition}

We consider the full subcategory
\begin{equation}
        \RPKFrob_\FC^{I} \subset \RPKFrob^I
\end{equation}
on those objects with $\scL \cong \FC$ as Frobenius algebras, and recall the following result  \cite[Proposition~3.3.8]{BW:Transgression_and_regression_of_D-branes}.

\begin{proposition}
\label{prop:wtfrob}
There is an equivalence of categories
\begin{equation*}
\wtfrob: \TBG(I) \to \RP\text{-}\KFrob^{I}_{\FC}\text{.}
\end{equation*}
\end{proposition}

We shall briefly recall the definition of the functor $\wtfrob$. It assigns to an object $(\CG, \CE) \in \TBG(I)$ an object
\begin{equation}
\label{eq:frob(CG,CE)}
        (\scL, \scR, \chi, \epsilon, \theta, \iota, \iota^*, \tilde{\lambda}, \tilde{\alpha})
        \coloneqq \wtfrob (\CG, \CE)
        \ \in \RPKFrob^{I}_\FC\,,
\end{equation}
whose complex line $\scL$ is the fibre of the transgression line bundle $\scL$ over the loop space $L(*) = *$.
For $i,j \in I$, the vector space $\scR_{ij}$ is the fibre of the transgression bundle over the path space $P_{ij}(\pt) \cong \pt$ as constructed in Section~\ref{sect:bundles on path spaces}.
The product on the $\FC$-algebra $\scL$ is the fusion product on the transgression line bundle $\scL$.
In general, it can be computed from compatible trivialisations of pullbacks of $\CG$ along the loops, as laid out in~\cite[Section 4.1]{BW:Transgression_and_regression_of_D-branes}, but see also~\cite{Waldorf--Trangression_II} for more details.
In the present special case of $M = *$, with respect to a trivialisation $\hat{\CT} \colon \CG \to \CI_0$ on $M = *$, it simply reads as
\begin{equation}
        \big[ [\gamma^*\hat{\CT}], z \big] \otimes \big[ [\gamma'{}^*\hat{\CT}], z' \big]
        \mapsto \big[ [(\gamma {*} \gamma')^*\hat{\CT}], z z' \big]\,,
\end{equation}
where $\gamma' {*} \gamma$ denotes concatenation.
The trace on $\scL$ is given  as
\begin{equation}
        \vartheta \big[ [\gamma^*\hat{\CT}], z \big] = z\,.
\end{equation}
Next, we consider the morphisms $\chi$, which stem from the concatenation of paths $\gamma_{ij}$ and $\gamma_{jk}$ in $M$ for D-brane indices $i,j,k \in I$.
Let $\gamma_{ij} \colon [0,1] \to \pt$ denote the unique smooth path in $M = \pt$ from $Q_i$ to $Q_j$.
Then,
\begin{equation}
        \chi_{ijk} \colon \big[ \gamma_{ij}^*\hat{\CT}, \psi_{ij} \big] \otimes \big[ \gamma_{jk}^*\hat{\CT},\, \psi_{jk} \big]
        \mapsto \big[ \gamma_{ik}^*\hat{\CT},\, \psi_{jk} \circ \psi_{ij} \big]\,,
\end{equation}
in the notation of Section~\ref{sect:bundles on path spaces}.
We can equivalently formulate this as an evaluation:
if $\psi_{ki} \in \scR_{ki} \cong \scR_{ik}^\vee$, then
\begin{equation}
        \< \beta_{ik}^{-1} \big[ \gamma_{ki}^*\hat{\CT},\, \psi_{ki} \big],\ \chi_{ijk} \big( \big[ \gamma_{ij}^*\hat{\CT}, \psi_{ij} \big] \otimes \big[ \gamma_{jk}^*\hat{\CT},\, \psi_{jk} \big] \big) \>
        = \tr (\psi_{ki} \circ \psi_{jk} \circ \psi_{ij})\,.
\end{equation}
The unit $\epsilon_i \in \scR_{ii}$ is the element represented by the image of the identity 2-morphism $1_{\CE_i}$.
In other words, with respect to the trivialisation $\gamma_{ii}^*\hat{\CT}$ we have that
\begin{equation}
        \epsilon_i = \big[ \gamma_{ii}^*\hat{\CT}, 1_{\CE_i} \big]\,,
\end{equation}
compare also~\cite[Section 4.6]{BW:Transgression_and_regression_of_D-branes}.
The trace $\vartheta = h_\scL(\One_\scL,-\,)$ is defined using the hermitean metric on the transgression line bundle, where $\One_\scL$ is the unit with respect to the algebra structure on the 1-dimensional vector space $\scL$, and analogously one defines the traces $\theta_i$ on $\scR_{ii}$ as
\begin{equation}
        \theta_i [\hat{\CT}, \psi_{ii}]
        \coloneqq h_{ii} \big( \epsilon_i, [\hat{\CT}, \psi_{ii}] \big)
        = \tr(\psi_{ii})\,.
\end{equation}
On representatives with respect to $\gamma^*\hat{\CT}$, the morphism $\iota_i \colon \scL \to \scR_{ii}$ reads as
\begin{equation}
        \iota_i \colon z \mapsto z\, \epsilon_i
        = z\, 1_{\CE_i}\,.
\end{equation}
In the other direction, the morphism $\iota_i^* \colon \scR_{ii} \to \scL$ acts on representatives as taking the trace.
The last piece of data in~\cite[Def.\ 3.3.3, Def.\ 3.3.8]{BW:Transgression_and_regression_of_D-branes} are the antilinear maps $\tilde{\lambda} \colon \scL \to \overline{\scL}$ and $\tilde{\alpha}_{ij} \colon \scR_{ij} \to \overline{\scR_{ji}}$ for every $i,j \in I$, which define a positive reflection structure on the $I$-coloured knowledgeable Frobenius algebra $(\scL, \scR, \chi, \epsilon, \theta, \iota, \iota^*)$.
With respect to a trivialisation $\gamma_{ij}^*\hat{\CT}$, they act as
\begin{equation}
        \tilde{\lambda} \colon \big[ [\gamma^*\hat{\CT}], z \big]
        \mapsto \big[ [\overline{\gamma}^*\hat{\CT}], \overline{z} \big]\,,
        \qquad
        \tilde{\alpha}_{ij} \colon \big[ \gamma_{ij}^*\hat{\CT},\, \psi \big]
        \mapsto \big[ \overline{\gamma}_{ji}^*\hat{\CT},\, \psi^* \big]\,,
\end{equation}
where $\psi^*$ is the hermitean adjoint of $\psi$.
On morphisms, $\wtfrob$ acts simply as the transgression functor.

The following is the main result of this section.

\begin{theorem}
\label{eq:diagram for red to pt}
\label{st:equivs for (*,I)}
The diagram of functors
\begin{equation}
\xymatrix{\TBG(I) \ar[r]^-{\scZ} \ar[d]_{\wtfrob} & \RP\text{-}\OCFFT_{2}(I)^{\times} \ar[d]^{\ev_*}
        \\
        \RPKFrob_\FC^{I}& (\RP\text{-}\OCTQFT_{2}^{I})^{\times} \ar[l]^-{\scF}
}
\end{equation}
is commutative, and all functors are equivalences of symmetric monoidal categories. 
\end{theorem}

\begin{proof}
The evaluation functor on the right hand side of the diagram is an equivalence  by Theorem~\ref{st:ev_* is equivalence}. Propositions \ref{prop:wtfrob} and \ref{prop:classfunctorRP} show that $\scF$ and $\wtfrob$ are equivalences. The two-out-of-three property of equivalences of categories implies that $\scZ$ is an equivalence, provided we prove that the diagram is commutative, which we do in the following.

Consider an arbitrary object $(\CG, \CE) \in \TBG(I)$.
We will spell out its image
\begin{equation}
        (\scL', \scR', \chi', \epsilon', \theta', \iota', \iota^{\prime*}, \tilde{\lambda}', \tilde{\alpha}')
        \coloneqq \scF (\ev_* (\scZ_{\CG, \CE}))
\end{equation}
under the functor $\scF \circ \ev_* \circ \scZ$ and compare it to its image
\begin{equation}
        (\scL, \scR, \chi, \epsilon, \theta, \iota, \iota^*, \tilde{\lambda}, \tilde{\alpha})
        \coloneqq \wtfrob (\CG, \CE)
        \ \in \RPKFrob^{I}_\FC
\end{equation}
under $\wtfrob$, which we described above.
By construction, $\scL'$ is the vector space that $\scZ_{\CG,\CE}$ assigns to the unique loop in $M = \pt$, and by construction of $\scZ_{\CG,\CE}$, this agrees with $\scL$.
The vector space $\scR'_{ij}$, for $i,j \in I$, is the vector space that $\scZ_{\CG,\CE}$ assigns to the unique path in $M = \pt$ from $Q_i = \pt$ to $Q_j = \pt$; it is precisely the vector space $\scR_{ij}$.
That is, we have
\begin{equation}
        \scL = \scL'
        \qandq
        \scR = \scR'\,.
\end{equation}
Letting $\gamma \colon \bbS^1 \to \pt$ denote the unique loop in $M = \pt$ and $\hat{\CT}$ a trivialisation of $\CG$ over $\pt$, the fusion product on the transgression line bundle $\scL$ takes the simple form
\begin{equation}
        \lambda \colon \big[ [\gamma^*\hat{\CT}], z \big] \otimes \big[ [\gamma^*\hat{\CT}], z' \big]
        \mapsto \big[ [\gamma^*\hat{\CT}], z\, z' \big]\,.
\end{equation}
On the other hand, the product $\lambda'$ on $\scL'$ is obtained as the value of $\scZ_{\CG,\CE}$ on the closed pair of pants.
Using the pullback of $\hat{\CT}$ to compute this amplitude and by inspection of the definitions in Section~\ref{sect:amplitudes-definition}, we see that
\begin{equation}
        \lambda = \lambda'\,.
\end{equation}
The morphism $\chi'$ is the value of $\scZ_{\CG,\CE}$ on an open (i.e.~flat) pair of pants $[\Sigma, \sigma]$ with two incoming intervals and one outgoing interval.
We can compute the amplitude of this bordism decorated with the states $[\gamma_{ij}^*\hat{\CT}, \psi_{ij}]$ and $[\gamma_{jk}^*\hat{\CT}, \psi_{jk}]$ on its incoming string boundary and with the dual state $[\gamma_{ik}^*\hat{\CT},\, \beta_{ki}(\psi_{ki})]$ on its outgoing string boundary.
The contributions of the string boundaries to that amplitude are obvious from these decorations, but for the brane boundary we need parallel transports in $\CV_i \coloneqq \Delta(b^*\CE_i, b^*\hat{\CT})$, where $b$ is the restriction of $\sigma$ to a connected component of the brane boundary of $\Sigma$.
However, since this bundle is pulled back from a bundle over the point, it is trivialisable as a vector bundle with connection, and so all the parallel transports in the amplitude are trivial.
Consequently, the amplitude reads as $\tr (\psi_{ki} \circ \psi_{jk} \circ \psi_{ij})$, and we conclude that
\begin{equation}
        \chi = \chi'\,.
\end{equation}
In the open-closed TQFT defined by $(\CG, \CE)$, the element $\epsilon'_i \in \scR'_{ii} = \scR_{ii}$ is obtained as follows.
Under $\scZ_{\CG,\CE}$, the cap bordism $[\Sigma_{cap,i}, \sigma] \colon \emptyset \to ([0,1], \gamma_{ii})$ yields a linear map $\FC \to \scR_{ii}$, and $\epsilon'_i$ is the evaluation of this linear map on $1 \in \FC$.
From the formalism in Section~\ref{sect:amplitudes-definition} and the above arguments regarding the triviality of the parallel transport along the brane boundary of bordisms in the case $M = \pt$ we infer that
\begin{equation}
        \< \beta_{ii}(\psi), \scZ_{\CG,\CE} [\Sigma_{cap,i}, \sigma] (z) \>
        = \CA^{\CG,\CE} \big(\Sigma_{cap,i},\, \sigma^*\hat{\CT}, \{z,\beta_{ii}(\psi)\} \big)
        = z\, \tr(\psi)
\end{equation}
for every incoming state $z \in \FC$ and outgoing dual state $\psi \in \scR_{ii} \cong \scR_{ii}^\vee$.
This implies that
\begin{equation}
        \scZ_{\CG,\CE} [\Sigma_{cap,i}, \sigma] (z) = z\, 1_{\CE_i}\,,
\end{equation}
and thus
\begin{equation}
        \epsilon_i = \epsilon'_i
        \qquad \forall\, i \in I\,.
\end{equation}
The trace $\vartheta'$ on the algebra $(\scL', \lambda')$ is obtained as the evaluation of $\scZ_{\CG,\CE}$ on the cap bordism $\bbS^1 \to \emptyset$.
This acts as
$       \big[ [\gamma^*\hat{\CT}],\, z \big] \mapsto z$,
and thus agrees with the trace $\vartheta$ on $\scL$, i.e.
\begin{equation}
        \vartheta = \vartheta'\,.
\end{equation}
Therefore, the rank-one, commutative Frobenius algebras $(\scL, \lambda, \theta)$ and $(\scL', \lambda', \theta')$ are equal.
Recalling that the hermitean metric on $\scR_{ij}$ is defined on representatives as
\begin{equation}
        h_{ij}(\psi_{ij}, \psi'_{ij}) = \tr(\psi_{ij}^* \circ \psi'_{ij})\,,
\end{equation}
we see from an analogous argument using the cap bordism $[0,1] \to \emptyset$ that also the traces on the algebras $\scR_{ii}$ and $\scR'_{ii}$ agree for every D-brane label, i.e.
\begin{equation}
        \theta = \theta'\,.
\end{equation}
In the field theory, $\iota'_i$ is the morphism induced by the unzip bordism $\bbS^1 \to [0,1]$.
Similarly to the arguments used for the cap bordisms $\emptyset \to [0,1]$, we see that this morphism coincides with $\iota_i$.
That is,
\begin{equation}
        \iota = \iota'\,.
\end{equation}
In the other direction, we have to consider the zip bordism $[0,1] \to \bbS^1$, which on representatives also maps to the trace under $\scZ_{\CG,\CE}$.
This yields
\begin{equation}
        \iota^* = \iota^{\prime*}\,.
\end{equation}
Finally, we turn to the reflection structures.
On the field theory side a reflection structure $\alpha$ on an OCTQFT naturally induces antilinear morphisms $\tilde{\lambda}$ and $\tilde{\alpha}$ on the coloured Frobenius algebra that it defines, and these antilinear morphisms behave exactly as in~\cite[Definition 3.3.8]{BW:Transgression_and_regression_of_D-branes}.
In the case of the field theory $\scZ_{\CG,\CE}$ we directly read off from~\eqref{eq:reflection structure alpha} and the following paragraph that
\begin{equation}
        \tilde{\lambda} = \tilde{\lambda}'
        \qandq
        \tilde{\alpha} = \tilde{\alpha}'\,.
\end{equation}
To summarise, we have shown that
\begin{equation}
        \wtfrob (\CG, \CE)
        = \scF (\ev_* (\scZ_{\CG, \CE}))
\end{equation}
for all objects $(\CG, \CE) \in \TBG(I)$.
That is, the diagram~\eqref{eq:diagram for red to pt} strictly commutes on objects.

On morphisms in $\TBG(I)$, we only need to compare the induced linear maps on the vector spaces underlying the $I$-coloured Frobenius algebras.
However, comparing the definition of the action of $\wtfrob$ on morphisms in~\cite[Section 4.8]{BW:Transgression_and_regression_of_D-branes} with the action~\eqref{eq:Z_A,T} of $\scZ$ on morphisms we readily see that the induced linear maps agree.
\end{proof}

We point out two consequences of Theorem \ref{st:equivs for (*,I)}.
The fact that the diagonal of the diagram in Theorem \ref{st:equivs for (*,I)} is an equivalence leads to the following result.

\begin{corollary}
The  2-dimensional, invertible, reflection-positive  OCTQFTs  are precisely those that can be obtained  from target space brane geometry over a point, i.e., those that are restrictions of classical field theories to the point.  
\end{corollary}

Secondly, we emphasise that our construction of smooth OCFFTs from target space brane geometry is an equivalence, in the special case of a one-point target space.   

\begin{corollary}
The construction $\scZ$ of a smooth OCFFT from a target space brane geometry is an equivalence,
\begin{equation*}
\TBG(I) \cong \RP\text{-}\OCFFT_{2}(I)^{\times}\text{.}  
\end{equation*}
\end{corollary}

We mention this because we believe that the functor $\scZ$ is an equivalence for \emph{any} target space $(M,\scQ)$, i.e., we conjecture that it is an equivalence
\begin{equation*}
\TBG(M,\scQ) \cong \RP\text{-}\OCFFT_{2}^{\sf}(M,Q)^{\times}\text{,} 
\end{equation*}
We plan to prove this equivalence in forthcoming work.

\begin{appendix}

    \def\theequation{\thesection.\arabic{equation}}    
    \def\theproposition{\thesection.\arabic{proposition}}    
    \def\thedefinition{\thesection.\arabic{definition}}    
    \def\thelemma{\thesection.\arabic{lemma}}    
    \def\theremark{\thesection.\arabic{remark}}

\section{Monoidal categories with fixed duals}
\label{app:Duals and Hermitean Structures}
\label{app:duals in monoidal cats}

We briefly recall some basic facts about categorical duals, mostly based on~\cite{EGNO}.

\begin{definition}
[{{\cite[Definition~2.10.1]{EGNO}}}]
Let $\scC$ be a monoidal category with monoidal structure $\otimes$ and unit object $u \in \scC$.
Let $x \in \scC$ be an object.
An object $x^\vee \in \scC$ is called a \emph{(left) dual of $x^\vee$} if there exist morphisms $\ev_x \colon x^\vee \otimes x \to u$ and $\coev_x \colon u \to x \otimes x^\vee$ such that the diagrams (where we are omitting the structural isomorphisms)
\begin{equation}
\label{eq:triangle identities for duals}
\xymatrix@C=2cm{        x \ar[r]^-{\coev_x \otimes 1_x} \ar[dr]_{1_x} & x \otimes x^\vee \otimes x \ar[d]^{1_x \otimes \ev_x} & x^\vee \ar[r]^-{1_{x^\vee} \otimes \coev_x} \ar[dr]_{1_{x^\vee}} & x^\vee \otimes x \otimes x^\vee \ar[d]^{\ev_x \otimes 1_{x^\vee}}
        \\
        & x & & x^\vee
}
\end{equation}
commute.
\end{definition}

The morphisms $\ev_x$ and $\coev_x$ are called \emph{evaluation} and \emph{coevaluation}, respectively.
We refer to the triple $(x^\vee, \ev_x, \coev_x)$ as \emph{duality data for $x$}.
Duality data is unique in the following sense.

\begin{proposition}
[{\cite[Proposition~2.10.5]{EGNO}}]
\label{st:duals are essentially unique}
\label{st:Duals are ess unique}
If an object $x \in \scC$ admits a  dual, then any two choices of  duality data $(x^{\vee},\ev_x,\coev_x)$ and $(x',\ev'_x,\coev'_x)$ for $x$ are related by a unique isomorphism $f:x^{\vee} \to x'$  which is compatible with the evaluation and coevaluation morphisms in the sense that
\begin{equation*}
\ev_x' \circ (f \otimes 1_x)=\ev_x
\quad\text{ and }\quad
(1_x \otimes f) \circ \coev_x = \coev'_x.
\end{equation*}
\end{proposition}

A monoidal category $\scC$ is said to have \emph{fixed duals}, if every object is dualisable and duality data is chosen and fixed for all objects of $\scC$.   
In a monoidal category with fixed duals we have  the following adjunction.

\begin{proposition}
[{\cite[Proposition~2.10.8, Remark~2.10.9]{EGNO}}]
\label{prop:adjunction}
Let $x \in \scC$ be an object with  duality data $(x^\vee, \ev_x, \coev_x)$.
Then, there is a canonical adjunction
\begin{equation}
\xymatrix{
        (-) \otimes x: \scC \ar[r]<0.1cm> & \scC :(-) \otimes x^\vee \ar[l]<0.1cm>
}
\end{equation}
The adjunction is established by the  bijection
\begin{equation}
\label{eq:duality adjunction}
        \tau^x_{y,z} \colon \scC(y \otimes x, z) \arisom \scC(y, z \otimes x^\vee)\,,
        \quad
        \tau^x_{y,z}(f) = (f \otimes 1_{x^\vee}) \circ (1_y \otimes \coev_x)\,,
\end{equation}
for objects $y,z \in \scC$.
\end{proposition}

We remark that the inverse of $\tau^{x}_{y,z}$ is given by
\begin{equation}
        \tau_{y,z}^{x\, -1}(g) = (1_z \otimes \ev_x) \circ (g \otimes 1_x)\,.
\end{equation}
In a monoidal category with fixed duals, duality extends automatically to the morphisms.
For a morphism $f:x \to y$  one defines the \emph{dual morphism $f^\vee :y^{\vee} \to x^{\vee}$} as the composition
\begin{equation}
\xymatrix@C=1.5cm{y^\vee \ar[r]^-{1 \otimes \coev_x} & y^\vee \otimes x \otimes x^\vee \ar[r]^-{1 \otimes f \otimes 1}
        & y^\vee \otimes y \otimes x^\vee \ar[r]^-{\ev_y \otimes 1} & x^\vee\,.
}
\end{equation}
One can show that for $g :y \to z$ one has $(g \circ f)^\vee = f^\vee \circ g^\vee$. We note the following simple fact, of which we could not find an explicit reference.

\begin{lemma}
\label{st:dual behaves like transpose}
Let $f:x \to y$ be a morphism in a monoidal category with fixed duals; then, we have
\begin{equation}
        \ev_y \circ (1_y \otimes f) = \ev_x \circ (f^\vee \otimes 1_x)\,.
\end{equation}
\end{lemma}

\begin{proof}
The identity follows from the commutativity of the diagram
\begin{equation}
\xymatrix@C=6em{   y^\vee \otimes x \ar[r]^{f^\vee \otimes 1} \ar[d]^{1 \otimes \coev_x \otimes 1}
        \ar@/_7pc/[dd]_{1 \otimes 1}
        & x^\vee \otimes x \ar@/^1pc/[rd]^{\ev_x}
        & 
        \\
        y^\vee \otimes x \otimes x^\vee \otimes x \ar[r]^{1 \otimes f \otimes 1 \otimes 1} \ar[d]^{1 \otimes 1 \otimes \ev_x}
        & y^\vee \otimes y \otimes x^\vee \otimes x \ar[u]^{\ev_y \otimes 1 \otimes 1} \ar[d]_{1 \otimes 1 \otimes \ev_x}
        & u
        \\
        y^\vee \otimes x \ar[r]_{1 \otimes f}
        & y^\vee \otimes y \ar@/_1pc/[ru]_{\ev_y}
        & 
}
\end{equation}
The left-hand triangle commutes because of the first triangle identity~\eqref{eq:triangle identities for duals}.
The upper central rectangle is  the definition of $f^\vee$ and hence commutes.
The lower central rectangle and the right-hand triangle commute by monoidality; morphisms applied to different factors of a tensor product commute in any monoidal category.
\end{proof}

If $\scC$ is a monoidal category with fixed duals, then the assignments $x \mapsto x^{\vee}$ and $f \mapsto f^{\vee}$ define a functor $d_{\scC}:\scC^{\opp} \to \scC$. Since $x^{\vee\vee}$ and $x$ are both dual to $x^{\vee}$, there exists a unique isomorphism $\delta_x: x^{\vee\vee} \to x$ by Proposition~\ref{st:duals are essentially unique}.
These define a monoidal natural isomorphism $\delta_{\scC}:d_{\scC}\circ d_{\scC}^{\opp} \to 1_{\scC}$, which additionally satisfies
\begin{equation}
1_{d_{\scC}} \circ (\delta_{\scC}^\opp)^{-1} = \delta_{\scC} \circ 1_{d_{\scC}}
\end{equation}
as natural transformations $d_{\scC} \circ d_{\scC}^{\opp} \circ d_{\scC} \to d_{\scC}$.
In the terminology used in Section \ref{sect:Duals and Opposites}, the pair $(d_{\scC},\delta_{\scC})$ forms a twisted involution, and is called the \emph{duality involution} of the monoidal category $\scC$ with fixed duals.

\begin{proposition}
\label{st:duality data and monoidal functors}
Let $\scC$, $\scD$ be monoidal categories with fixed duals and duality involutions $(d_{\scC},\delta_{\scC})$ and $(d_{\scD},\delta_{\scD})$, respectively.
Suppose $F:\scC \to \scD$ is a monoidal functor, i.e.~it comes with coherent natural isomorphism $f_{x,y} \colon Fx \otimes Fy \to F(x \otimes y)$.
Then, there is a unique monoidal natural isomorphism
\begin{equation}
        \beta \colon F \circ d_\scC \arisom d_\scD \circ F^{\opp}\text{.}
\end{equation}
that is compatible with evaluations and coevaluations in the sense that
\begin{equation*}
\label{eq:fctr and duality data}
\ev_{Fx} \circ (\beta_x \otimes 1_{Fx})=F(\ev_x) \circ f_{x^\vee,x}
\quad\text{ and }\quad
(1_{Fx} \otimes \beta_x) \circ f_{x,x^\vee}^{-1} \circ F(\coev_x) = \coev_{Fx}
\end{equation*}
for all $x\in \scC$.
Moreover, it is compatible with the natural isomorphisms $\delta_{\scC}$ and $\delta_{\scD}$ in the sense that  the diagram
\begin{equation}
\label{eq:coh for duality nat trafo}
\xymatrix@C=2cm{
F \circ d_{\scC} \circ d_{\scC}^{\opp} \ar[r]^-{\beta \circ 1} \ar[d]_{1 \circ \delta_{\scC}} & d_\scD \circ F^{\opp} \circ d_\scC^{\opp} \ar[d]^{1 \circ (\beta^{\opp})^{-1}}
        \\
        F & d_\scD \circ d_{\scC}^{\opp} \circ F \ar[l]^{\delta_{\scD} \circ 1}
}
\end{equation}
is commutative.
\end{proposition}

\begin{proof}
Let $x \in \scC$ be an object with fixed left duality data $(x^\vee, \ev_x, \coev_x)$.
Let $((Fx)^\vee, \ev_{Fx}, \coev_{Fx})$ be the fixed left duality data on $Fx \in \scD$.
Since $F$ is monoidal, the triple
\begin{equation}
        \big( F(x^\vee), e_{Fx}, c_{Fx} \big) \coloneqq \big( F(x^\vee),\, F(\ev_x) \circ f_{x^\vee,x},\, f_{x,x^\vee}^{-1} \circ F(\coev_x) \big)
\end{equation}
is another set of left duality data for $Fx$.
By Proposition~\ref{st:Duals are ess unique} there exists a unique isomorphism $\beta_x \colon F(x^\vee) \to (Fx)^\vee$ compatible with the two sets of left duality data on $Fx$; the compatibility is expressed precisely by the two equations in~\eqref{eq:fctr and duality data}.
Explicitly, $\beta_x$ is given by the composition~\cite[Proposition~2.10.5]{EGNO}
\begin{equation}
\xymatrix@C=5em{\beta_x \colon F(x^\vee) \ar[r]^-{1 \otimes \coev_{Fx}} & F(x^\vee) \otimes Fx \otimes (Fx)^\vee \ar[r]^-{e_{Fx} \otimes 1} & (Fx)^\vee\,,
}
\end{equation}
We have to show that $\beta$ is natural.
To that end, suppose $\psi \colon x \to y$ is a morphism in $\scC$.
We have to prove that the diagram
\begin{equation}
\xymatrix@C=3em{    F(y^\vee) \ar[d]_{\beta_y} \ar[r]^{F(\psi^\vee)} & F(x^\vee) \ar[d]^{\beta_x}
        \\
        (Fy)^\vee \ar[r]_{(F\psi)^\vee} & (Fx)^\vee
}
\end{equation}
commutes.
Inserting the definition of the dual morphism, we expand this diagram to
\begin{equation}
\begin{small}
\xymatrix@C=4em{     F(y^\vee) \ar[r]^-{1 \otimes c_{Fx}} \ar[d]_-{\beta_y} & F(y^\vee) \otimes Fx \otimes F(x^\vee) \ar[r]^{1 \otimes F\psi \otimes 1} \ar[d]^{\beta_y \otimes 1 \otimes \beta_x} & F(y^\vee) \otimes Fy \otimes F(x^\vee) \ar[r]^-{e_{Fy} \otimes 1} \ar[d]^{\beta_y \otimes 1 \otimes \beta_x} & F(x^\vee) \ar[d]^{\beta_x}
        \\
        (Fy)^\vee \ar[r]_-{1 \otimes \coev_{Fx}} & (Fy)^\vee \otimes Fx \otimes (Fx)^\vee \ar[r]_-{1 \otimes F\psi \otimes 1} & (Fy)^\vee \otimes Fy \otimes (Fx)^\vee \ar[r]_-{\ev_{Fy} \otimes 1} & (Fx)^\vee
}
\end{small}
\end{equation}
Here, the left-hand and right-hand squares commute because of the compatibility of $\beta$ with the duality data, and the centre square commutes simply by monoidality.

To see that $\beta$ is monoidal, first observe that for any two objects $x,y \in \scC$ the triple
\begin{equation}
        \big (y^\vee \otimes x^\vee, \ev_y \circ (1 \otimes \ev_x \otimes 1), (1 \otimes \coev_y \otimes 1) \circ \coev_x \big)
\end{equation}
is left duality data for $x \otimes y$.
Thus, there exists a unique isomorphism $\delta_{x,y} \colon y^\vee \otimes x^\vee \to (x \otimes y)^\vee$ that is compatible with the duality data.
Then consider the diagram
\begin{equation}
\label{eq:monoidality of beta}
\xymatrix@C=4em{   F((x \otimes y)^\vee) \ar[r]^{F\delta_{x,y}^{-1}} \ar[d]_{\beta_{x \otimes y}} & F(y^\vee \otimes x^\vee) \ar[r]^{{f_{y^\vee,x^\vee}}^{-1}} & F(y^\vee) \otimes F(x^\vee) \ar[d]^{\beta_x \otimes \beta_y}
        \\
        (F(x \otimes y))^\vee \ar[r]_{{f_{x,y}}^\vee} & (Fx \otimes Fy)^\vee \ar[r]_{\delta_{Fx, Fy}^{-1}} & (Fy)^\vee \otimes (Fx)^\vee
}
\end{equation}
All of the objects in this diagram come with choices of left duality data for $F(x \otimes y)$, and the morphisms in the diagram are each compatible with these left duality data (this uses the coherence of the natural isomorphism $f$ and Lemma~\ref{st:dual behaves like transpose}).
In particular, the diagram yields two possible ways of going from $F((x \otimes y)^\vee)$ to $(Fy)^\vee \otimes (Fx)^\vee$, both of which are compatible with the left duality data on these two objects.
Thus, Proposition~\ref{st:Duals are ess unique} implies that the diagram~\eqref{eq:monoidality of beta} commutes, which shows that $\beta$ respects the tensor product.

The commutativity of the diagram~\eqref{eq:coh for duality nat trafo} can be shown by similar methods; since it will not be used in this article we leave this as an exercise to the reader.
\end{proof}

\end{appendix}

%\vspace{1cm}
\begin{small}

\bibliographystyle{alphaurl}
\addcontentsline{toc}{section}{References}
\bibliography{GerbyTFTBib}

\vspace{1cm}

\noindent
(Severin Bunk)
Universität Hamburg, Fachbereich Mathematik, Bereich Algebra und Zahlentheorie,
\\
Bundesstraße 55, 20146 Hamburg, Germany
\\
severin.bunk@uni-hamburg.de

\noindent
(Konrad Waldorf)
Universit\"{a}t Greifswald,
Institut f\"{u}r Mathematik und Informatik,
\\
Walther-Rathenau-Str. 47,
17487 Greifswald,
Germany
\\
konrad.waldorf@uni-greifswald.de
\end{small}

\end{document}